\documentclass{article}
\usepackage[margin=1in]{geometry}
\usepackage{booktabs} 
\usepackage[ruled]{algorithm2e} 

\SetAlFnt{\small}
\SetAlCapFnt{\small}
\SetAlCapNameFnt{\small}
\SetAlCapHSkip{0pt}
\IncMargin{-\parindent}

\usepackage{graphicx}
\usepackage{subcaption}
\usepackage{tikz}
\usepackage{natbib}

\usepackage{url}
\usepackage{amsmath}
\usepackage{amsfonts}
\usepackage{amssymb}
\usepackage{amsthm}

\usepackage{hyperref}

\usepackage{authblk}

\newtheorem{theorem}{Theorem}

\newtheorem{proposition}[theorem]{Proposition}
\newtheorem{lemma}[theorem]{Lemma}

\newtheorem{corollary}[theorem]{Corollary}

\newtheorem{example}{Example}

\newcommand{\TokFn}{g} 
\newcommand{\TokExp}{\beta} 
\newcommand{\TokExpSecond}{\gamma} 
\newcommand{\subfee}{\alpha}
\newcommand{\price}{\nu}
\newcommand{\costadditional}{\price^{H}}
\newcommand{\costmanual}{\price_0}
\newcommand{\NumConsumers}{C}
\newcommand{\NumProviders}{P}

\newcommand{\w}{w} 

\newcommand{\ic}{\price^*}

\newcommand{\marg}{\gamma}

\newcommand{\argmax}{\mathop{\mathrm{arg\,max}}}

\newcommand{\providerchoice}[1]{i_{#1}} 
\newcommand{\providerchoicemiddleman}{\providerchoice{m}}
\newcommand{\wmiddleman}{\w_{m}} 
\newcommand{\wuser}[1]{\w_{c, #1}} 
\newcommand{\action}[1]{a_{#1}}
\newcommand{\Direct}{D} 
\newcommand{\Middleman}{M} 

\newcommand{\LowerThreshold}{T_L}
\newcommand{\UpperThreshold}{T_U}

\newcommand{\LowerThresholdMarg}{T^{\text{marg}}_L}
\newcommand{\UpperThresholdMarg}{T^{\text{marg}}_U}

\newcommand{\LowerThresholdMon}{T^{\text{mon}}_L}
\newcommand{\UpperThresholdMon}{T^{\text{mon}}_U}

\title{Flattening Supply Chains: When do Technology Improvements lead to Disintermediation?\footnote{Authors in alphabetical order. Part of this work was conducted while MJ was at Microsoft Research.}}

\author[1]{S. Nageeb Ali}
\author[2]{Nicole Immorlica}
\author[3]{Meena Jagadeesan}
\author[2]{Brendan Lucier}
\affil[1]{\textit{Penn State}}
\affil[2]{\textit{Microsoft Research}}
\affil[3]{\textit{University of California, Berkeley}}

\begin{document}

\maketitle

\begin{abstract}
In the digital economy, technological innovations make it cheaper to produce high-quality content. For example, generative AI tools reduce costs for creators who develop content to be distributed online, but can also reduce production costs for the users who consume that content. These innovations can thus lead to disintermediation, since consumers may choose to use these technologies directly, bypassing intermediaries. To investigate when technological improvements lead to disintermediation, we study a game with an intermediary, suppliers of a production technology, and consumers. First, we show disintermediation occurs whenever production costs are too high or too low. We then investigate the consequences of disintermediation for welfare and content quality at equilibrium. While the intermediary is welfare-improving, the intermediary extracts all gains to social welfare and its presence can raise or lower content quality. We further analyze how disintermediation is affected by the level of competition between suppliers and the intermediary's fee structure. More broadly, our results take a step towards assessing how production technology innovations affect the survival of intermediaries and impact the digital economy.
\end{abstract}

\section{Introduction}\label{sec:intro}

The digital economy thrives on the ability to distribute content and services at scale. Consider an artist or other creator that uses the latest technology to develop high-quality content, which can be shared through online distribution platforms for a portion of advertising or subscription revenue. A key feature of these platforms is that they enable creators to distribute content to a broad audience at negligible marginal cost.

As technological innovations lead to new tools for content creation, the fixed costs of content production also continue to fall. For example, generative AI tools are making it cheaper to produce art and other digital content.\footnote{We use the term \textit{Generative AI tools} to refer to a wide range of generative models such as language models, text-to-image models, and text-to-video models.} 
At first glance, it might appear that these technology improvements would benefit creators by reducing the cost of creating high-quality content. 
However, technology improvements can also threaten \emph{disintermediation}: if the cost of content development gets sufficiently low, consumers may consider bypassing creators in favor of using the underlying technology themselves. 
Since content creators typically have no control over pricing on these platforms, the only way for them to retain their audience is to provision quality content.  Improvements in technology make this strategy cheaper for content creators, but also make self-production more enticing for consumers.

In this work, we investigate how shifts in production technology impact disintermediation and, through it, the overall welfare and the landscape of content quality enjoyed at the equilibrium of the content market.  We study the strategic choices and pressures faced by intermediaries that sit between an underlying production technology and content consumers (Figure \ref{fig:model}).  The intermediaries can produce content at a chosen level of quality
to be distributed
at scale using a fixed distribution platform. 
Our focus is the relationship between the underlying production technology (which determines costs), the quality of content produced, and users' choice to use (or not) the intermediary. 

Intermediaries provide a social good by unlocking economies of scale: multiple consumers can derive benefit from a one-time investment in content quality.\footnote{There are many other reasons intermediaries are a desirable feature of markets, particularly in the context of content production.  For example, when intermediaries distribute the same content to multiple consumers, it creates a sense of community among the audience which has arguably additional positive externalities.  We leave investigation of these forces open for future work.}  As long as there are multiple consumers, it is socially optimal for an intermediary to create content and distribute it broadly.
If the underlying technology improves and production costs decrease then, at this socially optimal solution, the intermediary will produce at a higher level of quality.
However, depending on the cost of production, equilibrium forces can push against this first-best outcome.  On one hand, if producing high-quality content is too expensive, being a creator that earns some fixed subscription fee or advertising revenue stream may simply not be profitable.  On the other hand, if production technology improves to the point of being extremely cheap and easy to use, then consumers are likely to prefer self-made solutions that bypass platform fees or the nuisance cost of advertising.  Anticipation of these issues leads to a distortion in the level of quality produced by the intermediary, who must stave off these potential market failures.

\paragraph{Model and results overview.} To focus on the competitive disintermediation pressure between the intermediary and the consumers that form their target audience, we study a model with a single intermediary that offers content to a population of consumers. We think of the intermediary as having built an audience for their content, perhaps by developing a niche or reputation, and our model abstracts away from horizontal differentiation and other inter-creator competitive forces that would shape that audience. Consumers pay a fee to access the intermediary's content, and we likewise view the fee structure as determined by the distribution channel(s) used by the intermediary and not the intermediary itself.  This is motivated by a lack of market power on the part of any single creator of content: rather than directly influencing market prices or advertising deals, the intermediary controls the quality of the content they produce.

The intermediary contracts with one of multiple possible suppliers of a production technology, which maps quality levels to production costs.  The intermediary then chooses to produce their content at a strategically-chosen level of quality.  Crucially, the production technology is also available to the consumers, who can choose whether to consume the content created by the intermediary or to bypass it and create their own content. If the consumers consume the content of the intermediary, the intermediary receives a fixed fee per consumer.

We analyze the subgame perfect equilibrium of the resulting game, which we show to be effectively unique.  We focus on how this equilibrium changes as the production technology improves, modeled as a multiplicative shift in the cost of production.  We summarize our findings as follows: 

\begin{itemize}
   \item  \textbf{Disintermediation at the extremes of production technology.}  We find that the intermediary stays in the market at a bounded range of technology levels, whereas disintermediation (i.e., all consumers choosing to bypass the intermediary) occurs when the production costs are either too high or too low (Section \ref{sec:disintermediation}; Figure \ref{fig:thresholds}).

   \item \textbf{Intermediary is welfare-improving.}  When the intermediary is present in the market, social welfare at equilibrium is higher than in a counterfactual where the intermediary does not exist (Figure \ref{fig:welfare-comparison}).  The equilibrium welfare is increasing as the production technology improves, and there is at most one ``bliss point'' where equilibrium welfare matches the first-best (Section \ref{subsec:welfare}; Figure \ref{fig:socialwelfare}).
    
  \item \textbf{Intermediary extracts all gains to social welfare.} Any increase in welfare over the counterfactual with no intermediary is captured by the intermediary itself. In other words, consumers and suppliers are indifferent to the intermediary's presence (Section \ref{subsec:consumers}; Figure \ref{fig:consumer-utility-comparison}).  We emphasize that this full extraction of the surplus improvement occurs even though the price of consuming the intermediary's content is fixed: the only strategic knob available to the intermediary is the quality of content they produce.  The intermediary's utility is ``inverse U-shaped'' as a function of the technology level: it first rises as production technology improves, but eventually levels out and begins to drop as the impact of potential disintermediation grows stronger, until eventually the intermediary's utility falls to $0$ and disintermediation occurs (Section \ref{subsec:middleman}; Figure \ref{fig:middleman}).

   \item \textbf{Intermediary can raise or lower quality, and reduces the sensitivity to technology improvements.}  We also consider the impact of the intermediary and technology costs on the quality of content produced (Section \ref{subsec:quality}; Figure \ref{fig:quality-comparison}).  Cheaper production costs lead to higher equilibrium levels of quality.  However, whenever the intermediary is being used at equilibrium, their presence dampens the sensitivity of quality to technology costs.  When production costs are high, the intermediary produces higher-quality content than what consumers would produce themselves.  As technology improves and production costs go down, quality improvements lag, ultimately reaching a point where the intermediary is producing lower-quality content than what consumers would produce. Once production costs drop low enough that disintermediation occurs, the quality of content increases sharply as the intermediary leaves the market (Figure \ref{fig:quality}). 
\end{itemize}

\paragraph{Model extensions.}
Finally, as a robustness check, we consider extensions to our base model (Section \ref{sec:extensions}).  First, we consider what changes if the production technology is controlled by a monopolist who can strategically set the relative price of production to maximize revenue, given the underlying technology cost (Section \ref{subsec:monopolist}; Figure \ref{fig:monopolist}).  We still find that the intermediary stays in the market at intermediate levels of technology. However, the range of technology levels supporting an intermediary is shifted to be lower: the intermediary can survive when the technology costs are lower, but is less likely to enter when the technology costs are high.

We also consider a variation of the model where marginal distribution is not free, but rather occurs at a reduced cost that is still proportional to content creation costs (Section \ref{subsec:marginalcosts}; Figure \ref{fig:marginalcosts}).  This can capture settings where the intermediary 
{produces content that might need to be specialized for each consumer at a small marginal cost.  For example, the intermediary might be a graphic designer with a suite of standard wedding invitation designs that can be tweaked for any particular client.} We again find that the intermediary stays in the market at intermediate levels of technology. However, the range of technology levels supporting an intermediary reduces in width.

Finally, we note the importance of our modeling assumption that the intermediary does not have full control over both the content quality and the fee structure.  We consider a variation where the 
intermediary can charge a fee that increases linearly with content quality (Section \ref{subsec:transfers}).  In this case, disintermediation can be avoided entirely, at all technology levels, as long as the marginal fee is not too high. We conclude that the nature of the fee structure in digital content distribution, and in particular that individual creators may not be able to differentiate on price, has a substantial impact on the market's sensitivity to underlying technology changes.

\subsection{Related work}\label{subsec:relatedwork}

Our work connects with research threads on the \textit{economic ramifications of generative AI tools}, \textit{content creator incentives in recommender systems}, and \textit{fragility of endogenous supply chain networks}.

\paragraph{Economic ramifications of generative AI tools.} Our work contributes to an emerging literature on the impact of generative AI tools on labor. Many of these models consider how generative AI substitutes for or is complementary to workers at a macroeconomic level~\citep{acemoglu2025simple,ide2024artificial}. Our paper views AI tools as complementing creators through cost reduction while also threatening to substitute for them via disintermediation.

In field experiments across specific domains as diverse as software development (\cite{cui2024effects,yeverechyahu2024impact}), customer service (\cite{brynjolfsson2025generative}), research (\cite{toner2024artificial}) and art (\cite{zhou2024generative}), generative AI also shows promising gains. For example, the last of these papers showed that
for artists, adoption of generative AI tools resulted in 25\% more artworks and 25\% more “favorites” per view. Our paper tries to understand the implication of  these demonstrated effects (that quality production is ``cheaper'' with the introduction of AI), especially in the artistic domain as consumers bypass experts and use AI to satisfy their own individual demand.

\paragraph{Content creator incentives in recommender systems.} Our work relates to a rich line of work on content creator incentives in recommender systems. Most of these works have focused on the \textit{impact of the recommendation algorithm} on the supply of digital content (e.g., \citep{ghoshincentivizing2011, benporatgametheoretic2018, JGS22, hronmodeling2022, yaocontentcreation2023}), pricing decisions by creators (e.g., \citep{calvano2023artificial, castellini2023recommender}), and creator participation decisions (e.g., \cite{ghoshincentivizing2011, MCBSZB20, BT23, HHLOS24}). In contrast, we study the impact of generative AI tools on whether creators can survive in the market. 

Similar to our work, a handful of recent works study how generative AI tools affect the digital content economy. For example, \citet{YLNWX24, EBFWX24} study competition between creators and generative models, capturing how the quality of the generative model depends on the quality of human-generated content. \citet{BP24, BP25} study the interplay between participation on human-based platforms and the generative model-provider's strategic retraining decisions. \citet{R24} studies  how competition between creators who use generative AI tools affects content diversity. \citet{burtch2024consequences} empirically study how LLMs affect participation in online knowledge communities. 
In contrast to these works, we focus on how generative AI tools lead to implicit competition between creators and consumers: specifically, we analyze when consumers are incentivized to bypass creators and create content on their own.\footnote{Our work focuses on a monopolist creator, and also does not account for interdependence between quality of the generative model and the quality of human-generated content.}

\paragraph{Fragility of endogenous supply chain networks.} Our work concerns the study of a specific supply chain network -- one where there is a single intermediary positioned between firms providing improving production technologies and consumers.  A long line of literature in economics, computer science and operations research considers more complex supply chain networks and how they react to shocks such as the severance of links or the bankruptcy of nodes~\citep{acemoglu2024macroeconomics,bimpikis2019supply,blume2013network,elliott2022supply}.  
Similar to our work, these papers allow relationships to form (or be severed) endogenously, but these decisions are driven by the desire to be robust to failures as opposed to changing production technology.  
Furthermore, in contrast to our work where quality selection is the only lever available to the intermediary, these papers sometimes assume price setting capability. 
The general message of these papers is that equilibrium supply chain networks can be inefficient and small shocks can cause disproportionate disruptions, suggesting an endogenous fragility.  Our work also demonstrates a discrete jump in some comparative statics, namely provisioned quality and disintermediation (i.e., network structure), from small changes.  However these outcomes are not disastrous for overall welfare.  

Some papers also consider how changes in production technology impact networks \citep{acemoglu2020endogenous} (see also cited papers on economic ramifications of generative AI tools above).  These papers tend to show that improvements in technology  reduce prices as they diffuse through a fixed production network and lead to a denser endogenous production network.

\section{Model}\label{sec:model}

\begin{figure}
\centering
\begin{tikzpicture}[scale=0.9, every node/.style={scale=0.9},
    box/.style={draw, minimum width=2cm, minimum height=1cm, align=center, fill=cyan!5},
    >=latex
]
\node[box] (s0) at (0,2) {\textbf{Manual} \\ \textbf{production}};

\node[box] (s1) at (0,0) {\textbf{Supplier 1} \\ \textit{Offers price $\price_1$}};

\node[box] (s2) at (0,-2) {\textbf{Supplier 2} \\ \textit{Offers price $\price_2$}};

\node[box] (c1) at (12,1) {\textbf{Consumer 1} \\ \textit{\textcolor{blue}{Decides whether to produce content}} \\ \textit{Consumes content $\wuser{1}$}};
\node[box] (c2) at (12,-1) {\textbf{Consumer 2} \\ \textit{\textcolor{blue}{Decides whether to produce content}} \\ \textit{Consumes content $\wuser{2}$} };

\draw[<->, blue, very thick] (s2) to[bend left=20] node[midway] {} (c1);
\draw[<->, blue, very thick] (s0) to[bend right=20] node[midway, below] {} (c2);

\node[box] (m) at (5,0) {\textbf{Intermediary} \\ \textit{Produces content $\wmiddleman$}};

\draw[<->, very thick] (s0) -- node[midway, above] {} (m);
\draw[<->, very thick] (s1) -- node[midway, above] {} (m);
\draw[<->, very thick] (s2) -- node[midway, below] {} (m);

\draw[<->, very thick] (m) -- node[midway, above] {} (c1);
\draw[<->, very thick] (m) -- node[midway, below] {} (c2);

\draw[<->, blue, very thick] (s0) to[bend left=20] node[midway, below] {} (c1);
\draw[<->, blue, very thick] (s1) to[bend left=20] node[midway, above] {} (c1);
\draw[<->, blue, very thick] (s1) to[bend right=20] node[midway] {} (c2);
\draw[<->, blue, very thick] (s2) to[bend right=20] node[midway, below] {} (c2);

\end{tikzpicture}
\caption{Our model for a digital content supply chain with suppliers, a intermediary, and consumers (Section \ref{sec:model}). 
The supplier offers a technology to produce content, the intermediary produces content, and the consumers consume content. 
The suppliers also offer the technology to the consumers, so the consumers have the option to directly produce content and bypass the intermediary (the blue arrows). }
\label{fig:model}
\end{figure}
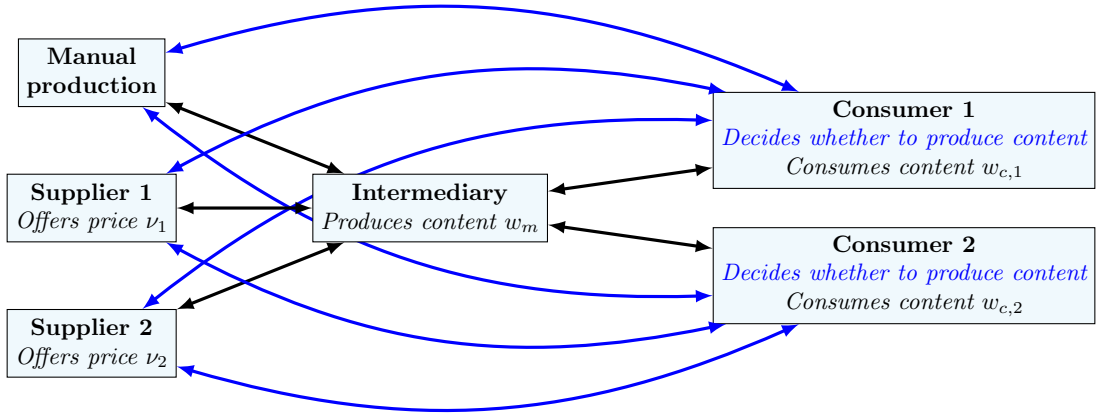

We develop the following model of a digital content supply chain (Figure \ref{fig:model}). There are $\NumProviders \ge 2$ homogeneous suppliers\footnote{See Section \ref{subsec:monopolist} for a generalization to a monopolist supplier.}, a monopolist intermediary (a content creator), and $\NumConsumers \ge 2$ homogeneous consumers.\footnote{While the consumer population is finite, we emphasize that they are homogeneous and do not interact strategically with each other in our model. We model consumers as finite, rather than a large-market continuum, to allow an apples-to-apples comparison of production costs borne by consumers and by the intermediary.}  We embed content into a single dimension that captures content quality.  We normalize our quality measure according to consumer utility: we say that content has quality $\w \in \mathbb{R}_{\ge 0}$ if a consumer derives utility $\w$ from consuming it.  Producing content incurs a cost that scales with its quality: we write $\TokFn(\w)$ for the cost of (manually) producing content of quality $\w$, where $\TokFn$ is an increasing function. We will place technical conditions on 
cost function $\TokFn$ in Section \ref{subsec:generalcase}.

The suppliers offer a technology that assists with digital content creation. The technology automates some---though not necessarily all---aspects of content creation. Furthermore, the technology is available not only to the intermediary, but also to consumers. The supplier incurs supply-side costs for operating the technology that scale with content quality: it costs the supplier $\ic \cdot \TokFn(\w)$ to operate the technology to assist in creating content with quality $\w$, where $\ic > 0$ describes the supplier's relative marginal cost.  When the technology is used for content creation, the additional human labor required to create content is captured by the human-driven production costs, which scale with content quality $\w$ according to $\costadditional \cdot \TokFn(\w)$ where $\costadditional \ge 0$.
Alternatively, content can be created manually (i.e., without the use of the technology), and these manual production costs scale with content quality $\w$ according to $\costmanual \cdot \TokFn(\w)$. It will be notationally convenient to define $\costmanual = 1$, so that $\TokFn(\w)$ is the cost of creating content manually. 

Each supplier sets a price $\nu$ (which can be different from $\ic$) and charges $\price \cdot \TokFn(\w)$ for operating the technology to assist in creating content with quality $\w$. The intermediary selects which supplier to link to (or whether to create content manually), and also strategically chooses the quality of the content that they produce. Each consumer then strategically decides whether to consume the content that the intermediary has produced, paying a fee to the intermediary in order to do so, or to directly produce content themselves for their own personal consumption.   

To further motivate our model, we illustrate how it provides a framework to study the impact of generative AI on content creation supply chains. 
\begin{example} 
\label{example:genai}
Consider the ongoing trend of generative AI being integrated into content creation supply chains. To capture this, let the suppliers correspond to companies which serve generative models (e.g., text-to-image models and text-to-video models) to users. Let the function $\TokFn(\w)$ capture the number of tokens needed to produce content with quality $\w$.\footnote{We might expect that the number of tokens $\TokFn(\w)$ to increase with content quality if the user needs to go back and forth over more rounds of dialogue to obtain higher-quality, or if higher-quality content requires more tokens to generate than lower-quality content.} Let the supply-side costs $\ic \cdot g(w)$ capture company's inference costs from querying the model, which tend to scale with the number of generated tokens. Let the pricing structure of suppliers capture setting a price-per-token and charging users based on the number of tokens, which is a common approach used in practice.\footnote{See \url{https://openai.com/api/pricing/}.} In contrast, the marginal transaction fee $\alpha$ can represent either a fixed platform fee (perhaps in the form of advertising) or a market rate for the commissioned service of the creator's content.

At a conceptual level, our model captures 
how generative AI models reduce the expertise needed for content creation, as reflected by how consumers and the intermediary can both leverage generative AI to assist with content creation. Nonetheless, our model also accounts the possibility of human-driven production costs which persist even in the presence of generative AI.

\end{example}

\paragraph{Stages of the Game.} The game proceeds in the following stages: 
\begin{enumerate}
    \item Each supplier $i \in [\NumProviders]$ chooses a price $\price_i$ for their technology. 
    \item The intermediary choose a content quality $\wmiddleman$.\footnote{We could have included an option for the intermediary to opt out. However, this is already implicitly captured by the intermediary producing content with quality $0$.} To produce $\wmiddleman$, they either choose a provider $\providerchoicemiddleman\in [\NumProviders]$ whose technology to use,  or they decide to create content manually $\providerchoicemiddleman = 0$. If $\providerchoicemiddleman = 0$, they incur a manual production cost of $\costmanual \cdot \TokFn(\wmiddleman)$; if $\providerchoicemiddleman \in [\NumProviders]$, they 
    pay $\price_{\providerchoicemiddleman} \cdot \TokFn(\wmiddleman)$ to supplier $i$ and also incur a human-driven production cost of $\costadditional \cdot \TokFn(\wmiddleman)$.  
    \item Each consumer $j \in [\NumConsumers]$ chooses a mode of consumption $\action{j} \in \left\{\Middleman, \Direct \right\}$.\footnote{We could have also included an option for the consumer to opt out. However, the consumer would never be incentivized to opt out, since they can always choose the direct creation option $\Direct$ and create content $\w = 0$ with quality level $0$ for free.}  
    \begin{itemize}
        \item If they choose $\action{j} = \Middleman$ (the intermediary option), they pay the fee $\alpha > 0$ to the intermediary and consume $\wuser{j} := \wmiddleman$. 
        \item   If they choose $\action{j} = \Direct$ (the direct creation option), they instead  choose quality of the content $\wuser{j}$ which they will produce and consume. To produce the content, they either choose a provider $\providerchoice{j} \in [\NumProviders]$ whose technology to use,  or they decide to create content manually $\providerchoice{j} = 0$. If $\providerchoice{j} = 0$, they incur a manual production cost of $\costmanual \cdot \TokFn(\wuser{j})$. If $\providerchoice{j} \in [\NumProviders]$, they 
    pay $\price_{\providerchoice{j}} \cdot \TokFn(\wuser{j})$ to supplier $i$ and also incur a human-driven production cost of $\costadditional \cdot \TokFn(\wuser{j})$. 
    \end{itemize}
 
\end{enumerate}

This game captures several features of digital content creation. First, observe that the intermediary incurs the same production costs regardless of how many consumers consume the content: this captures how digital content is typically free to distribute, regardless of how many consumers consume the content.\footnote{See Section \ref{subsec:marginalcosts} for an extension to the case of nonzero marginal costs.} Moreover, observe the intermediary (the content creator) receives the same fee from consumers regardless of content quality. This captures how when content creators rely on online platforms to share their content with consumers, it is common for creators to be rewarded based on exposure.  This exposure is proportional to the size of the creator's audience; our model abstracts away from inter-creator forces and content types that would determine the size of that audience.  On the consumer side, this fee could capture either the subscription fees paid to the platform or disutility from being shown advertisements.

\subsection{Utility functions}

We specify the utility functions of the suppliers, intermediary, and consumers. Each supplier $i$ derives profit equal to their revenue from usage of their technology minus supply-side costs: 
\[
\underbrace{1[\providerchoicemiddleman = i] \cdot (\price_i \cdot \TokFn(\wmiddleman) - \ic \cdot \TokFn(\wmiddleman))}_{\text{profit from intermediary usage}} + \sum_{j=1}^{\NumConsumers} \underbrace{1[\action{j} = \Direct] \cdot 1[\providerchoice{j} = i] \cdot (\price_i \cdot \TokFn(\wuser{j}) - \ic \cdot \TokFn(\wuser{j}))}_{\text{profit from consumer $j$'s usage}}.\]
The intermediary derives utility equal to their revenue from consumer fees minus production costs: 
\[ \sum_{j=1}^{\NumConsumers} \underbrace{1[\action{j} = M] \cdot \alpha}_{\text{revenue from consumer $j$}} - \underbrace{1[\providerchoicemiddleman \in [\NumProviders]] \cdot \left(\price_{\providerchoicemiddleman} \cdot \TokFn(\wmiddleman) + \costadditional \cdot \TokFn(\wmiddleman)\right)}_{\text{costs if technology is used}} - \underbrace{1[\providerchoicemiddleman = 0] \cdot \left( \price_{0} \cdot \TokFn(\wmiddleman)\right)}_{\text{manual costs}}.\]
Each consumer $j$ derives utility equal to the quality of the content they consume minus fees paid to the intermediary or production costs, depending on their chosen mode of consumption: 
\[\underbrace{\wuser{j}}_{\text{quality}} - \underbrace{1[\action{j} = M] \cdot \alpha}_{\text{intermediary fee}} -  \underbrace{1[\action{j} = D] \cdot \left(1[\providerchoice{j} \in [\NumProviders]] \cdot (\price_{\providerchoice{j}} \cdot \TokFn(\wuser{j}) + \costadditional \cdot \TokFn(\wuser{j})) + 1[\providerchoice{j} =  0] \cdot \left( \costmanual \cdot \TokFn(\wuser{j})\right)\right)}_{\text{production costs}}.\]

\subsection{Equilibrium concept and equilibrium existence}\label{subsec:equilibriumexistence}

We focus on the pure strategy \textit{subgame perfect equilibria} in the game between suppliers, the intermediary, and users. 
The following result shows that a pure strategy equilibrium exists (proof deferred to Appendix \ref{appendix:proofsequilibria}).
\begin{theorem}
\label{thm:equilibriumconstruction}
There exists a pure strategy equilibrium in the game between suppliers, the intermediary, and consumers. 
\end{theorem}

When we place additional structure on tiebreaking, we can also show a partial uniqueness result. In particular, we assume the following structure on tiebreaking: each consumer $j$ tiebreaks in favor of the intermediary (i.e., $a_j = M$), the intermediary tiebreaks in favor of producing higher-quality content over lower-quality content, the intermediary and consumers tiebreak in favor of suppliers over manual production, the intermediary and consumers tiebreak in favor of suppliers with a lower index.   
\begin{theorem}
\label{thm:equilibriumuniqueness}
Under the tiebreaking assumptions described above, the actions of the intermediary and consumers are the same at every pure strategy equilibrium. Moreover, the production cost $\price = \min(\costadditional + \min_{i \in [\NumProviders]} \price_i, \costmanual)$ is the same at every pure strategy equilibrium.  
\end{theorem}
We will focus on this class of equilibria throughout our analysis in Sections \ref{sec:disintermediation}-\ref{sec:utility}. 
The formal equilibrium construction is deferred to Appendix~\ref{appendix:proofsequilibria}, but let us briefly describe its structure.  At equilibrium, all suppliers will select the same price $\price_i$, which (due to competitive pressure) will be equal to the supplier's marginal production cost $\ic$.  There are then two cases, corresponding to two types of equilibria. 
\begin{itemize}
    \item In the disintermediation case, the intermediary chooses to produce at quality $0$ (i.e., exits the market).  Each of the $C$ consumers then produces, directly and separately, their own content at a utility-maximizing level of quality, either by contracting with the minimum-index supplier or generating the content manually, whichever is cheapest. 
    \item In the intermediation case, the intermediary produces at positive quality and all consumers choose to consume the content created by the intermediary. 
\end{itemize}
 For each choice of model parameters, exactly one of these two cases applies.

\section{Characterization of Disintermediation}\label{sec:disintermediation}

\begin{figure}
    \centering
    \begin{subfigure}[b]{0.49\textwidth}
        \centering
        \includegraphics[width=\textwidth]{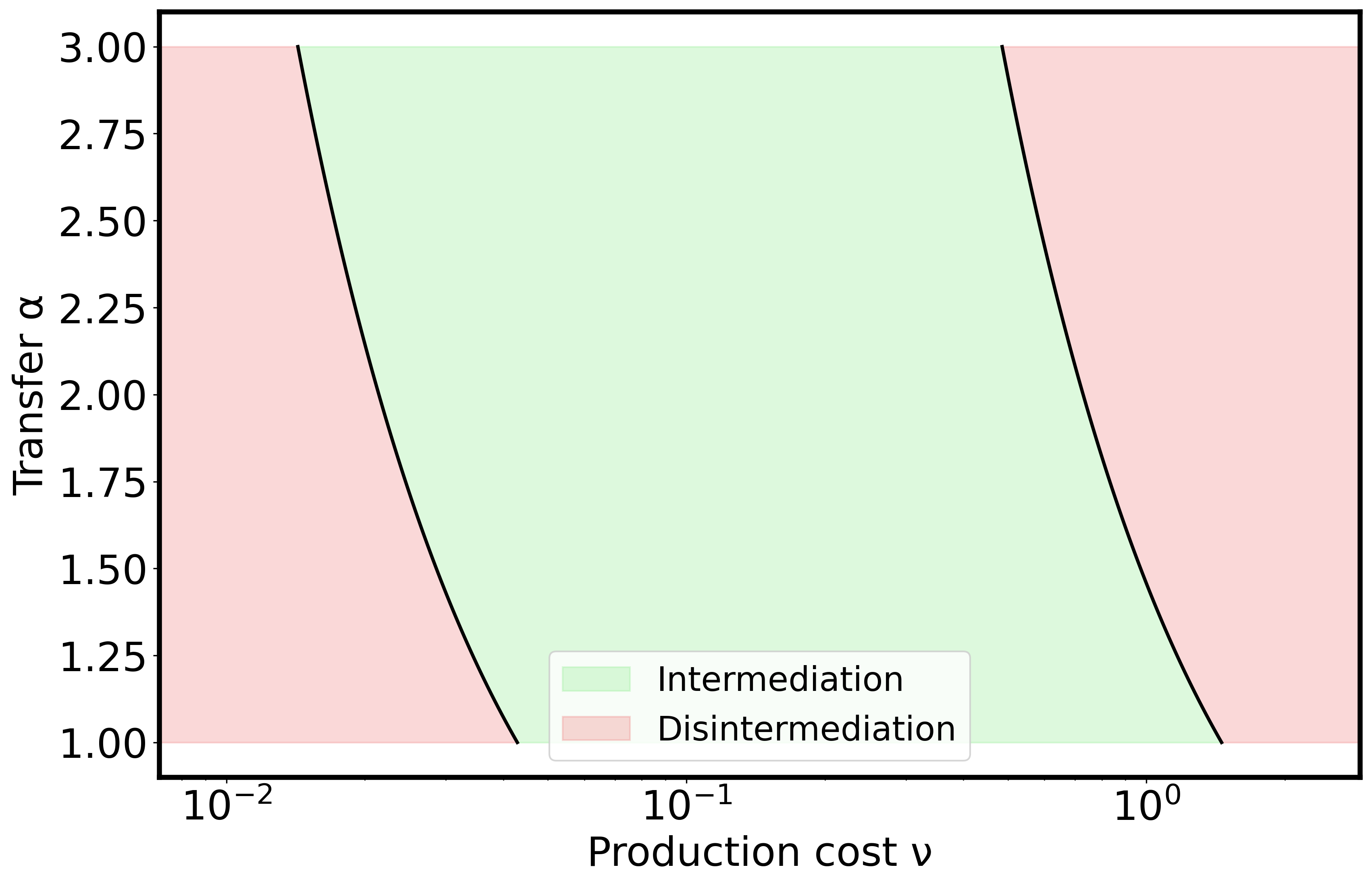}
        \caption{Impact of transfer $\subfee$}
        \label{fig:thresholds-alpha}
    \end{subfigure}
    \hfill
    \begin{subfigure}[b]{0.49\textwidth}
        \centering
        \includegraphics[width=\textwidth]{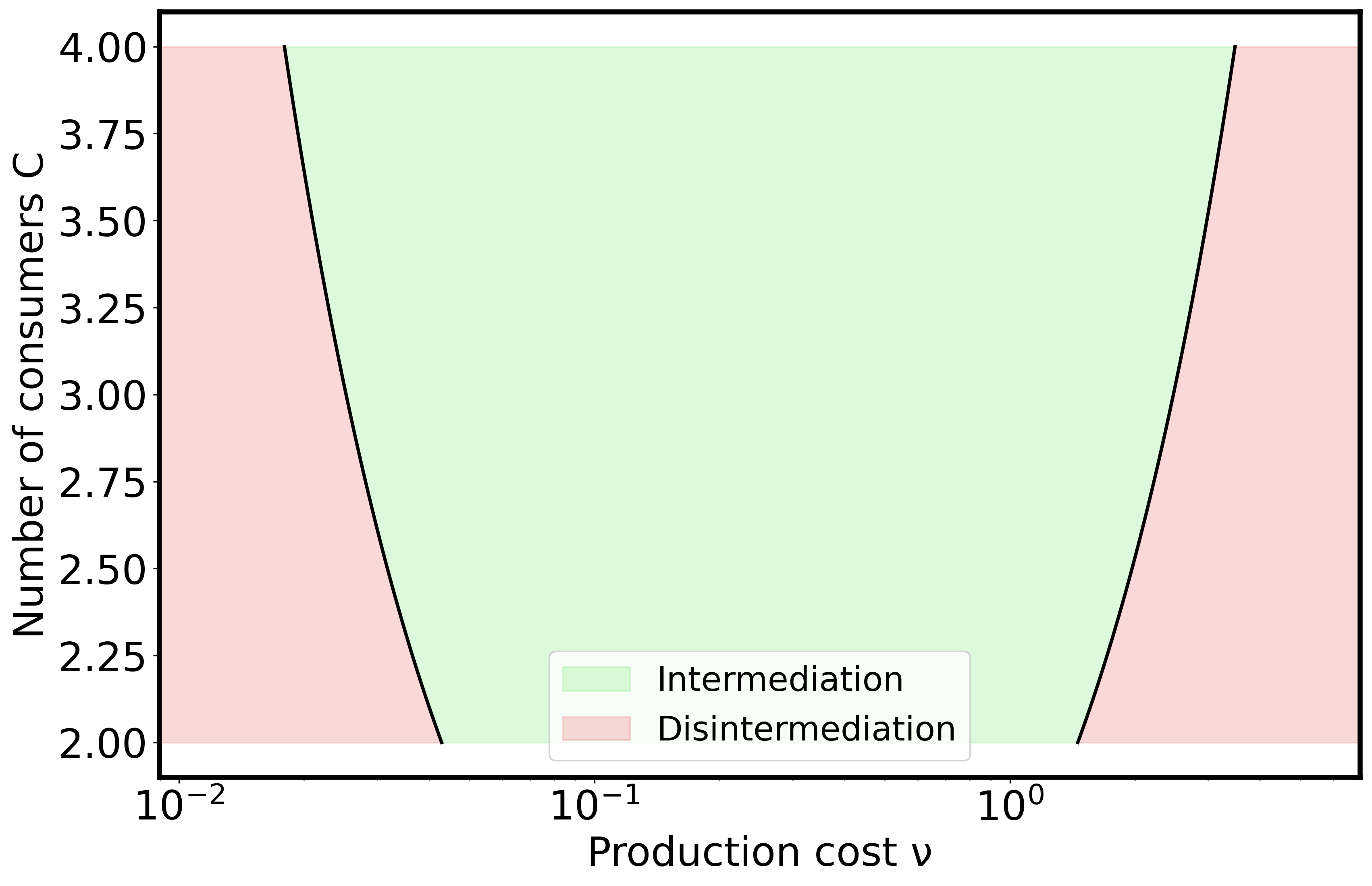}
        \caption{Impact of number of consumers $\NumConsumers$}
        \label{fig:thresholds-C}
    \end{subfigure}
    \caption{Production costs where disintermediation (red) vs. intermediation (green) occurs, for $g(\w) = \w^2$ (Theorem \ref{thm:specialcasemiddlemanusageexposure}). We vary the transfer $\alpha$ (left) and number of consumers $C$ (right). The intermediary only survives in the market when the production costs are at intermediate levels: the intermediary is driven out of the market when production costs are sufficiently low or sufficiently high. The range of values where intermediation occurs shifts lower when the fees $\subfee$ are higher, and the range expands when the number of consumers is larger. (We have generated these plots with a small number of consumers for ease of visualization, but our results apply for any number of consumers.)}
    \label{fig:thresholds}
\end{figure}

In this section, we characterize when the disintermediation occurs: that is, when the intermediary does not survive in the market. To study this, we analyze the \textit{intermediary usage} $\sum_{j=1}^{\NumConsumers} \mathbb{E}[1[\action{j} = \Middleman]]$ which measures the number of consumers who consume content produced by the intermediary, at equilibrium. In Section \ref{subsec:specialcase}, we characterize the intermediary usage in the special class of cost functions of the form $\TokFn(\w) = \w^{\TokExp}$. In Section \ref{subsec:generalcase}, we extend this result to general cost functions $\TokFn$. 

\subsection{Special Class of Cost Functions: $\TokFn(\w) = \w^{\TokExp}$}\label{subsec:specialcase}

To gain intuition, we first consider 
cost functions that are power functions of the form $\TokFn(\w) = \w^{\TokExp}$ for $\TokExp > 1$. This class of functions permits the following characterization. 

\begin{theorem}
\label{thm:specialcasemiddlemanusageexposure}
Let $\TokFn(\w) = \w^{\TokExp}$ where $\TokExp > 1$. Fix $\subfee > 0$, $\NumConsumers > 1$. Suppose there are $\NumProviders > 1$ providers. There exist thresholds $0 < \LowerThreshold( \NumConsumers, \subfee, \TokExp) <  \UpperThreshold(\NumConsumers, \subfee, \TokExp) < \infty$ such that the intermediary usage at equilibrium satisfies:
\[ 
\sum_{j=1}^{\NumConsumers} \mathbb{E}[1[\action{j} = \Middleman]] =
\begin{cases}
0 & \text{ if } \min(\ic + \costadditional, \costmanual) < \LowerThreshold(\NumConsumers, \subfee, \TokExp) \\
\NumConsumers & \text{ if } \min(\ic + \costadditional, \costmanual) \in [\LowerThreshold(\NumConsumers,\subfee,  \TokExp), \UpperThreshold(\NumConsumers, \subfee, \TokExp)]\\
0 & \text{ if } \min(\ic + \costadditional, \costmanual) > \UpperThreshold( \NumConsumers, \subfee, \TokExp) \\
\end{cases}
\]
The thresholds $\LowerThreshold( \NumConsumers,\subfee, \TokExp)$ and $\UpperThreshold( \NumConsumers,\subfee, \TokExp)$ are the two unique solutions to: 
\[\price^{-\frac{1}{\TokExp(\TokExp-1)}} \cdot \left(\TokExp^{-\frac{1}{\TokExp-1}} - \TokExp^{-\frac{\TokExp}{\TokExp-1}} \right) \cdot \alpha^{-\frac{1}{\TokExp}} + \price^{\frac{1}{\TokExp}} \cdot \alpha^{\frac{\TokExp-1}{\TokExp}}  = \NumConsumers^{\frac{1}{\TokExp}}.\] Moreover, the lower threshold $\LowerThreshold( \NumConsumers,\subfee, \TokExp)$ is decreasing as a function of the number of consumers $\NumConsumers$, and the upper threshold $\UpperThreshold( \NumConsumers, \subfee, \TokExp)$ is increasing as a function of $\NumConsumers$. 
\end{theorem}

Theorem \ref{thm:specialcasemiddlemanusageexposure} (Figure \ref{fig:thresholds}) demonstrates that the equilibrium intermediary usage exhibits up to \textit{three regimes of behavior} as a function of the production costs. In the first regime, which occurs when production costs are small, disintermediation occurs and consumers do not leverage the intermediary for content production. In the middle regime, intermediation instead occurs: all of the consumers rely on the intermediary for content production.  In the last regime, where production costs are very large, disintermediation again occurs. Notably, the middle regime of intermediation is larger when the number of consumers $\NumConsumers$ is large.   

The intuition for Theorem \ref{thm:specialcasemiddlemanusageexposure} is as follows. The intermediary's main advantage 
is that they only need to incur the production cost once regardless of how many consumers consume their content, and yet their revenue still scales with the number of consumers that they attract. This economy of scale is more pronounced when the number of consumers is large. However, the intermediary's ability to leverage this advantage is constrained by the fact that the fee paid by consumers is fixed and might not align with production costs.  This creates friction that can impede the intermediary's ability to generate utility. When costs are in the middle regime, the 
misalignment is outweighed by the intermediary's fixed-cost advantage: the intermediary can create much higher quality than what the consumer can afford to create for themselves, and the consumers are willing to pay the requisite fee for this additional quality.

When costs are sufficiently low, the consumers are incentivized to create even higher quality content than what the intermediary can afford with the fees they collect.
When costs are sufficiently high, the fee is insufficient to cover the cost of producing quality that the consumers find acceptable for the price, so consumers are incentivized to create lower-quality content to reduce costs.

Taking a closer look at the structure of production costs, Theorem \ref{thm:specialcasemiddlemanusageexposure} also disentangles how 
two different axes of technological advances--- reductions in supply-side costs and reductions in human-driven production costs---both affect whether disintermediation occurs. Specifically, production costs are captured by $\min(\ic + \costadditional, \costmanual)$. Assuming that manual production costs $\costmanual$ exceed production costs $\ic+\costadditional$ from leveraging the technology, the first regime of disintermediation only occurs when technology improves to a sufficient degree along both of these axes. If supply-side costs $\ic$ are overly high, then even if the human-driven costs $\costadditional$ are pushed down to zero (i.e., fully automation is possible), this regime will not occur; similarly, if human-driven costs $\costadditional$ are overly high, then even if the supply-side costs $\ic$ are pushed down to zero (i.e., technology operation is free), then this regime will also not occur. 

These findings have interesting implications for the integration of generative AI into content supply chains (Example \ref{example:genai}). 
Prior to the generative AI era, the manual production costs likely placed the 
content creation ecosystem within the second regime, where all consumers rely on the intermediary (i.e., content creators) for production. If generative AI continues to reduce production costs, the ecosystem could move to the first regime where creators are cut out of the ecosystem. However, this shift would rely on technological advances along two axes: inference costs (which depend on the computational efficiency of querying these models), and human-driven production costs (which depends on the balance between automation vs. augmentation in content creation).

We provide a proof sketch of Theorem \ref{thm:specialcasemiddlemanusageexposure}. 
\begin{proof}[Proof sketch of Theorem \ref{thm:specialcasemiddlemanusageexposure}]
Disintermediation occurs when the intermediary can't afford to match the consumer utility from direct usage: that is, when the costs of producing content  achieving that consumer utility level exceeds the intermediary's revenue from consumer usage. Lemma \ref{lemma:equilibriumcharacterization} shows that this occurs if and only if 
\begin{equation}
\label{eq:conditionrestated}
 (\ic + \costadditional) \cdot \TokFn(\subfee + \max_{w \ge 0} (\w - (\ic + \costadditional)\TokFn(w))) > \subfee \NumConsumers.   
\end{equation}
The intuition for \eqref{eq:conditionrestated} is that $\max_{w \ge 0} (\w - (\ic + \costadditional)\TokFn(w))$ is the utility that the consumer would have achieved from direct usage, so a content quality $\subfee + \max_{w \ge 0} (\w - (\ic + \costadditional)\TokFn(w))$ is needed to match that utility level while also offsetting the fee that the consumer pays to the intermediary. The intermediary can't survive in the market if their production cost $ (\ic + \costadditional) \cdot \TokFn(\subfee + \max_{w \ge 0} (\w - (\ic + \costadditional)\TokFn(w)))$ exceeds their revenue $\subfee \NumConsumers$ from fees. 

Using the fact that $\TokFn(\w) = \w^{\TokExp}$, we explicitly compute when this condition is satisfied: 
\begin{equation}
\label{eq:simplifiedcondition}
\price^{-\frac{1}{\TokExp(\TokExp-1)}} \cdot \left(\TokExp^{-\frac{1}{\TokExp-1}} - \TokExp^{-\frac{\TokExp}{\TokExp-1}} \right) \cdot \subfee^{-\frac{1}{\TokExp}}  + \price^{\frac{1}{\TokExp}} \cdot \alpha^{\frac{\TokExp - 1}{\TokExp}}  > \NumConsumers^{\frac{1}{\TokExp}},
\end{equation}
where $\price = \ic + \costadditional$. The left-hand side of \eqref{eq:simplifiedcondition} is concave in $\price$ (since it is the sum of two concave functions), and it approaches $\infty$ as $\price \rightarrow \infty$ and as $\price \rightarrow 0$. Moreover, the minimum value of the left-hand side across all $\price \in (0, \infty)$ is $1$ which violates \eqref{eq:simplifiedcondition}. This establishes the existence of thresholds $0 < \LowerThreshold( \NumConsumers, \subfee, \TokExp) <  \UpperThreshold(\NumConsumers, \subfee, \TokExp) < \infty$. Moreover, these properties, together with the fact that the right-hand of \eqref{eq:simplifiedcondition} is increasing in $C$, imply that the lower threshold is decreasing in $\NumConsumers$ and the upper threshold is increasing in $\NumConsumers$. The full proof is deferred to Appendix \ref{appendix:proofsdintermediation}.
\end{proof}

\subsection{General Cost Functions}\label{subsec:generalcase}

We now move beyond the particular functional form in Section \ref{subsec:specialcase}, and analyze disintermediation for  more general cost functions $\TokFn$.

Specifically, we consider strictly increasing, continuously differentiable cost functions $\TokFn$ which are (1) strictly convex, (2) satisfy $\TokFn(0) = \TokFn'(0) = 0$ and $\lim_{\w \rightarrow \infty} \TokFn(\w) = \lim_{\w \rightarrow \infty} \TokFn'(\w) = \infty$, and  (3) strictly log-concave. 
Some examples of cost functions that satisfy assumptions (1)-(3) are  $\TokFn(\w) = \w^{\TokExp}$ for $\TokExp > 1$, $\TokFn(\w) = \w^{\TokExp} \cdot e^{\w}$ for $\TokExp \ge 1$, $\TokFn(\w) = \w^{\TokExp} \cdot e^{\sqrt{\w}}$ for $\TokExp \ge 1$, and $\TokFn(\w) = \w^{\TokExp} \cdot (\log(\w+1)^{\TokExpSecond})$ for any $\TokExp, \TokExpSecond > 1$ (Proposition \ref{prop:costs}). To elucidate the role of these assumptions, suppose that the consumer selects the direct creation option and optimally chooses their content quality level to maximally their utility. The first two assumptions imply that the consumer chooses content quality in the interior of $(0, \infty)$. Taken together with the first two assumptions, the third assumption implies that as production costs become cheaper, a consumer who directly creates their own content would expend more on content production.

Under these assumptions on the cost function $\TokFn$, we show a partial generalization of Theorem \ref{thm:specialcasemiddlemanusageexposure}. The following result demonstrates that the intermediary usage exhibits up to three regimes of behavior as a function of the production costs.
\begin{theorem}
\label{thm:middlemanusageexposure}
Let $\TokFn$ be a strictly increasing, continuously differentiable function which is strictly convex, satisfies $\TokFn(0) = \TokFn'(0) = 0$ and $\lim_{\w \rightarrow \infty} \TokFn(\w) = \lim_{\w \rightarrow \infty} \TokFn'(\w) = \infty$, and strictly log-concave. 
Fix $\subfee > 0$, $\NumConsumers > 1$. Suppose there are $\NumProviders > 1$ providers. There exist thresholds $ \LowerThreshold(\NumConsumers, \subfee, \TokFn) < \UpperThreshold(\NumConsumers, \subfee, \TokFn) \le \infty$ such that the intermediary usage at equilibrium satisfies:
\[ 
\sum_{j=1}^{\NumConsumers} \mathbb{E}[1[\action{j} = \Middleman]] =
\begin{cases}
0 & \text{ if } \min(\ic + \costadditional, \costmanual) < \LowerThreshold(\NumConsumers, \subfee, \TokFn)   \\
\NumConsumers & \text{ if } \min(\ic + \costadditional, \costmanual) \in [\LowerThreshold(\NumConsumers, \subfee, \TokFn), \UpperThreshold(\NumConsumers, \subfee,  \TokFn)]\\
0 & \text{ if } \min(\ic + \costadditional, \costmanual) > \UpperThreshold(\NumConsumers, \subfee, \TokFn).
\end{cases}
\]
\end{theorem}

While Theorem \ref{thm:middlemanusageexposure}
 provides a partial generalization of Theorem \ref{thm:specialcasemiddlemanusageexposure}, a key difference is that Theorem \ref{thm:middlemanusageexposure} 
 only guarantees that there are \textit{up to} three regimes rather than \textit{exactly} three regimes. Crucially, Theorem \ref{thm:middlemanusageexposure} does not guarantee that disintermediation occurs when as production costs tend to zero. 
To help address this, we show a sufficient condition for having exactly three regimes. 
\begin{theorem}
\label{thm:sufficientcondition}
Consider the setup of Theorem \ref{thm:middlemanusageexposure}. Suppose also that 
\[\lim_{\w \rightarrow \infty} \frac{\TokFn\left(\w - \frac{\TokFn(\w)}{\TokFn'(\w)} \right)}{\TokFn'(\w)} = \infty.\]
Then it holds that $0 < \LowerThreshold(\NumConsumers, \subfee, \TokFn) < \UpperThreshold(\NumConsumers, \subfee, \TokFn) < \infty$. 
The lower threshold is decreasing as a function of the number of consumers $\NumConsumers$, and the upper threshold is increasing as a function of $\NumConsumers$. The thresholds are the two unique solutions to $\price \cdot g(\subfee + \max_{\w \ge 0} (\w - \price \TokFn(\w))) = \subfee \NumConsumers$. 
\end{theorem}

Theorem \ref{thm:sufficientcondition} provides a more complete generalization of Theorem \ref{thm:specialcasemiddlemanusageexposure} and guarantees that there are exactly three regimes. Specifically, disintermediation occurs when the production costs tend to zero and also occurs when production costs become sufficiently large. 

The condition in 
Theorem \ref{thm:sufficientcondition} bears resemblance to requiring that the log derivative approaches $\infty$ in the limit as $\w \rightarrow \infty$, which would instead take the form $\lim_{\w \rightarrow \infty} \frac{\TokFn(\w)}{\TokFn'(\w)} = \infty$. This simpler condition captures that the amount that the consumer expends on content production, if they directly use the technology, is unbounded in the limit as production costs become cheaper. However, the requirement in Theorem \ref{thm:sufficientcondition}  is slightly stricter since the numerator is replaced with 
$\TokFn\left(\w - \frac{\TokFn(\w)}{\TokFn'(\w)}\right)$. Some examples of functions $\TokFn$ satisfying the conditions in Theorem \ref{thm:sufficientcondition} are $\TokFn(\w) = \w^{\TokExp}$ for $\TokExp > 1$, $\TokFn(\w) = \w^{\TokExp} \cdot e^{\sqrt{\w}}$ for $\TokExp > 1$, and $\TokFn(\w) = \w^{\TokExp} \cdot (\log(\w+1)^{\TokExpSecond})$ for any $\TokExp, \TokExpSecond > 1$ (Proposition \ref{prop:costsstronger}).

We note that not all functions satisfy this condition. For example, functions of the form $\TokFn(\w) = \w^{\TokExp} \cdot e^{\w}$ for $\TokExp > 1$ do not satisfy the conditions in Theorem \ref{thm:sufficientcondition}, but do satisfy the conditions in Theorem \ref{thm:middlemanusageexposure}. 
For this function class, this is not just an artifact of the analysis:  we show that the intermediary survives even when production costs tend to zero. 

\begin{proposition}
\label{prop:counterexample}
Consider the setup of Theorem \ref{thm:middlemanusageexposure}. Let $\TokFn = \w^{\TokExp} \cdot e^{\w}$ for $\TokExp > 1$. Fix $\subfee > 1$ and $\NumConsumers > 1$ satisfying $e^{\subfee } < \subfee \cdot \NumConsumers$. Then,  $\LowerThreshold(\NumConsumers, \subfee, \TokFn) \le 0$. That is, the intermediary usage $\sum_{j=1}^{\NumConsumers} \mathbb{E}[1[\action{j} = \Middleman]] = \NumConsumers$ even when the production costs $\min(\ic + \costadditional, \costmanual)$ are sufficiently small. 
\end{proposition}

\section{Welfare Consequences}\label{sec:utility}

\begin{figure}
    \centering
    \begin{subfigure}[b]{0.49\textwidth}
        \centering
        \includegraphics[width=\textwidth]{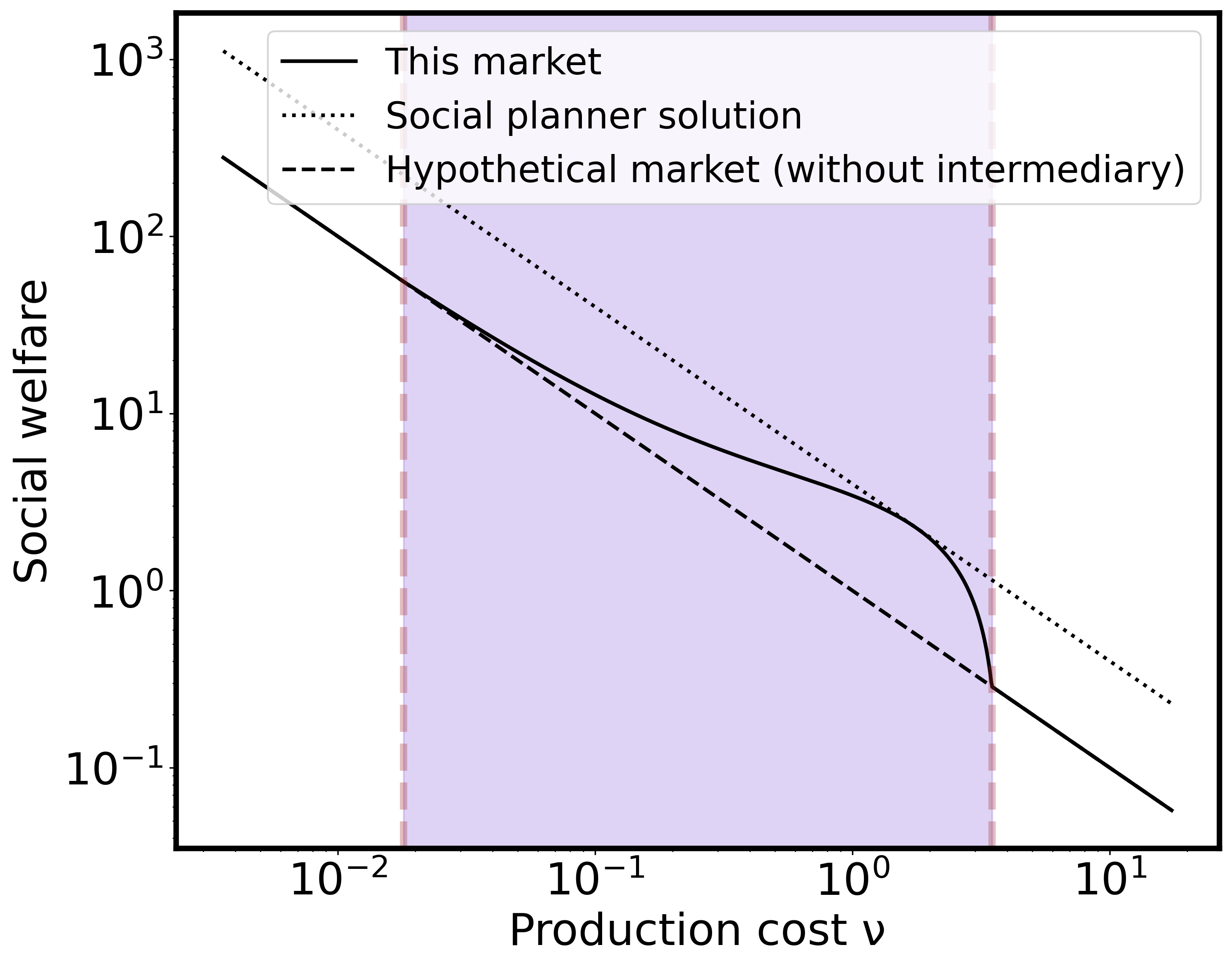}
        \caption{Social welfare}
        \label{fig:welfare-comparison}
    \end{subfigure}
    \begin{subfigure}[b]{0.49\textwidth}
        \centering
        \includegraphics[width=\textwidth]{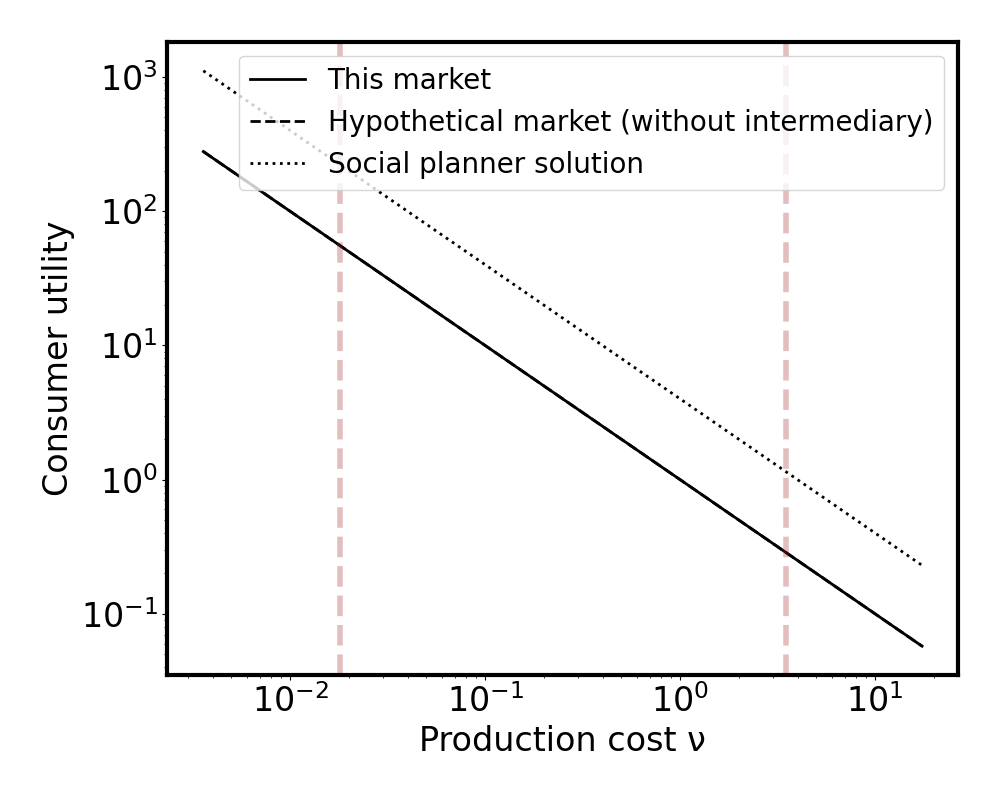}
        \caption{Consumer utility}
        \label{fig:consumer-utility-comparison}
    \end{subfigure}
    \begin{subfigure}[b]{0.49\textwidth}
        \centering
        \includegraphics[width=\textwidth]{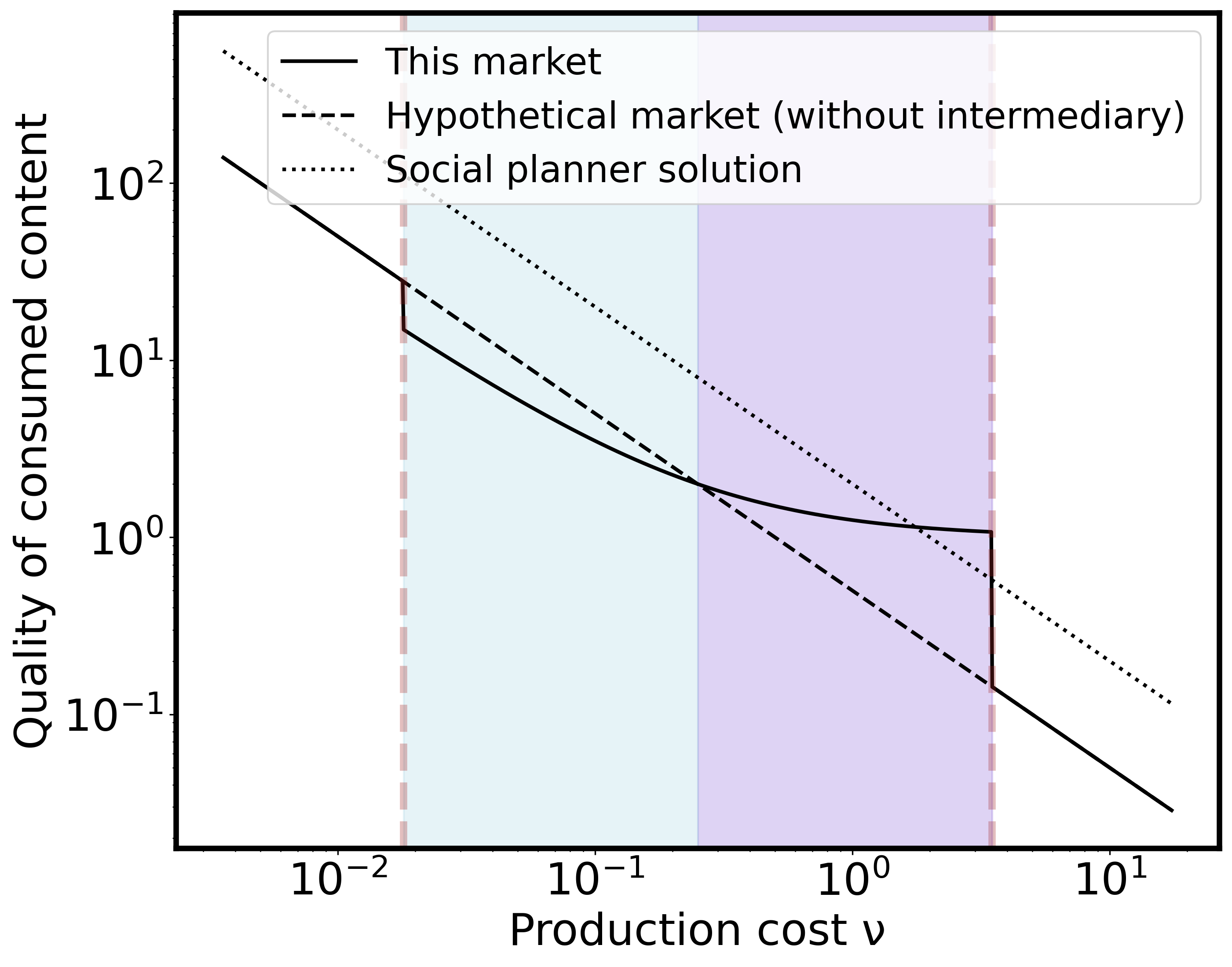}
        \caption{Content quality}
        \label{fig:quality-comparison}
    \end{subfigure}
    \caption{Analysis of social welfare, consumer utility, and content quality in this market in comparison to a hypothetical market where the intermediary does not exist, for $g(\w) = \w^2$. We show how the the intermediary increases (purple), decreases (blue), or does not affect (white) each of these metrics. The intermediary is always welfare-improving (left; Theorem \ref{thm:socialwelfare}). However, the intermediary does not increase consumer utility (middle; Theorem \ref{thm:consumerutility}), and instead extracts all of the surplus for themselves. The intermediary can increase content quality or decrease content quality (right; Theorem \ref{thm:contentquality}).
    } 
    \label{fig:metrics-comparison}
\end{figure}

Having established when disintermediation occurs, we now turn to the consequences for social welfare and the overall digital economy. We study the impact on content quality (Section \ref{subsec:quality}), the intermediary's utility (Section \ref{subsec:middleman}), consumer utility (Section \ref{subsec:consumers}), and social welfare (Section \ref{subsec:welfare}). We analyze how these metrics change with technology improvements, and to gain intuition for the impact of the intermediary, we also make comparisons to a hypothetical market where the intermediary does not exist (Figure \ref{fig:metrics-comparison}). 
Throughout this section, we focus on the setup of Theorem \ref{thm:sufficientcondition} where disintermediation occurs exactly at the extreme values of production costs (i.e., where there are three regimes of behaviors).

\subsection{Quality of content}\label{subsec:quality}

We first examine how disintermediation impacts content quality. To gain intuition, we compute a closed-form characterization of the quality of the content consumed at equilibrium. 
\begin{proposition}
\label{prop:contentquality}
Consider the setup of Theorem \ref{thm:sufficientcondition}. Let $\price = \min(\ic + \costadditional, \costmanual)$. At equilibrium, the quality $\wuser{j}$ of the content consumed by any consumer $j \in [\NumConsumers]$ is: 
\[ 
\begin{cases}
\argmax_{\w \ge 0}(\w - \price \cdot \TokFn(\w)) & \text{ if } \price < \LowerThreshold( \NumConsumers, \subfee, \TokFn)  \\
 \subfee + \max_{\w \ge 0}(\w - \price \cdot \TokFn(\w))  & \text{ if } \price \in [\LowerThreshold(\NumConsumers, \subfee, \TokFn), \UpperThreshold(\NumConsumers, \subfee,  \TokFn)]\\
\argmax_{\w \ge 0}(\w - \price \cdot \TokFn(\w))  & \text{ if } \price > \UpperThreshold( \NumConsumers, \subfee, \TokFn).
\end{cases}.
\]
\end{proposition}

\begin{figure}
    \centering
    \begin{subfigure}[b]{0.49\textwidth}
        \centering
        \includegraphics[width=\textwidth]{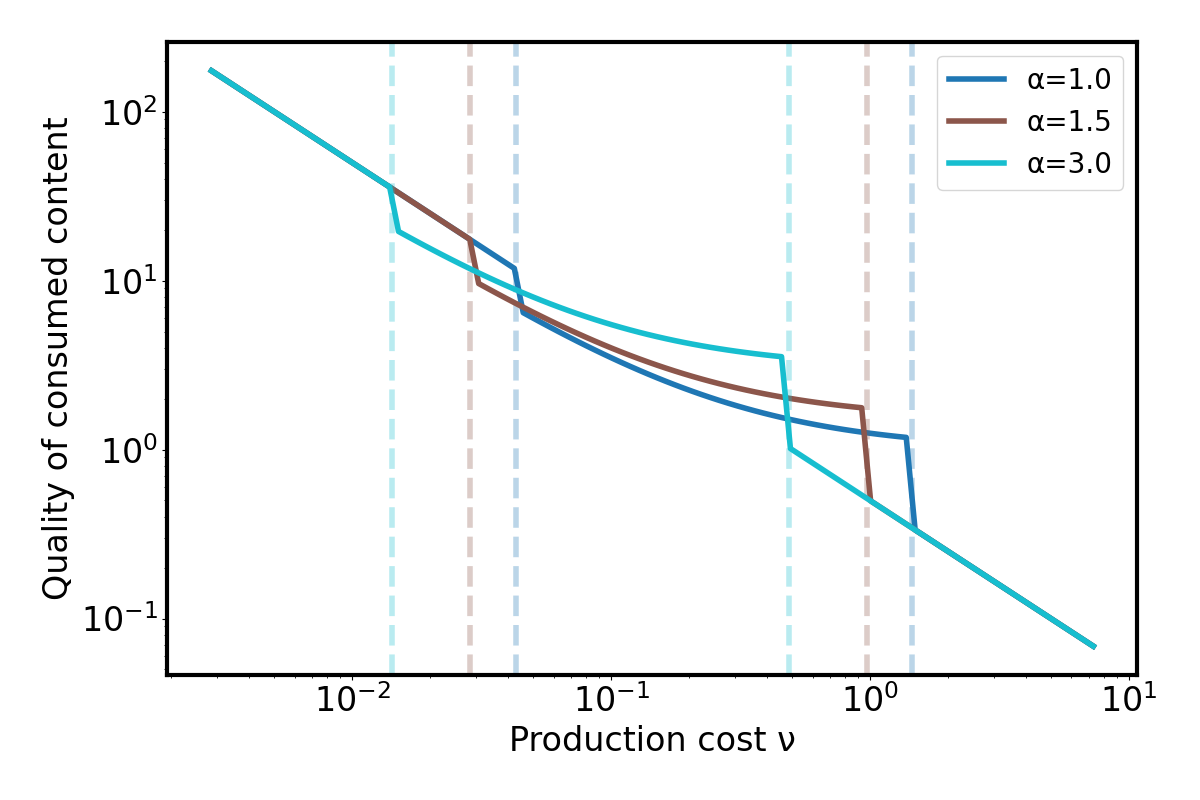}
        \caption{Impact of transfer $\subfee$}
        \label{fig:quality-alpha}
    \end{subfigure}
    \hfill
    \begin{subfigure}[b]{0.49\textwidth}
        \centering
        \includegraphics[width=\textwidth]{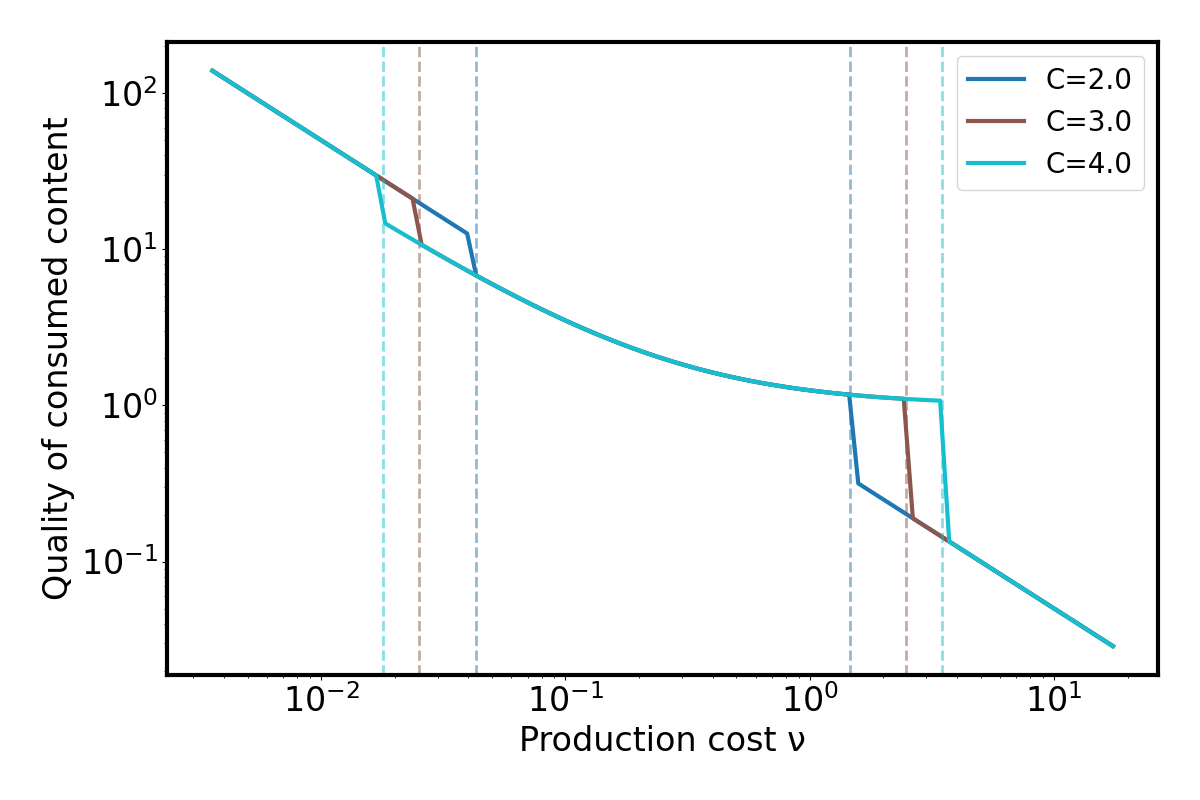}
        \caption{Impact of number of consumers $\NumConsumers$}
        \label{fig:quality-C}
    \end{subfigure}
    \caption{Quality of the content consumed at a pure strategy equilibrium as a function of production costs, for $g(\w) = \w^{2}$ (Proposition \ref{prop:contentquality}) We vary the transfers $\alpha$ (left), and the number of consumers $\NumConsumers$ (right). The vertical dashed lines show the production costs at which disintermediation starts to occur. Observe that the quality is decreasing in production costs, and is discontinuous at the thresholds where disintermediation starts to occur (Theorem \ref{thm:contentquality}).
    } 
    \label{fig:quality}
\end{figure}

Proposition \ref{prop:contentquality} demonstrates that the content quality has three regimes of behavior as a function of the production costs. These are conceptually the same three regimes as in Theorem \ref{thm:middlemanusageexposure}. In the first and third regimes, the consumer consumes  the content that they produce themselves; in the second regime, the intermediary survives in the market, and consumers consume the content that is produced by the intermediary.

Using Proposition \ref{prop:contentquality}, we analyze how the content quality changes as production costs improve, and we also compare content quality to a hypothetical market where the intermediary does not exist, where the content quality would have been  $\argmax_{\w \ge 0}(\w - \price \cdot \TokFn(\w))$. 
\begin{theorem}
\label{thm:contentquality}
Consider the setup of Proposition \ref{prop:contentquality}. The quality of content consumed at equilibrium is decreasing in $\price$. Moreover, the quality is continuous in $\price$ except for  at the thresholds $\LowerThreshold( \NumConsumers,\subfee, \TokFn)$ and $\UpperThreshold( \NumConsumers, \subfee, \TokFn)$. The quality when $\price \in [\LowerThreshold(\NumConsumers, \subfee, \TokFn), \UpperThreshold(\NumConsumers, \subfee,  \TokFn)]$ can be higher or lower than  $\argmax_{\w \ge 0}(\w - \price \cdot \TokFn(\w))$ when $\price$ is at the higher or lower end of the range, respectively. 
\end{theorem}

Theorem \ref{thm:contentquality} (Figure \ref{fig:quality-comparison}) illustrates how the intermediary distorts content quality. First, the shape of the curve qualitatively changes in the presence of the intermediary: specifically, the slope of the curve becomes flatter. This means that the intermediary reduces the responsiveness of quality to technology changes. Moreover, the intermediary can raise or lower content quality compared to a hypothetical market where the intermediary does not exist. Specifically, we see that when the production costs at the lower end of the regime where the intermediary survives, the content quality is lower than in this hypothetical market; when the production costs are at the upper end of the regime, then the content quality is higher than in this hypothetical market. A striking consequence is that disintermediation at the lower threshold \textit{improves} content quality, even though the market no longer benefits from economies of scale from the intermediary.  

Theorem \ref{thm:contentquality} (Figure \ref{fig:quality}) also uncovers other global properties of the content quality in this market. Even though the market transitions between intermediation and disintermediation, content quality is decreasing with production costs: this illustrates how technological improvements consistently improves content quality. However, perhaps counterintuitively, increasing the number of consumers can lead to \textit{lower} content quality for some production costs (Figure \ref{fig:quality-C}). This comes as a side effect of intermediation, since the number of consumers impacts the range of production costs where intermediation occurs. The impact of the fees $\subfee$ is similarly ambiguous, since the fees also impact when intermediation occurs (Figure \ref{fig:quality-alpha}). 

\subsection{Intermediary utility}\label{subsec:middleman}

\begin{figure}
    \centering
    \begin{subfigure}[b]{0.49\textwidth}
        \centering
        \includegraphics[width=\textwidth]{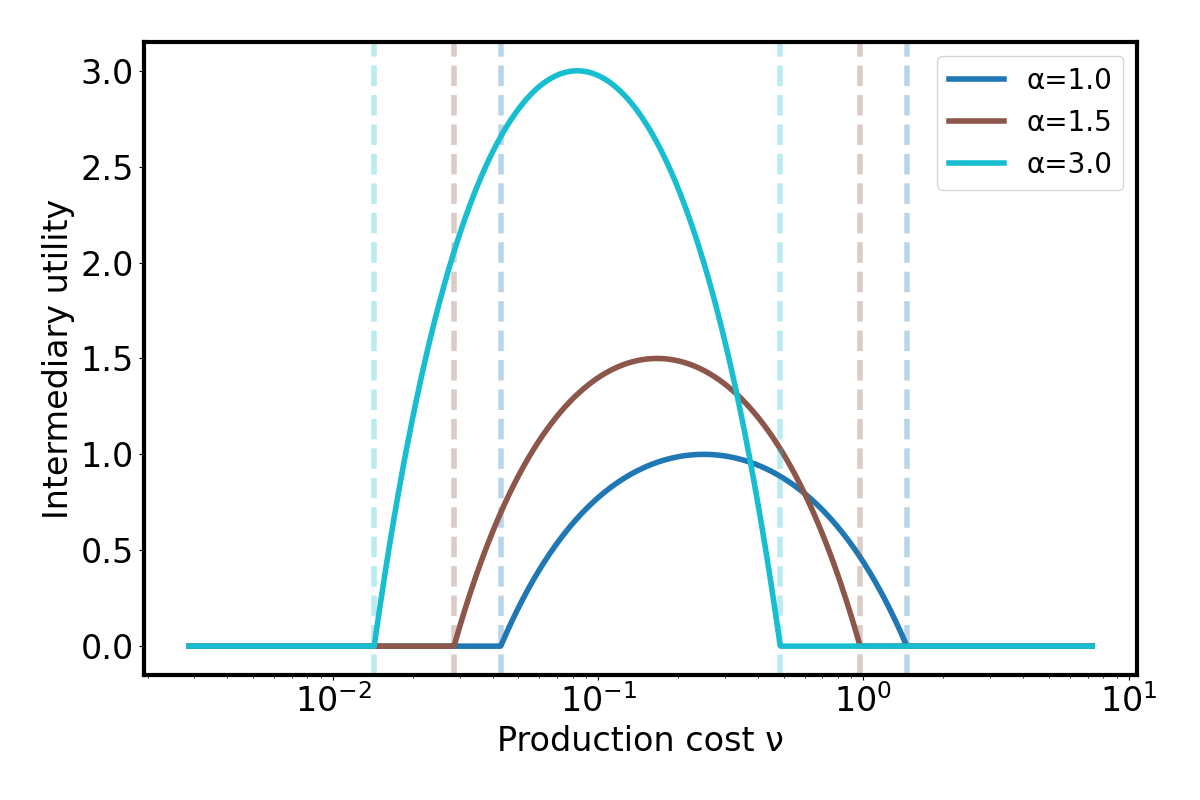}
        \caption{Impact of transfer $\subfee$}
        \label{fig:middleman-alpha}
    \end{subfigure}
    \hfill
    \begin{subfigure}[b]{0.49\textwidth}
        \centering
        \includegraphics[width=\textwidth]{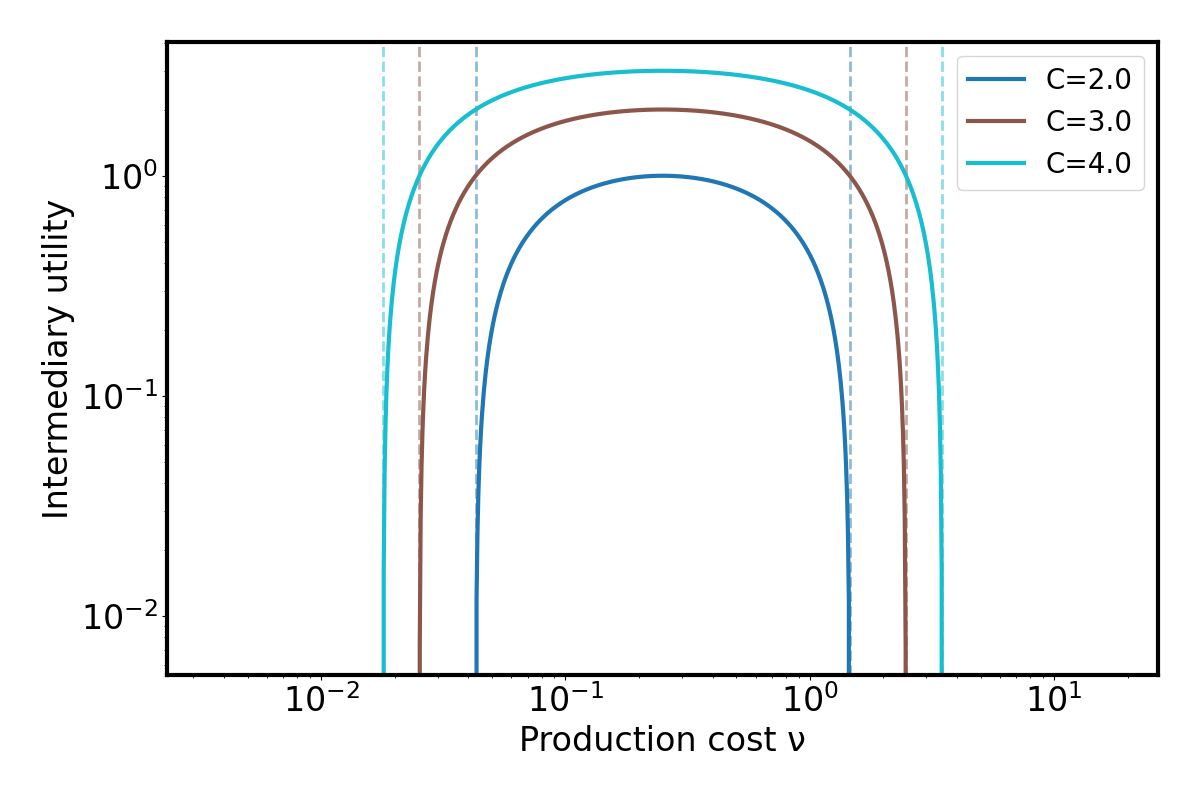}
        \caption{Impact of number of consumers $\NumConsumers$}
        \label{fig:middleman-C}
    \end{subfigure}
    \caption{Intermediary utility at a pure strategy equilibrium as a function of production costs, for $g(\w) = \w^{2}$ (Proposition \ref{prop:middlemanutility}).  We vary the transfers $\alpha$ (left), and the number of consumers $\NumConsumers$ (right). The vertical dashed lines show the production costs at which disintermediation starts to occur. Observe that the intermediary utility is inverse U-shaped in production costs (Theorem \ref{thm:middlemanutility}).}
    \label{fig:middleman}
\end{figure}

We next turn to the intermediary's utility. To gain intuition, we first compute a closed-form characterization of the intermediary's utility at equilibrium. 
\begin{proposition}
\label{prop:middlemanutility}
Consider the setup of Theorem \ref{thm:sufficientcondition}. Let $\price = \min(\ic + \costadditional, \costmanual)$. The intermediary's utility at equilibrium is of the form:  
\[ 
\begin{cases}
0 & \text{ if } \price < \LowerThreshold(\NumConsumers, \subfee, \TokFn)  \\
 \subfee \NumConsumers - \price \TokFn\left(\subfee +  \max_{w \ge 0} \left( w - \price \cdot \TokFn(\w)\right) \right)  & \text{ if } \price \in [\LowerThreshold(\NumConsumers, \subfee, \TokFn), \UpperThreshold(\NumConsumers, \subfee,  \TokFn)]\\
0 & \text{ if } \price > \UpperThreshold(\NumConsumers, \subfee, \TokFn).
\end{cases}
\]  
\end{proposition}

Proposition \ref{prop:middlemanutility} lets us analyze how intermediary utility changes with production costs. 
\begin{theorem}
\label{thm:middlemanutility}
Consider the setup of Proposition \ref{prop:middlemanutility}. 
As a function of $\price$, the intermediary utility at equilibrium is continuous and inverse U-shaped. The maximum intermediary utility across all values $\price > 0$ is equal to $\subfee (\NumConsumers - 1)$ and occurs when $\wmiddleman = \argmax_{w \ge 0} (w - \price g(w))$.  
\end{theorem}

Theorem \ref{thm:middlemanutility} (Figure \ref{fig:middleman}) illustrates how the intermediary's utility is inverse U-shaped. This non-monotone behavior means that even though technology improvements first benefit the intermediary, the intermediary's utility later starts to fall until the intermediary is eventually driven out of the market. The intermediary's utility is maximized when production costs are in the middle of the range. 
At the optima, the intermediary expends one consumer's fee on content production, creating the same content that the consumer would have created if the intermediary did not exist. The intermediary retains the rest of the consumers' fees for themselves: in this sense, the intermediary  extracts all the value from the economies of scale. The intermediary benefits from increasing the number of consumers (Figure \ref{fig:middleman-C}), but the impact of the fee $\subfee$ is ambiguous (Figure \ref{fig:middleman-alpha}).

\subsection{Consumer utility}\label{subsec:consumers}

\begin{figure}
    \centering
    \begin{subfigure}[b]{0.49\textwidth}
        \centering
        \includegraphics[width=\textwidth]{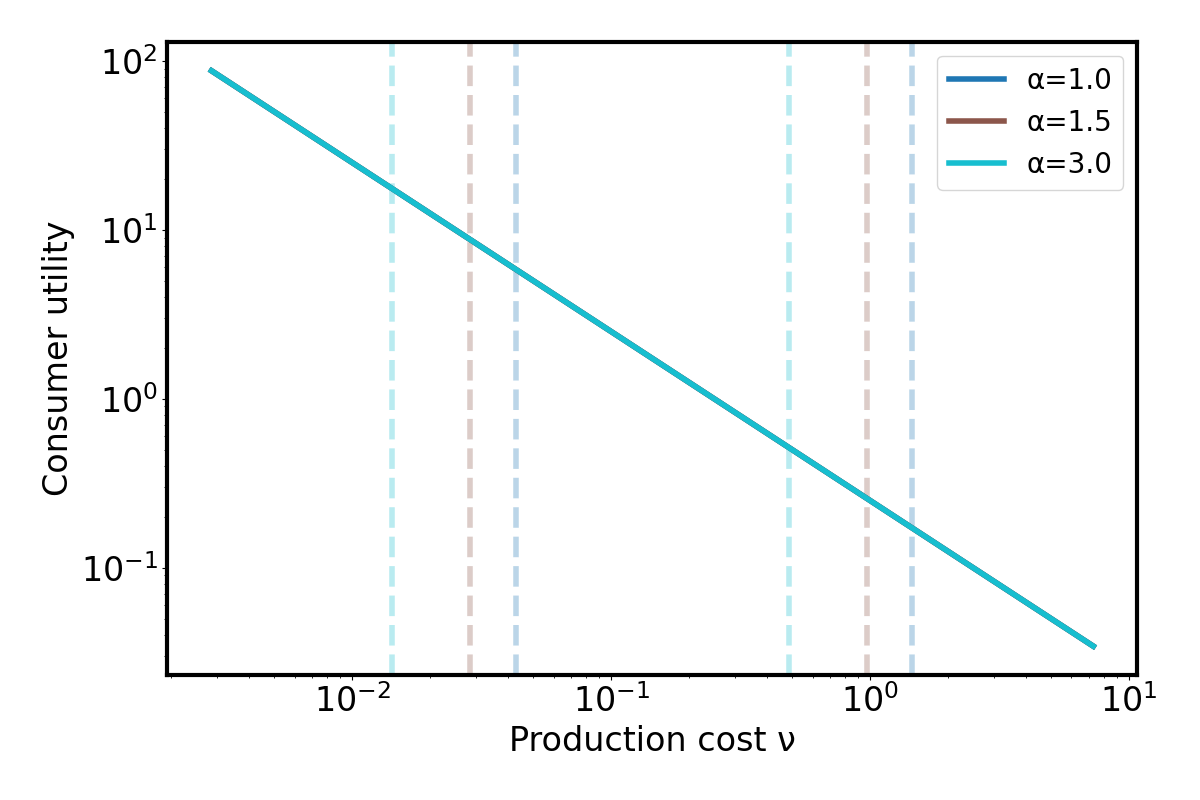}
        \caption{Impact of transfer $\subfee$}
        \label{fig:consumer-alpha}
    \end{subfigure}
    \hfill
    \begin{subfigure}[b]{0.49\textwidth}
        \centering
        \includegraphics[width=\textwidth]{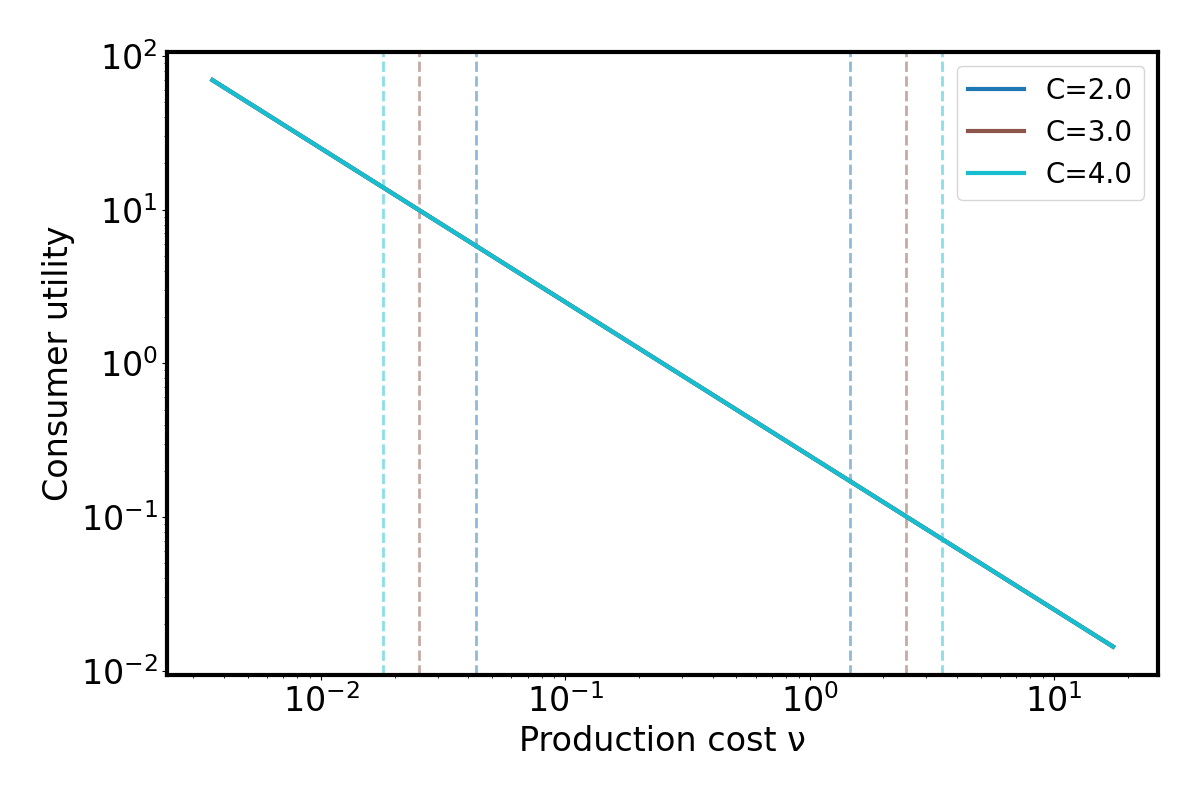}
        \caption{Impact of number of consumers $\NumConsumers$}
        \label{fig:consumer-C}
    \end{subfigure}
    \caption{Consumer utility at a pure strategy equilibrium as a function of production costs, for $g(\w) = \w^{2}$ (Theorem \ref{thm:consumerutility}). We vary the transfers $\alpha$ (left), and the number of consumers $\NumConsumers$ (right). Observe that the consumer utility is continuous, decreasing in production costs, and independent of $\NumConsumers$ and $\subfee$ (Corollary \ref{cor:consumerutility}).}.
    \label{fig:consumer}
\end{figure}

We next turn to consumer utility, which can be characterized at equilibrium in closed-form. 
\begin{theorem}
\label{thm:consumerutility}
Consider the setup of Theorem \ref{thm:sufficientcondition}. Let $\price = \min(\ic + \costadditional, \costmanual)$. At equilibrium, the utility of any consumer $j \in [\NumConsumers]$ is equal to $\max_{\w \ge 0} (\w - \price \cdot \TokFn(\w))$. 
\end{theorem}

Theorem \ref{thm:consumerutility} 
(Figure \ref{fig:consumer-utility-comparison})  illustrates how consumer utility is unaffected by the intermediary. Specifically, the consumer utility is the same as  in a hypothetical market where the intermediary does not exist. This means that consumer utility is independent of the fees $\subfee$ (Figure \ref{fig:consumer-alpha}) as well as the number of other consumers in the market (Figure \ref{fig:consumer-C}). The intuition is that the intermediary is a monopolist, and is able to extract all of the value from the economies of scale for themselves. Interestingly, this occurs even though the intermediary can't influence the price $\price$: instead the intermediary extracts all of the surplus through the choice of quality produced.

Using Theorem \ref{thm:consumerutility}, we show that consumer utility is decreasing is production costs, so consumers still do benefit from technological improvements.
\begin{corollary}
\label{cor:consumerutility}
 Consider the setup of Theorem \ref{thm:consumerutility}. As a function of $\price$, the the utility of each consumer $j \in [\NumConsumers]$ is continuous and decreasing.    
\end{corollary}

\subsection{Social welfare}\label{subsec:welfare}

Finally, we turn to social welfare. We first analyze the social welfare at equilibrium in closed-form.  
\begin{figure}
    \centering
    \begin{subfigure}[b]{0.49\textwidth}
        \centering
        \includegraphics[width=\textwidth]{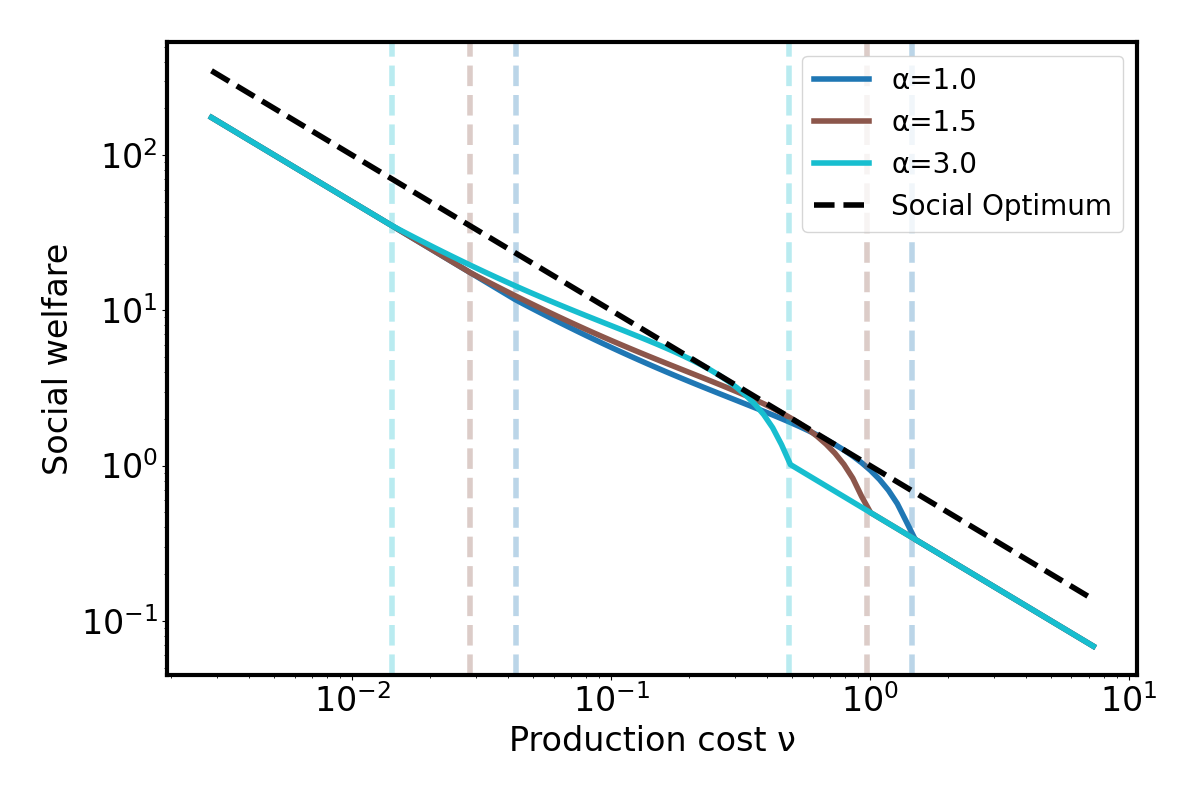}
        \caption{Impact of transfer $\subfee$}
        \label{fig:social-welfare-alpha}
    \end{subfigure}
    \hfill
    \begin{subfigure}[b]{0.49\textwidth}
        \centering
        \includegraphics[width=\textwidth]{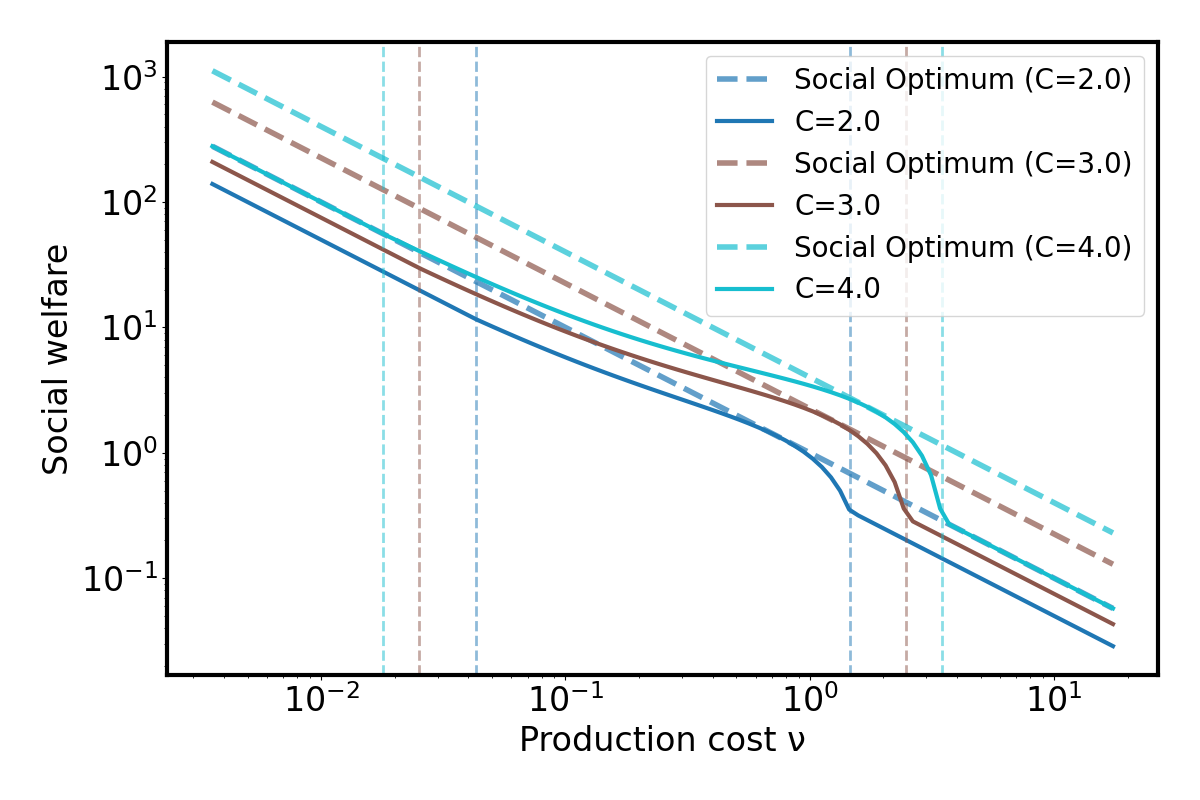}
        \caption{Impact of number of consumers $\NumConsumers$}
        \label{fig:social-welfare-C}
    \end{subfigure}
    \caption{Social welfare at a pure strategy equilibrium as a function of production costs, for $g(\w) = \w^{2}$ (Proposition \ref{prop:socialwelfare}) We vary the transfers $\alpha$ (left), and the number of consumers $\NumConsumers$ (right). The black line shows the social welfare of the optimal social planner solution. Observe that the social welfare utility is continuous, decreasing in production costs, and increasing in $\NumConsumers$ (Theorem \ref{thm:socialwelfare})}
    \label{fig:socialwelfare}
\end{figure}

\begin{proposition}
\label{prop:socialwelfare}
Consider the setup of Theorem \ref{thm:sufficientcondition}. Let $\price = \min(\ic + \costadditional, \costmanual)$. The equilibrium social welfare takes the form: 
\[
\begin{cases}
\NumConsumers \cdot \left(\max_{\w \ge 0}\left(w - \price \cdot \TokFn(\w) \right)\right) & \text{ if } \price < \LowerThreshold( \NumConsumers, \subfee, \TokFn)  \\
\NumConsumers \cdot \left(\subfee + \max_{\w \ge 0}\left(w - \price \cdot \TokFn(\w) \right)\right) -\price \TokFn\left(\subfee +  \max_{w \ge 0} \left( w - \price \cdot \TokFn(\w)\right) \right)  & \text{ if } \price \in [\LowerThreshold(\NumConsumers, \subfee, \TokFn), \UpperThreshold(\NumConsumers, \subfee,  \TokFn)]\\
\NumConsumers \cdot \left(\max_{\w \ge 0}\left(w - \price \cdot \TokFn(\w) \right)\right) & \text{ if } \price > \UpperThreshold( \NumConsumers,\subfee, \TokFn).
\end{cases}
\]
\end{proposition}

To interpret the social welfare achieved in this market, we consider a social planner whose goal is to maximize the total social welfare of the suppliers, intermediary, and consumers. We characterize the optimal social planner solution, both in the case where the intermediary exists and where the intermediary does not exist.
\begin{proposition}
\label{prop:socialplanner}
 Let $\price = \min(\ic + \costadditional, \costmanual)$. If the intermediary exists, then the social planner's solution achieves social welfare 
\[
\max_{\w \ge 0} \left(C \w - \price \cdot \TokFn(\w) \right).
\] 
If the intermediary does not exist, then the social planner's solution achieves social welfare 
\[
C \cdot \max_{\w \ge 0} \left(\w - \price \cdot \TokFn(\w) \right).
\]
\end{proposition}

We now analyze how the social welfare achieved in the market changes with production costs and compares to the social planner's solutions. 
\begin{theorem}
\label{thm:socialwelfare}
Consider the setup of Proposition \ref{prop:socialwelfare}. The social welfare is continuous and decreasing in production costs. It is strictly below the social planner's optimal  except at at most one bliss point. Moreover, when $\price \in (\LowerThreshold(\NumConsumers, \subfee, \TokFn), \UpperThreshold(\NumConsumers, \subfee, \TokFn))$, the social welfare is strictly greater than the social planner's optimal without the intermediary (i.e., $\NumConsumers \cdot \left(\max_{\w \ge 0}\left(w - \price \cdot \TokFn(\w) \right)\right)$). 
\end{theorem}

Theorem \ref{thm:socialwelfare} (Figure \ref{fig:welfare-comparison}) shows that the intermediary is welfare-improving. Specifically, when the intermediary is present in the market, the social welfare at equilibrium is higher than the social planner solution in a hypothetical market where the intermediary does not exist. However, the social welfare almost always falls below the social planner solution which can take advantage of the intermediary, except at at most one bliss point. This bliss point always exists for costs of the form $g(\w) = w^{\beta}$ for $\beta > 1$ (Proposition \ref{prop:blisspoint}). The intuition is that the intermediary distorts content quality, producing too high-quality content when the production costs are above the bliss point and producing too low-quality content when the production costs are below the bliss point (Figure \ref{fig:quality-comparison}). Figure \ref{fig:socialwelfare} also suggests that the location of the bliss point appears to occur at higher production costs when the number of consumers is large (Figure \ref{fig:social-welfare-C}) and when the fees are smaller (Figure \ref{fig:social-welfare-alpha}).

Taken together with the earlier results in this section, this welfare analysis illustrates that while technology improvements lead to higher social welfare, the intermediary extracts all gains to social welfare. Specifically, any increase in social welfare over the hypothetical market without the intermediary is captured by the monopolist intermediary themselves. Interestingly, this occurs even though the intermediary controls only the quality $\wmiddleman$, not the price $\price$.

\section{Extensions}\label{sec:extensions}

\begin{figure}
    \centering
    \begin{subfigure}[b]{0.49\textwidth}
        \centering
        \includegraphics[width=\textwidth]{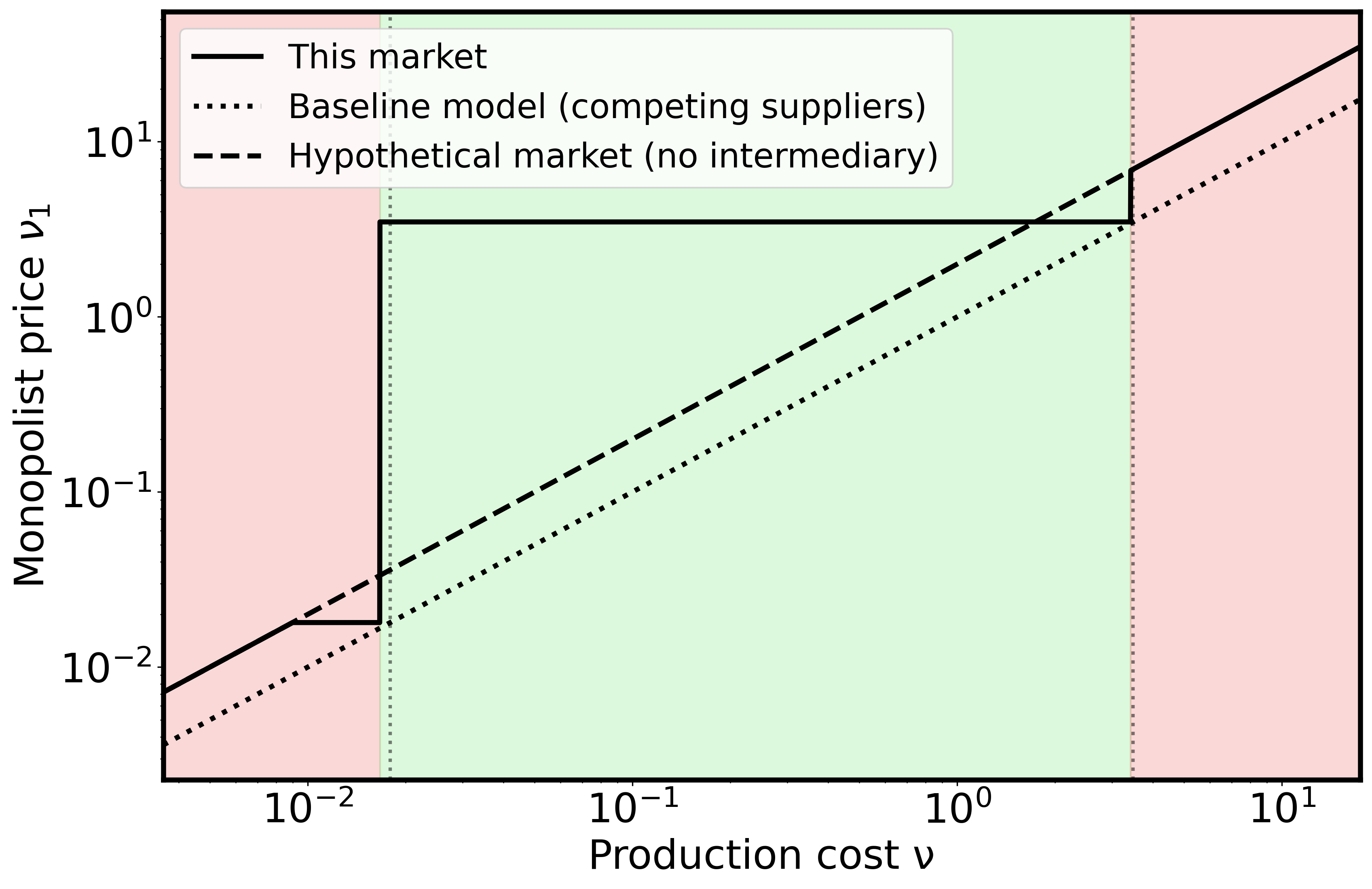}
        \caption{Monopolist supplier}
        \label{fig:monopolist}
    \end{subfigure}
    \hfill
    \begin{subfigure}[b]{0.49\textwidth}
        \centering
 \includegraphics[width=\textwidth]{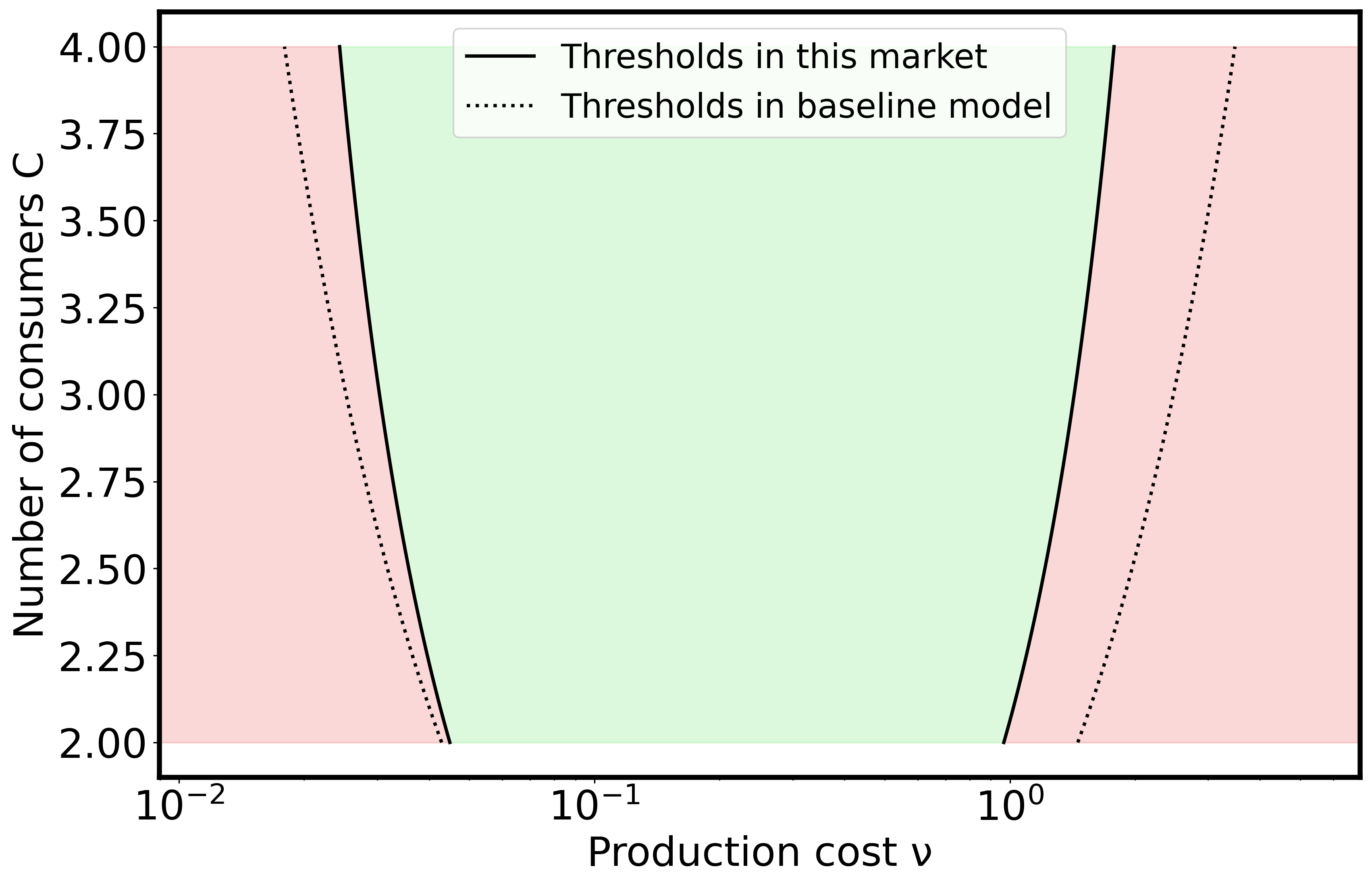}
        \caption{Nonzero marginal costs}
        \label{fig:marginalcosts}
    \end{subfigure}
    \caption{Production costs where disintermediation (red) vs. intermediation (green) for $\TokFn(\w) = w^2$. We consider extensions of the baseline model with a monopolist supplier (left; Theorem \ref{thm:extensionmonopolist}) and with nonzero marginal costs of production (right; Theorem \ref{thm:extensionmarginal}). In both cases, disintermediation still occurs when production costs are sufficiently low or sufficiently high. However, relative to our baseline model, the range of technology levels that support intermediation changes: the range shifts to be lower with a monopolist supplier (though by a small amount) and shrinks in width with nonzero marginal costs.} 
    \label{fig:extensions}
\end{figure}

To check the robustness of our findings  we consider extensions to our base model, focusing on the case $g(\w) = \w^{\TokExp}$ from Section \ref{subsec:specialcase} for simplicity. 
We find that our characterization of disintermediation from Section \ref{sec:disintermediation} readily generalizes to settings with a monopolist supplier (Figure \ref{fig:monopolist}; Section \ref{subsec:monopolist}) and where the intermediary faces nonzero marginal costs of production to serve each consumer (Figure \ref{fig:marginalcosts}; Section \ref{subsec:marginalcosts}): that is, disintermediation still occurs occurs at the extremes of production technology. However, we also show that disintermediation can be avoided entirely when users pay a fee that increases linearly with the quality of the content consumed: this demonstrates the importance of our modeling assumption in Section \ref{sec:disintermediation} that the intermediary does not have full control the fee structure offered to consumers in addition to content quality (Section \ref{subsec:transfers}).

\subsection{Monopolist supplier}\label{subsec:monopolist}

While Section \ref{sec:disintermediation} assumed competition between multiple suppliers, we now turn to the case of a single monopolist supplier. For simplicity, we focus on the case where manual production costs $\costmanual$ are infinite, meaning that it is cheaper to produce content using the technology rather than without the technology. The following result characterizes when disintermediation occurs in this setting. 
\begin{theorem}
\label{thm:extensionmonopolist}
Let $\TokFn(\w) = \w^{\TokExp}$ where $\TokExp > 1$.\footnote{For this result we assume that each consumer $j$ tiebreaks in favor of direct usage (i.e., $\action{j} = \Direct$) rather than in favor of the intermediary (i.e., $\action{j} = \Middleman$) when $\price < \UpperThreshold(\NumConsumers, \subfee, \TokExp)$, but in favor of the intermediary when $\price \ge \UpperThreshold(\NumConsumers, \subfee, \TokExp)$.} Fix $\costmanual = \infty$, $\subfee > 0$, and assume that $\NumConsumers > \frac{\TokExp}{\TokExp - 1}$.\footnote{We assume that the number of consumers is sufficiently large (i.e., $\NumConsumers > \frac{\TokExp}{\TokExp - 1}$) for technical convenience.} Suppose there is a monopolist supplier (i.e., $\NumProviders = 1$). There exist thresholds $0 < \LowerThresholdMon(\NumConsumers,\subfee,  \TokExp) < \UpperThresholdMon(\NumConsumers, \subfee, \TokExp) < \infty$ such that the intermediary usage at equilibrium satisfies:
\[ 
\sum_{j=1}^{\NumConsumers} \mathbb{E}[1[\action{j} = \Middleman]] =
\begin{cases}
0 & \text{ if } \min(\ic + \costadditional, \costmanual) \le \LowerThresholdMon( \NumConsumers,\subfee, \TokExp)  \\
\NumConsumers & \text{ if } \min(\ic + \costadditional, \costmanual) \in (\LowerThresholdMon(\NumConsumers, \subfee, \TokExp), \UpperThresholdMon( \NumConsumers, \subfee,\TokExp)]\\
0 & \text{ if }  \min(\ic + \costadditional, \costmanual) > \UpperThresholdMon( \NumConsumers, \subfee, \TokExp).
\end{cases}
\]
In comparison to the thresholds from Theorem \ref{thm:specialcasemiddlemanusageexposure}, these thresholds satisfy $\TokExp^{-1} \cdot \LowerThreshold(\NumConsumers, \subfee, \TokExp) < \LowerThresholdMon( \NumConsumers, \subfee, \TokExp) < \LowerThreshold(\NumConsumers, \subfee, \TokExp)$ and $\TokExp^{-1} \cdot \UpperThreshold(\NumConsumers, \subfee, \TokExp) < \UpperThresholdMon(\NumConsumers, \subfee, \TokFn) < \UpperThreshold( \NumConsumers, \subfee,\TokExp)$.
\end{theorem}

Theorem \ref{thm:extensionmonopolist} (Figure \ref{fig:monopolist}) shows that the insights from 
Theorem \ref{thm:specialcasemiddlemanusageexposure} readily generalize to the case of a monopolist supplier, albeit with the intermediation range shifted to lower production costs. Disintermediation still occurs whenever production costs are sufficiently low or sufficiently high. However, the upper and lower thresholds for intermediation in Theorem \ref{thm:extensionmonopolist} occur at lower production costs than the corresponding thresholds with competing suppliers.  The intuition is that suppliers set prices above the marginal production costs, so the consumers and intermediary face higher prices, which shifts the thresholds downwards. The intermediary can more easily survive when technology costs are lower, but is less likely to enter when technology costs are high. 

To prove Theorem \ref{thm:extensionmarginal}, a key step is to analyze how the monopolist supplier sets prices. Unlike for competing suppliers, the price $\price_1$ is no longer driven down to the marginal production cost $\ic$. The following lemma characterizes the optimal pricing decisions of the supplier. 
\begin{lemma}
\label{lemma:price}
Consider the setup of Theorem \ref{thm:extensionmonopolist}, and let $\LowerThreshold$ and $\UpperThreshold$ be defined according to Theorem \ref{thm:specialcasemiddlemanusageexposure}. Then the supplier's price $\price_1$ satisfies:
\[ 
\price_1 =
\begin{cases}
\TokExp \cdot \ic & \text{ if } \ic < \TokExp^{-1} \cdot \LowerThreshold( \NumConsumers, \subfee, \TokExp)  \\
\LowerThreshold(\NumConsumers,\subfee,  \TokExp) & \text{ if } \ic \ge \TokExp^{-1} \cdot \LowerThreshold(\NumConsumers, \subfee, \TokExp) \text{ and }
\ic \le \LowerThresholdMon(\NumConsumers, \subfee, \TokExp),   \\
\UpperThreshold(\NumConsumers,\subfee,  \TokExp) & \text{ if } \ic \in (\LowerThresholdMon(\NumConsumers, \subfee, \TokExp), \UpperThresholdMon( \NumConsumers, \subfee,\TokExp)]\\
\TokExp \cdot \ic & \text{ if } \ic > \UpperThresholdMon(\NumConsumers,\subfee,  \TokExp).
\end{cases}
\] 
\end{lemma}

To interpret Lemma \ref{lemma:price} (Figure \ref{fig:monopolist}), consider a hypothetical market where the intermediary does not exist.  In this hypothetical the monopolist provider would set $\price_1 = \ic \cdot \TokExp$ (Lemma \ref{lemma:maximumpriceconsumer}). By Lemma \ref{lemma:price}, when production costs are sufficiently low (or high), the monopolist sets the prices exactly as they would if the intermediary did not exist. However, for intermediate production costs, the monopolist supplier distorts prices to influence whether the intermediary survives in the market or not. To see why, first consider production costs at the lower end of this range, near (but higher than) the lower threshold $\LowerThreshold(\NumConsumers, \subfee, \TokExp)$ (where intermediation would start to occur if the price were equal to marginal production costs $\ic$). In this case, the monopolist supplier holds prices at $\LowerThreshold(\NumConsumers, \subfee, \TokExp)$ in order to \emph{avoid intermediation}.  That is, the supplier suppresses their price to make direct usage by consumers more attractive and prevent the intermediary from entering the market.  But if we grow the supplier's marginal production costs, these costs become sufficiently close to the supplier's price $\LowerThreshold(\NumConsumers, \subfee, \TokExp)$ that the supplier's profit becomes too low. At this point, the supplier allows the intermediary to enter and discontinuously shifts to the maximal price $\UpperThreshold(\NumConsumers, \subfee, \TokExp)$  that keeps the intermediary in the market. If we continue to increase costs, eventually the supplier's marginal production costs become sufficiently close to the price $\UpperThreshold(\NumConsumers, \subfee, \TokExp)$, at which point they drive the intermediary out of the market and once again set prices as they would if the intermediary did not exist.

\subsection{Marginal costs}\label{subsec:marginalcosts}

In Section \ref{sec:disintermediation}, we assumed that the intermediary faces no marginal costs for distributing content to additional consumers. We now relax this assumption and consider scenarios where the intermediary not only pays a fixed cost of $\price \cdot \TokFn(\w)$ to produce content $\w$, but also pays a small additional cost of $\marg \cdot \price \cdot \TokFn(\w)$ for every consumer to whom they serve the content. Here, we assume that $\marg < 1$. The consumer likewise faces the same cost structure: they pay a total cost of $\price (1 + \marg) \cdot \TokFn(\w)$ to produce content $\w$ for themselves. The following result characterizes when disintermediation occurs. 
\begin{theorem}
\label{thm:extensionmarginal}
Let $\TokFn(\w) = \w^{\TokExp}$ where $\TokExp > 1$. Fix $\NumConsumers > 1$, and fix $\marg < 1$. Suppose there are $\NumProviders > 1$ providers. Let $\NumConsumers' = \frac{\NumConsumers (1 + \marg)}{1 + \marg \cdot \NumConsumers}$. 
There exist thresholds $0 < \LowerThresholdMarg(C, \alpha, \beta, \marg) < \UpperThreshold(C, \alpha, \beta, \marg) < \infty$ such that the intermediary usage at equilibrium satisfies:
\[ 
\sum_{j=1}^{\NumConsumers} \mathbb{E}[1[\action{j} = \Middleman]] =
\begin{cases}
0 & \text{ if } \min(\ic + \costadditional, \costmanual) < \LowerThresholdMarg(\NumConsumers, \subfee, \TokExp, \marg) \\
\NumConsumers & \text{ if } \min(\ic + \costadditional, \costmanual) \in [\LowerThresholdMarg(\NumConsumers, \subfee, \TokExp, \marg), \UpperThresholdMarg(\NumConsumers, \subfee, \TokExp, \marg)]\\
0 & \text{ if } \min(\ic + \costadditional, \costmanual) > \UpperThresholdMarg(\NumConsumers, \subfee, \TokExp, \marg) \\
\end{cases}
\]
In fact, the thresholds are related to the thresholds from Theorem \ref{thm:specialcasemiddlemanusageexposure} as follows: $\LowerThresholdMarg(\NumConsumers, \subfee, \TokExp, \marg) = (1 + \marg)^{-1} \cdot \LowerThreshold(\NumConsumers', \subfee, \TokExp)$ and $\UpperThresholdMarg(\NumConsumers, \subfee, \TokExp, \marg) = (1 + \marg)^{-1} \cdot \UpperThreshold(\NumConsumers', \subfee, \TokExp)$.
\end{theorem}

Theorem \ref{thm:extensionmarginal} (Figure \ref{fig:marginalcosts}) shows that the insights from Theorem \ref{thm:specialcasemiddlemanusageexposure} generalize to the case where the intermediary faces marginal costs, albeit with the intermediation range reduced in width. Intuitively, Theorem \ref{thm:extensionmarginal} captures how the market with marginal costs behaves the same as a market without marginal costs but with a smaller ``effective'' number of consumers
$\NumConsumers' = \frac{\NumConsumers (1 + \marg)}{1 + \marg \NumConsumers} < \NumConsumers$ and also with a multiplicative reduction factor $(1+\marg)^{-1}$. The effective number of consumers, which is decreasing in marginal costs, captures the extent to which the intermediary still enjoys economies of scale given that they face marginal costs for distribution; the multiplicative reduction factor captures how consumers also face higher costs of production relative to our baseline model. Given our prior finding that the range of thresholds supporting the intermediary shrinks in width as the number of consumers decreases in our baseline model (Theorem \ref{thm:specialcasemiddlemanusageexposure}), this implies that marginal costs lead the intermediary to be supported on a more narrow range of technology levels.

\subsection{Other fee structures}\label{subsec:transfers}

A key assumption in our baseline model is that the intermediary has no control over the fee structure: 
their marginal fee is $\subfee$ regardless of production costs.  We verify the importance of this assumption by considering a model where the intermediary charges a fee that increases linearly with the quality of the content consumed. Specifically, instead of the consumer paying a fixed fee $\subfee$ to the intermediary, the consumer pays a \textit{linear fee} $\subfee \cdot \wmiddleman$, which scales with the quality $\wmiddleman$ of the content that the intermediary produces. In this setup, the value $\subfee$ captures the \textit{fee scaling} rather than the fee itself. The following result characterizes when disintermediation occurs. 
\begin{theorem}
\label{thm:extensionfees}
Let $\TokFn(\w) = \w^{\TokExp}$ where $\TokExp > 1$. Fix $\NumConsumers > 1$ and $\subfee \in (0,1)$. Suppose there are $\NumProviders > 1$ providers, and suppose that fees are linear. Then the intermediary usage $\sum_{j=1}^{\NumConsumers} \mathbb{E}[1[\action{j} = \Middleman]]$ is independent of the production cost parameters $\ic, \costadditional, \costmanual$. If the number of consumers $\NumConsumers$ is sufficiently high, the intermediary usage at equilibrium is $\sum_{j=1}^{\NumConsumers} \mathbb{E}[1[\action{j} = \Middleman]] = \NumConsumers$; if the number of consumers is sufficiently low, the intermediary usage is $\sum_{j=1}^{\NumConsumers} \mathbb{E}[1[\action{j} = \Middleman]] = 0$. 
\end{theorem}

Theorem \ref{thm:extensionmarginal} shows that linear fees fundamentally change the nature of disintermediation: the intermediary usage at equilibrium is \textit{independent} of the production costs $\price = \min(\ic + \costadditional, \costmanual)$.  This means that the intermediary always survives in the market when the number of consumers is sufficiently high, and never survives in the market when the number of consumers is sufficiently low. This result highlights the importance of our modeling assumption that the intermediary does not have full control over both the content quality and fee structure. We conclude that the this assumption---which is motivated by common practices in digital content recommendation (Example \ref{example:genai})---has a substantial impact on the market's sensitivity to technology improvements.\footnote{Specifically, digital content distribution platforms usually typically do not allow individual creators to individually set prices for their content. That being said, some platforms do reward creators based on engagement: when engagement is correlated with content quality (rather than other metrics, such as the length of videos), linear fees would capture online platforms that reward creators based on engagement rather than exposure.}

\section{Discussion}

In this work, we investigate the relationship between production technology improvements and disintermediation. We focus on markets where the technology is available to both the intermediary and consumers, and where the intermediary's strategic choice is restricted to the level of production quality. We find that reduced production costs eventually drive the intermediary out of the market entirely. We also show that even at production cost levels where the intermediary does survive, the threat of disintermediation leads to striking implications for welfare and content quality. While the intermediary is welfare-improving, the intermediary extracts all gains to social welfare for themselves. Furthermore, the intermediary's utility is inverse U-shaped in production costs, and the presence of the intermediary can raise or lower content quality. 

Our model and results open the door to several interesting avenues for future work. While our model focuses on a a single creator and a target audience of homogeneous consumers, it would be interesting to endogenize the audience-formation process and investigate differential impacts on different types of consumers and creators.  For instance, one could consider multiple intermediaries who can differentiate horizontally as well as vertically.  These could compete for heterogeneous consumers who might differ in horizontal taste, sensitivity to quality, and/or proclivity for niche content.  In such an environment, which types of creators face the heaviest threat of disintermediation?  Which types of consumers are better off in the world with new and improved production tools, and which (if any) are disadvantaged?  How does the \emph{type} of content available in the market vary as production technology improves?

\section{Acknowledgments}

We would like to thank Jon Kleinberg, Clayton Thomas, Rakesh Vohra, and Ruqing Xu for useful feedback. This research was supported in part by Microsoft's AI, Cognition, and the Economy (AICE) program, and by an Open Philanthropy AI fellowship. 

\newpage 

\bibliographystyle{plainnat}
\bibliography{ref.bib}

\appendix

\section{Useful lemmas}\label{appendix:lemmas}

\begin{lemma}
\label{lemma:min}
Suppose that $g$ is continuously differentiable, strictly convex, and satisfies $g(0) = g'(0) = 0$. Then it holds that:
\[\lim_{\w \rightarrow^+ 0} \frac{g(\w)}{g'(\w)} = 0. \]
\end{lemma}
\begin{proof}
Since $g(\w)$ and $g'(\w)$ are both positive for $\w > 0$, it suffices to show that:
\[\lim_{\w \rightarrow^+ 0} \frac{g(\w)}{g'(\w)} \le 0. \]
Using convexity, we know that:
\[ 0 = g(0) \ge g(w) + g'(w) (0 - w) = g(w) - w g'(w),\]
which means that $g(w) \le w \cdot g'(w)$. This means that:
\[\lim_{\w \rightarrow^+ 0} \frac{g(\w)}{g'(\w)} \le \lim_{\w \rightarrow^+ 0} w = 0 \]
as desired. 
\end{proof}

\subsection{Properties of $\max_{w \ge 0} (w - \price g(w)))$}
\begin{lemma}
\label{lemma:uniqueoptima}
Suppose that $g$ is continuously differentiable, strictly convex, satisfies $g(0) = g'(0) = 0$, and satisfies $\lim_{w \rightarrow \infty} g(w) = \lim_{w \rightarrow \infty} g'(w) = \infty$. For any $\price > 0$, then $\max_{w \ge 0} (w - \price g(w))$ has a unique optima, which is in $(0, \infty)$. 
\end{lemma}
\begin{proof}
Using that $g$ is convex, we see that $w^*$ is a maximum of $\max_{w \ge 0} (w - \price g(w))$ if and only if:
\[g'(w) = \frac{1}{\price}. \]
Using the other conditions on $g$, we see that this occurs at a unique value $w^* \in (0, \infty)$. 
\end{proof}

\begin{lemma}
\label{lemma:optimacomparativestatic}
Suppose that $g$ is continuously differentiable, strictly convex, satisfies $g(0) = g'(0) = 0$, and satisfies $\lim_{w \rightarrow \infty} g(w) = \lim_{w \rightarrow \infty} g'(w) = \infty$. Let $\w^*(\price)$ be an optima of $\max_{\w \ge 0} (w - \price \TokFn(\w))$. Then, it holds that:
\[\frac{\partial{\w^*(\price)}}{\partial{\price}} < 0. \]
Moreover, for any $\w \in (0, \infty)$, there exists a unique value $\price > 0$ such that $\w^*(\price) = \w$. 
\end{lemma}
\begin{proof}
By Lemma \ref{lemma:uniqueoptima}, we know that $\max_{w \ge 0} (w - \price g(w))$ has a unique optima, so $\w^*(\price)$ is uniquely defined. Using that $g$ is convex, we see that: 
\[\TokFn'(\w^*(\price)) = \frac{1}{\price}. \]
This, coupled with the other conditions on $\TokFn$, give us the desired result. 
\end{proof}

\begin{lemma}
\label{lemma:basicenvelope}
Suppose that $g$ is continuously differentiable, strictly convex, satisfies $g(0) = g'(0) = 0$, and satisfies $\lim_{w \rightarrow \infty} g(w) = \lim_{w \rightarrow \infty} g'(w) = \infty$. The derivative of $\max_{w \ge 0} (w - \price g(w)))$ with respect to $\price$ is equal to $-g(\argmax_{w \ge 0} (w - \price g(w)))$. 
\end{lemma}
\begin{proof}
We apply Lemma \ref{lemma:uniqueoptima} and let $w^*(\price)$ be the unique maximizer of $\max_{w \ge 0} (w - \price g(w))$. By the envelope theorem, we see that 
\[\frac{\partial{}}{\partial{\price}} \left(\max_{w \ge 0} (w - \price g(w)) \right) = -g(w^*),\]
as desired. 
\end{proof}

\subsection{Properties of $\price \cdot g(\alpha + \max_{w \ge 0} (w - \price g(w)))$}

\begin{lemma}
\label{lemma:derivative}
Suppose that $g$ is continuously differentiable, strictly convex, satisfies $g(0) = g'(0) = 0$, and satisfies $\lim_{w \rightarrow \infty} g(w) = \lim_{w \rightarrow \infty} g'(w) = \infty$. The derivative of $\price \cdot g(\alpha + \max_{w \ge 0} (w - \price g(w)))$ with respect to $\price$ is equal to:
\[g'\left(\alpha + \max_{w \ge 0} (w - \price g(w)) \right) \left( \frac{g\left(\alpha + \max_{w \ge 0} (w - \price g(w)) \right)}{g'\left(\alpha + \max_{w \ge 0} (w - \price g(w)) \right)} - \frac{g(\argmax_{w \ge 0} (w - \price g(w)))}{g'(\argmax_{w \ge 0} (w - \price g(w)))}\right).\] Moreover, the sign of the derivative is:
\[ 
\begin{cases}
0 &\text{ if }\;\;\;  \frac{g(w^*)}{g'(w^*)} = \alpha \\
    \text{positive} &\text{ if  } \;\;\; \frac{g(w^*)}{g'(w^*)} < \alpha \\
      \text{negative} &\text{ if  }\;\;\;  \frac{g(w^*)}{g'(w^*)} > \alpha ,    
\end{cases}
\]
where $w^*$ is the maximizer of  $\max_{w \ge 0} (w - \price g(w))) $. 
\end{lemma}
\begin{proof}
We apply Lemma \ref{lemma:uniqueoptima} and let $w^*$ be the unique maximizer of $\max_{w \ge 0} (w - \price g(w))) $ (note that this depends on $\price$). Using the first-order condition, we observe that:
\[g'(w^*) = \frac{1}{\price}. \]
By Lemma \ref{lemma:basicenvelope}, we observe that the derivative of  $\max_{w \ge 0} (w - \price g(w))) $  with respect to $\price$ is $g(w^*)$. 

Now, we can take a derivative of $\price \cdot g(\alpha + \max_{w \ge 0} (w - \price g(w)))$ to obtain:
\begin{align*}
 &\frac{\partial{}}{\partial{\price}} \left(\price \cdot g(\alpha + \max_{w \ge 0} (w - \price g(w)))\right)  \\
 &= g\left(\alpha + \max_{w \ge 0} (w - \price g(w))\right) - \price g'\left(\alpha + \max_{w \ge 0} (w - \price g(w))\right) g(w^*). \\
&=  g\left(\alpha + \max_{w \ge 0} (w - \price g(w))\right) - \frac{ g'\left(\alpha + \max_{w \ge 0} (w - \price g(w))\right) g(w^*)}{g'(w^*)} \\
&= g'\left(\alpha + \max_{w \ge 0} (w - \price g(w))\right) \left( \frac{g\left(\alpha + \max_{w \ge 0} (w - \price g(w))\right)}{g'\left(\alpha + \max_{w \ge 0} (w - \price g(w))\right)} - \frac{g(w^*)}{g'(w^*)}\right) \\
&= g'\left(\alpha + (w^* - \price g(w^*))\right) \left( \frac{g\left(\alpha + (w^* - \price g(w^*))\right)}{g'\left(\alpha + (w^* - \price g(w^*))\right)} - \frac{g(w^*)}{g'(w^*)}\right)
 \end{align*}  

Observe that the derivative has the same sign as 
\[\frac{g\left(\alpha + (w^* - \price g(w^*))\right)}{g'\left(\alpha + (w^* - \price g(w^*))\right)} - \frac{g(w^*)}{g'(w^*)}.\]
Since $g$ is strictly log-concave, we know that $\frac{g(w)}{g'(w)}$ is strictly increasing in $\w$. This means that the sign of the derivative is the same as the sign of 
\[\alpha + (w^* - \price g(w^*)) -  w^* = \alpha - \price \cdot g(w^*).\]
Using the first-order condition, this is equal to: 
\[\alpha - \frac{g(w^*)}{g'(w^*)}. \]
This proves the desired statement. 
\end{proof}

\begin{lemma}
\label{lemma:optima}
Consider the setup of Theorem \ref{thm:sufficientcondition}. Then, it holds that $\price \cdot g(\alpha + \max_{w \ge 0} (w - \price g(w)))$ is U-shaped in $\price$. Moreover, it holds that
\[\min_{\price \ge 0} \left(\price \cdot g(\alpha + \max_{w \ge 0} (w - \price g(w))) \right) = \alpha.\]
Furthermore, there is a unique global optimum $\price \in [0, \infty)$, and this value is the unique solution to 
\[\frac{g(\argmax(w - \nu g(w)))}{g'(\argmax(w - \nu g(w)))} = \alpha.\] 
\end{lemma}
\begin{proof}
We apply Lemma \ref{lemma:uniqueoptima} and let $w^*$ be the unique maximizer of $\max_{w \ge 0} (w - \price g(w))) $ (note that this depends on $\price$). We use Lemma \ref{lemma:derivative} to see that the sign of the derivative of $\price \cdot g(\alpha + \max_{w \ge 0} (w - \price g(w)))$ with respect to $\price$ is:
\[ 
\begin{cases}
0 &\text{ if }\;\;\;  \frac{g(w^*)}{g'(w^*)} = \alpha \\
    \text{positive} &\text{ if  } \;\;\; \frac{g(w^*)}{g'(w^*)} < \alpha \\
      \text{negative} &\text{ if  }\;\;\;  \frac{g(w^*)}{g'(w^*)} > \alpha ,    
\end{cases}
\]

We next show there is a unique value of $\price$ such that $\frac{g(w^*)}{g'(w^*)} = \alpha$. Since $g$ is strictly log-concave, we know that $\frac{g(w)}{g'(w)}$ is strictly increasing in $\w$. By Lemma \ref{lemma:min}, we know that $\lim_{w \rightarrow 0} \frac{g(w)}{g'(w)} = 0$ and by the assumed condition in the theorem statement, we know that:
\[ \lim_{w \rightarrow \infty} \frac{g(w)}{g'(w)} \ge \lim_{w \rightarrow \infty} \frac{g\left(w - \frac{g(w)}{g'(w)} \right)}{g'(w)} = \infty. \] Since $\frac{g(w)}{g'(w)}$ is continuous, this means that there exists $w > 0$ such that $\frac{g(w)}{g'(w)} = \alpha$. Using Lemma \ref{lemma:optimacomparativestatic}, this means that there exists a unique value of $\price > 0$ such that $w^* = w$. 

Next, we show that $\price \cdot g(\alpha + \max_{w \ge 0} (w - \price g(w))))$ is U-shaped with a unique global minimum. Let $w^*(\price')$ be the unique optimum of $\max_{w \ge 0} (w - \price' g(w))$ (Lemma \ref{lemma:uniqueoptima}). 
\begin{itemize}
    \item For prices $\price' < \price$ below this threshold, by Lemma \ref{lemma:optimacomparativestatic}, we observe that $w^*(\price') > w^*(\price)$. Using log-concavity of $g$, this means that:
\[\frac{g(w^*(\price'))}{g'(w^*(\price'))} > \frac{g(w^*)}{g'(w^*)}  = \alpha, \]
which means that the derivative  is negative. Applying this for every $\price' < \price$ means that 
\[\left(\price' \cdot g(\alpha + \max_{w \ge 0} (w - \price' g(w))) \right) > \alpha.\] 
\item For prices $\price' > \price$ below this threshold, by Lemma \ref{lemma:optimacomparativestatic}, we observe that $w^*(\price') < w^*(\price)$. Using log-concavity of $g$, this means that:
\[\frac{g(w^*(\price'))}{g'(w^*(\price'))} < \frac{g(w^*)}{g'(w^*)}  = \alpha, \]
which means that the derivative  is positive. Applying this for every $\price' < \price$ means that \[ \left(\price' \cdot g(\alpha + \max_{w \ge 0} (w - \price' g(w))) \right) > \alpha.\] 
\end{itemize}

Finally,  at the value of $\price$ such that $\frac{g(w^*)}{g'(w^*)} = \alpha$, it holds that:
\begin{align*}
 \price \cdot g(\alpha + \max_{w \ge 0} (w - \price g(w))) &= \price \cdot g(\alpha + w^* - \price g(w^*))) \\
 &= \price \cdot g(w^*) \\
 &= \frac{g(w^*)}{g'(w^*)} \\
 &= \alpha.
\end{align*}
\end{proof}

\section{Proofs for Section \ref{sec:model}}\label{appendix:proofsequilibria}

The main lemma is the following characterization of the equilibria in the subgame between the intermediary and consumers. 
\begin{lemma}
\label{lemma:middlemansubgame}
Suppose that suppliers choose prices $\price_1, \ldots, \price_P$ and consider the subgame between the intermediary and consumers (Stages 2-3). Under the tiebreaking assumptions discussed in Section \ref{subsec:equilibriumexistence}, there exists a unique pure strategy equilibrium in this subgame which takes the following form. Let $\price = \min(\costadditional + \min_{i \in [\NumProviders]} \price_i, \costmanual)$, and consider the condition 
\begin{equation}
\label{eq:condition}
 \price \cdot \TokFn\left(\subfee + \max_{w \ge 0} \left(\w -\price \TokFn(w)\right) \right) > \subfee \NumConsumers.   
\end{equation}
\begin{itemize}
    \item If \eqref{eq:condition} holds, then $\wmiddleman = 0$. Moreover,  for all $j \in [\NumConsumers]$, it holds that $\action{j} = \Direct$, $\wuser{j} = \argmax_{w \ge 0} \left(\w -\price \TokFn(w)\right)$. Moreover, if $\costmanual < \costadditional + \min_{i \in [\NumProviders]} \price_i$, then the consumer chooses $\providerchoice{j} = 0$. Otherwise, the consumer chooses $\providerchoice{j} = \text{argmin}_{i \in [\NumProviders]} \price_i$ (tie-breaking in favor of suppliers with a lower index). 
    \item If \eqref{eq:condition} does not hold, then $\wmiddleman = \wuser{j} = \subfee + \max_{w \ge 0} \left(\w -\price \TokFn(w)\right)$. Moreover, if $\costmanual < \costadditional + \min_{i \in [\NumProviders]} \price_i$, then the intermediary chooses $\providerchoicemiddleman = 0$; otherwise, the intermediary chooses $\providerchoicemiddleman = \text{argmin}_{i \in [\NumProviders]} \price_i$ (tie-breaking in favor of suppliers with a lower index). Finally, it holds that $\action{j} = \Middleman$ and $\wuser{j} = \wmiddleman$ for all $j \in [\NumConsumers]$. 
\end{itemize}
\end{lemma}
\begin{proof}
Recall that when consumers or the intermediary produce content, they choose the option that minimizes their production costs. If 
$\costadditional + \min_{i \in [\NumProviders]} \price_i < \costmanual$, they leverage the technology of the supplier who offers the lowest price, and otherwise, they produce content without using the technology. This means that they face production costs $\price = \min(\costadditional + \min_{i \in [\NumProviders]} \price_i, \costmanual)$. 

When consumer $j$ chooses $\action{j} = \Direct$, then they maximize their utility and thus produce content $w^*(\price) = \argmax(w - \price g(w))$ and achieve utility $\max(w - \price g(w))$. Since the consumer pays the intermediary a fee of $\subfee$, the intermediary must produce content satisfying $w' \ge \subfee + \max_{w \ge 0}(w - \price g(w))$ to incentivize the consumer to choose $\action{j} = \Middleman$. Producing content $w' \ge \subfee + \max_{w \ge 0}(w - \price g(w))$ would incentivize all of the consumers to choose the intermediary, so the intermediary would earn utility
\[ \subfee \cdot \NumConsumers - \price \cdot g(w').\]
This also means that the intermediary prefers producing content $\subfee + \max_{w \ge 0}(w - \price g(w))$ over any $w' > \subfee + \max_{w \ge 0}(w - \price g(w))$ in order to minimize costs. The intermediary prefers producing this content over producing content $w = 0$ which would not attract any consumers if and only if:
\[ \subfee \cdot \NumConsumers - \price \cdot g(\subfee + \max_{w \ge 0}(w - \price g(w))) \ge 0.\]
This, coupled with the tiebreaking rules, proves the desired statement.
\end{proof}

Using this lemma, we can characterize pure strategy equilibria in our game.  
\begin{lemma}
\label{lemma:equilibriumcharacterization} 
Under the tiebreaking assumptions in Section \ref{subsec:equilibriumexistence}, there exists a pure strategy equilibrium which takes the following form. All suppliers choose the price $\price_i = \ic$ for $i \in [\NumProviders]$, and the intermediary and consumers choose actions according to the subgame equilibrium constructed in Lemma \ref{lemma:middlemansubgame}.
\end{lemma}
\begin{proof}
If $\costmanual < \ic + \costadditional$, then by Lemma \ref{lemma:middlemansubgame}, then consumers and the intermediary produce content without the technology, so suppliers all have zero utility regardless of what price they choose. 

If $\costmanual \ge \ic + \costadditional$, then consumers and the intermediary choose manual production if $\min_{i \in [\NumProviders]} \price_i > \costmanual$ and otherwise choose supplier $\text{argmin}_{i \in [\NumProviders]} \price_i$. We show that $\price_i = \ic$ for $i \in [\NumProviders]$ is an equilibrium. At this equilibrium, note that all of the suppliers earn zero utility. If a supplier deviates to $\price_i < \ic$, then by Lemma \ref{lemma:middlemansubgame}, production would be done through the supplier which would result in negative utility. Deviating to $\price > \ic$ would result in zero utility. Thus, there are no profitable deviations for the suppliers. 
\end{proof}

Lemma \ref{lemma:equilibriumcharacterization} implies both Theorem \ref{thm:equilibriumconstruction} and Theorem \ref{thm:equilibriumuniqueness}.
\begin{proof}[Proof of Theorem \ref{thm:equilibriumconstruction}]
This follows from Lemma \ref{lemma:equilibriumcharacterization}.
\end{proof}

\begin{proof}[Proof of Theorem \ref{thm:equilibriumuniqueness}]
We show that the actions of the intermediary and consumers, as well as the production cost $\min(\costadditional + \min_{i \in [\NumProviders] \price_i}, \costmanual)$ is the same at every pure strategy equilibrium. To do this, we show that these values are the same as at the pure strategy equilibrium constructed in Lemma \ref{lemma:equilibriumcharacterization}. 

Suppose that $\costmanual < \ic + \costadditional$. If 
$\min_{i \in [\NumProviders] \price_i} > \costmanual$, then production is done without the technology, and the intermediary and consumers choose the same actions as in the equilibrium in Lemma \ref{lemma:equilibriumcharacterization}. If $\min_{i \in [\NumProviders] \price_i} \le \costmanual$, then by Lemma \ref{lemma:equilibriumcharacterization}, production would done through supplier 
$\text{argmin}_{i \in [\NumProviders] \price_i} \le \costmanual$, and that supplier would earn negative utility. This is not possible because the supplier could deviate to $\price_i = \ic$ and earn zero utility. 

Now, suppose that $\costmanual \ge \ic + \costadditional$. In this case, assume for sake of contradiction that $\min_{i \in [\NumProviders] \price_i} \neq \ic$. If $\min_{i \in [\NumProviders] \price_i} > \ic$, then using Lemma \ref{lemma:middlemansubgame}, a supplier could earn higher profit by choosing $\price = \min(\min_{i \in [\NumProviders] \price_i}, \costmanual) - \epsilon$ for sufficiently small $\epsilon$, which is a contradiction.  If $\min_{i \in [\NumProviders] \price_i} < \ic$, then using Lemma \ref{lemma:middlemansubgame}, the supplier $\text{argmin}_{i \in [\NumProviders] \price_i} < \ic$ with lowest index could earn higher utility by instead choosing $\price = \ic$, which is a contradiction. This means that  $\min_{i \in [\NumProviders] \price_i} = \ic$, so by Lemma \ref{lemma:middlemansubgame} the intermediary and the consumers take the same actions as in the equilibrium in Lemma \ref{lemma:equilibriumcharacterization}. 
\end{proof}




\section{Proofs for Section \ref{sec:disintermediation}}\label{appendix:proofsdintermediation}

\subsection{Analysis of specific cost function families}

We first analyze the derivative and log-derivative of several families of cost functions.
\begin{lemma}
\label{lemma:costderivatives}
The following statements hold:
\begin{enumerate}
    \item For $\TokFn(\w) = \w^{\TokExp}$ for $\TokExp > 1$, the derivative is $\TokFn'(\w) = \TokExp \cdot \w^{\TokExp-1}$, and the log-derivative is $\frac{\TokFn'(\w)}{\TokFn(\w)} = \frac{\beta}{\w}$ 
    \item For $\TokFn(\w) = \w^{\TokExp} \cdot e^{\sqrt{\w}}$ for $\TokExp \ge 1$, the derivative is \[\TokFn'(\w) = \TokExp \cdot \w^{\TokExp-1} \cdot e^{\sqrt{\w}} + 0.5 \cdot \w^{\TokExp - \frac{1}{2}} \cdot e^{\sqrt{\w}}\] and the log-derivative is:
    \[ \frac{\TokFn'(\w)}{\TokFn(\w)} = \frac{\TokExp + 0.5 \cdot \w^{\frac{1}{2}}}{\w} = \frac{\TokExp}{\w} + \frac{0.5}{w^{\frac{1}{2}}}.  \]
    \item For $\TokFn(\w) = \w^{\TokExp} \cdot (\log(\w+1)^{\TokExpSecond})$ for any $\TokExp, \TokExpSecond > 1$, the derivative is:
    \[\TokFn'(\w) = \TokExp \cdot \w^{\TokExp-1} \cdot (\log(\w+1)^{\TokExpSecond}) + \frac{\w^{\TokExp}}{\w + 1} \cdot \TokExpSecond (\log(\w+1)^{\TokExpSecond - 1})\]
    and the log-derivative is:
    \[ \frac{\TokFn'(\w)}{\TokFn(\w)} =  \frac{\TokExp \cdot \log(\w+1) + \frac{\w}{\w + 1} \cdot \TokExpSecond}{\w \cdot \log(\w+1)} = \frac{\TokExp}{\w} + \frac{\TokExpSecond}{(w+1)(\log(w+1))}\]
    \item For $\TokFn(\w) = \w^{\TokExp} \cdot e^{\w}$ for $\TokExp \ge 1$, the derivative is $\TokFn'(\w) = \TokExp \cdot \w^{\TokExp-1} \cdot e^{\w} + \w^{\TokExp} \cdot e^{\w}$ and the log-derivative is:
    \[ \frac{\TokFn'(\w)}{\TokFn(\w)} = 1 + \frac{\TokExp}{\w}. \]
    
\end{enumerate}
\end{lemma}
\begin{proof}
We analyze each family of cost functions separately.  

\paragraph{Family 1: $\TokFn(\w) = \w^{\TokExp}$ for $\TokExp > 1$.} The derivative is $\TokFn'(\w) = \TokExp \cdot \w^{\TokExp-1}$. The log derivative is:
\[ \frac{\TokFn'(\w)}{\TokFn(\w)} = \frac{\beta}{\w}. \]

\paragraph{Family 2: $\TokFn(\w) = \w^{\TokExp} \cdot e^{\sqrt{\w}}$.}  The derivative is 
\[\TokFn'(\w) = \TokExp \cdot \w^{\TokExp-1} \cdot e^{\sqrt{\w}} + \w^{\TokExp} \cdot e^{\sqrt{\w}} \cdot 0.5 \cdot \w^{-\frac{1}{2}} = \TokExp \cdot \w^{\TokExp-1} \cdot e^{\sqrt{\w}} + 0.5 \cdot \w^{\TokExp - \frac{1}{2}} \cdot e^{\sqrt{\w}}.\] The log-derivative is:
\[ \frac{\TokFn'(\w)}{\TokFn(\w)} = \frac{\TokExp \cdot \w^{\TokExp-1} \cdot e^{\sqrt{\w}} + 0.5 \cdot \w^{\TokExp - \frac{1}{2}} \cdot e^{\sqrt{\w}}}{\w^{\TokExp} \cdot e^{\sqrt{\w}}} = \frac{\TokExp + 0.5 \cdot \w^{\frac{1}{2}}}{\w} = \frac{\TokExp}{\w} + \frac{0.5}{w^{\frac{1}{2}}}.  \]

\paragraph{Family 3: $\TokFn(\w) = \w^{\TokExp} \cdot (\log(\w+1)^{\TokExpSecond})$ for any $\TokExp, \TokExpSecond > 1$.}  The derivative is 
\[\TokFn'(\w) = \TokExp \cdot \w^{\TokExp-1} \cdot (\log(\w+1)^{\TokExpSecond}) + \frac{\w^{\TokExp}}{\w + 1} \cdot \TokExpSecond (\log(\w+1)^{\TokExpSecond - 1}).\]
The log-derivative is: 
\begin{align*}
 \frac{\TokFn'(\w)}{\TokFn(\w)} &= \frac{\TokExp \cdot \w^{\TokExp-1} \cdot (\log(\w+1)^{\TokExpSecond}) + \frac{\w^{\TokExp}}{\w + 1} \cdot \TokExpSecond (\log(\w+1)^{\TokExpSecond - 1})}{\w^{\TokExp} \cdot (\log(\w+1)^{\TokExpSecond})}  \\
 &= \frac{\TokExp \cdot \log(\w+1) + \frac{\w}{\w + 1} \cdot \TokExpSecond}{\w \cdot \log(\w+1)} = \frac{\TokExp}{\w} + \frac{\TokExpSecond}{(w+1)(\log(w+1))}.
\end{align*}

\paragraph{Family 4: $\TokFn(\w) = \w^{\TokExp} \cdot e^{\w}$.} The derivative is 
\[\TokFn'(\w) = \TokExp \cdot \w^{\TokExp-1} \cdot e^{\w} + \w^{\TokExp} \cdot e^{\w}.\]
The log-derivative is:
\[ \frac{\TokFn'(\w)}{\TokFn(\w)} = \frac{\TokExp \cdot \w^{\TokExp-1} \cdot e^{\w} + \w^{\TokExp} \cdot e^{\w}}{\w^{\TokExp} \cdot e^{\w}} = \frac{\TokExp + \w}{\w} = 1 + \frac{\TokExp}{\w}. \]

\end{proof}

Using Lemma \ref{lemma:costderivatives}, we  prove that several families of cost functions satisfy the assumptions for Theorem \ref{thm:middlemanusageexposure}. 
\begin{proposition}
\label{prop:costs}
The following cost functions satisfy the assumptions of Theorem \ref{thm:middlemanusageexposure}: (1) $\TokFn(\w) = \w^{\TokExp}$ for $\TokExp > 1$, (2) $\TokFn(\w) = \w^{\TokExp} \cdot e^{\sqrt{\w}}$ for $\TokExp \ge 1$, (3) $\TokFn(\w) = \w^{\TokExp} \cdot (\log(\w+1)^{\TokExpSecond})$ for any $\TokExp, \TokExpSecond > 1$, and (4) $\TokFn(\w) = \w^{\TokExp} \cdot e^{\w}$ for $\TokExp \ge 1$. 
\end{proposition}
\begin{proof}
We analyze each family of cost functions separately. It suffices to prove that these functions are strictly increasing, continuously differentiable, strictly convex, satisfy $\TokFn(0) = \TokFn'(0) = 0$ and $\lim_{w \rightarrow \infty} \TokFn(\w) = \lim_{w \rightarrow \infty} \TokFn'(w) = \infty$, and strictly log-concave.

\paragraph{Family 1: $\TokFn(\w) = \w^{\TokExp}$ for $\TokExp > 1$.} By Lemma \ref{lemma:costderivatives}, the derivative is $\TokFn'(\w) = \TokExp \cdot \w^{\TokExp-1}$. This means that $\TokFn$ is strictly increasing and continuously differentiable. Moreover, the derivative is increasing, so the function is strictly convex. We also see that $\TokFn(0) = \TokFn'(0) = 0$ and $\lim_{w \rightarrow \infty} \TokFn(\w) = \lim_{w \rightarrow \infty} \TokFn'(w) = \infty$. To show that $\TokFn$ is strictly log-concave, it suffices to show that the log-derivative is strictly decreasing. By Lemma \ref{lemma:costderivatives}, the log-derivative is
\[ \frac{\TokFn'(\w)}{\TokFn(\w)} = \frac{\beta}{\w},\]
which is strictly decreasing as desired.

\paragraph{Family 2: $\TokFn(\w) = \w^{\TokExp} \cdot e^{\sqrt{\w}}$.} By Lemma \ref{lemma:costderivatives}, the derivative is 
\[\TokFn'(\w) = \TokExp \cdot \w^{\TokExp-1} \cdot e^{\sqrt{\w}} + 0.5 \cdot \w^{\TokExp - \frac{1}{2}} \cdot e^{\sqrt{\w}}.\] This means that $\TokFn$ is strictly increasing and continuously differentiable. Moreover, the derivative is increasing, so the function is strictly convex. We also see that $\TokFn(0) = \TokFn'(0) = 0$ and $\lim_{w \rightarrow \infty} \TokFn(\w) = \lim_{w \rightarrow \infty} \TokFn'(w) = \infty$. To show that $\TokFn$ is strictly log-concave, it suffices to show that the log-derivative is strictly decreasing. By Lemma \ref{lemma:costderivatives}, the log-derivative is
\[ \frac{\TokFn'(\w)}{\TokFn(\w)} = \frac{\TokExp + 0.5 \cdot \w^{\frac{1}{2}}}{\w} = \frac{\TokExp}{\w} + \frac{0.5}{w^{\frac{1}{2}}}  \]
which is strictly decreasing as desired. 

\paragraph{Family 3: $\TokFn(\w) = \w^{\TokExp} \cdot (\log(\w+1)^{\TokExpSecond})$ for any $\TokExp, \TokExpSecond > 1$.}  By Lemma \ref{lemma:costderivatives}, the derivative is

    \[\TokFn'(\w) = \TokExp \cdot \w^{\TokExp-1} \cdot (\log(\w+1)^{\TokExpSecond}) + \frac{\w^{\TokExp}}{\w + 1} \cdot \TokExpSecond (\log(\w+1)^{\TokExpSecond - 1}).\]
This means that $\TokFn$ is strictly increasing and continuously differentiable. Since $\frac{\w^{\TokExp}}{\w + 1}$ is increasing, the derivative is increasing, so the function is strictly convex. We also see that $\TokFn(0) = \TokFn'(0) = 0$ and $\lim_{w \rightarrow \infty} \TokFn(\w) = \lim_{w \rightarrow \infty} \TokFn'(w) = \infty$. To show that $\TokFn$ is strictly log-concave, it suffices to show that the log-derivative is strictly decreasing. By Lemma \ref{lemma:costderivatives}, the log-derivative is
 and the log-derivative is:
    \[ \frac{\TokFn'(\w)}{\TokFn(\w)} =  \frac{\TokExp \cdot \log(\w+1) + \frac{\w}{\w + 1} \cdot \TokExpSecond}{\w \cdot \log(\w+1)} = \frac{\TokExp}{\w} + \frac{\TokExpSecond}{(w+1)(\log(w+1))}\]
which is strictly decreasing as desired.

\paragraph{Family 4: $\TokFn(\w) = \w^{\TokExp} \cdot e^{\w}$.} 
 By Lemma \ref{lemma:costderivatives}, the derivative is 
\[\TokFn'(\w) = \TokExp \cdot \w^{\TokExp-1} \cdot e^{\w} + \w^{\TokExp} \cdot e^{\w}.\] This means that $\TokFn$ is strictly increasing and continuously differentiable. Moreover, the derivative is increasing, so the function is strictly convex. We also see that $\TokFn(0) = \TokFn'(0) = 0$ and $\lim_{w \rightarrow \infty} \TokFn(\w) = \lim_{w \rightarrow \infty} \TokFn'(w) = \infty$. To show that $\TokFn$ is strictly log-concave, it suffices to show that the log-derivative is strictly decreasing. By Lemma \ref{lemma:costderivatives}, the log-derivative is
\[ \frac{\TokFn'(\w)}{\TokFn(\w)} = \frac{\TokExp + \w}{\w} = 1 + \frac{\TokExp}{\w},   \]
which is strictly decreasing as desired. 
\end{proof}

We next identify several cost functions which satisfy the assumptions for Theorem \ref{thm:sufficientcondition}.

\begin{proposition}
\label{prop:costsstronger}
The following cost functions satisfy the assumptions of Theorem \ref{thm:sufficientcondition}: (1) $\TokFn(\w) = \w^{\TokExp}$ for $\TokExp > 1$, (2) $\TokFn(\w) = \w^{\TokExp} \cdot e^{\sqrt{\w}}$ for $\TokExp \ge 1$, and (3) $\TokFn(\w) = \w^{\TokExp} \cdot (\log(\w+1)^{\TokExpSecond})$ for any $\TokExp, \TokExpSecond > 1$. 
\end{proposition}
\begin{proof}
By Proposition \ref{prop:costs}, all three of these families satisfy the assumptions of Theorem \ref{thm:middlemanusageexposure}. Thus, it suffices to show that 
\[\lim_{\w \rightarrow \infty} \frac{\TokFn\left(\w - \frac{\TokFn(\w)}{\TokFn'(\w)} \right)}{\TokFn'(\w)} = \infty.\]

\paragraph{Family 1: $\TokFn(\w) = \w^{\TokExp}$ for $\TokExp > 1$.} By Lemma \ref{lemma:costderivatives}, the derivative is $\TokFn'(\w) = \TokExp \cdot \w^{\TokExp-1}$ and the log-derivative is
\[ \frac{\TokFn'(\w)}{\TokFn(\w)} = \frac{\beta}{\w}.\]
This means that:
\begin{align*}
  \lim_{\w \rightarrow \infty} \frac{\TokFn\left(\w - \frac{\TokFn(\w)}{\TokFn'(\w)} \right)}{\TokFn'(\w)}   &= \lim_{\w \rightarrow \infty} \frac{\left( \w - \frac{\w}{\beta} \right)^{\TokExp}}{\TokExp \cdot w^{\TokExp-1}} 
  &= (1 - \frac{1}{\beta})^\beta \cdot \lim_{\w \rightarrow \infty} \frac{w^{\TokExp}}{\TokExp \cdot w^{\TokExp-1}} 
  &= (1 - \frac{1}{\beta})^{\beta} \cdot \lim_{\w \rightarrow \infty} \frac{w}{\TokExp}
  &= \infty,
\end{align*}
as desired.

\paragraph{Family 2: $\TokFn(\w) = \w^{\TokExp} \cdot e^{\sqrt{\w}}$.} By Lemma \ref{lemma:costderivatives}, the derivative is 
\[\TokFn'(\w) = \TokExp \cdot \w^{\TokExp-1} \cdot e^{\sqrt{\w}} + 0.5 \cdot \w^{\TokExp - \frac{1}{2}} \cdot e^{\sqrt{\w}}\] and the log-derivative is
\[ \frac{\TokFn'(\w)}{\TokFn(\w)} = \frac{\TokExp + 0.5 \cdot \w^{\frac{1}{2}}}{\w} = \frac{\TokExp}{\w} + \frac{0.5}{w^{\frac{1}{2}}}.  \]
This means that:
\[ \frac{\TokFn(\w)}{\TokFn'(\w)} = \frac{1}{\frac{\TokExp}{\w} + \frac{0.5}{w^{\frac{1}{2}}}} =  \frac{\w}{\beta + 0.5 \sqrt{w}}\]
This means that:
\begin{align*}
  &\lim_{\w \rightarrow \infty} \frac{\TokFn\left(\w - \frac{\TokFn(\w)}{\TokFn'(\w)} \right)}{\TokFn'(\w)}  \\
  &= 
   \lim_{\w \rightarrow \infty} \frac{\TokFn\left(\w \right)}{\TokFn'(\w)} \cdot \frac{\TokFn\left(\w - \frac{\TokFn(\w)}{\TokFn'(\w)} \right)}{\TokFn(\w)}
   \\ 
   &=  \lim_{\w \rightarrow \infty}  \frac{\w}{\beta + 0.5 \sqrt{w}}\cdot \frac{\left(\w - \frac{\TokFn(\w)}{\TokFn'(\w)} \right)^{\TokExp} \cdot e^{\sqrt{\w - \frac{\TokFn(\w)}{\TokFn'(\w)}}}}{w^{\beta} e^{\sqrt{\w}}} \\
   &= \lim_{\w \rightarrow \infty}  \frac{\w}{\beta + 0.5 \sqrt{w}} \cdot \left(\frac{\w - \frac{\TokFn(\w)}{\TokFn'(\w)}}{\w}\right)^{\TokExp} \cdot e^{\sqrt{\w - \frac{\TokFn(\w)}{\TokFn'(\w)}} - \sqrt{\w}} \\
   &= \lim_{\w \rightarrow \infty}  \frac{\w}{\beta + 0.5 \sqrt{w}} \cdot \left(1 - \frac{1}{\beta + 0.5 \sqrt{\w}} \right)^{\TokExp} \cdot e^{\left(\sqrt{w - \frac{w}{\beta + 0.5 \sqrt{\w}}} - \sqrt{w} \right)}.
\end{align*}
Since 
$\lim_{\w \rightarrow \infty}  \frac{\w}{\beta + 0.5 \sqrt{w}} = \infty$ and $\lim_{\w \rightarrow \infty} \left(1 - \frac{1}{\beta + 0.5 \sqrt{\w}} \right)^{\TokExp} = 1$, it suffices to show that 
\[\lim_{\w \rightarrow \infty} e^{\left(\sqrt{w - \frac{w}{\beta + 0.5 \sqrt{\w}}} - \sqrt{w} \right)} = e^{-1}.\] It suffices to show that
\[\lim_{\w \rightarrow \infty} \left(\sqrt{w - \frac{w}{\beta + 0.5 \sqrt{\w}}} - \sqrt{w} \right) = -1, \]
which can be rewritten as:
\[\lim_{\w \rightarrow \infty}\frac{\sqrt{1 - \frac{1}{\beta + 0.5 \sqrt{\w}}} - 1}{w^{-1/2}} = -1. \] 
Using L'Hôpital's rule, we see that this is equal to:
\begin{align*}
\lim_{\w \rightarrow \infty}\frac{\sqrt{1 - \frac{1}{\beta + 0.5 \sqrt{\w}}} - 1}{w^{-1/2}}  &= \lim_{\w \rightarrow \infty} \frac{\frac{0.125}{\sqrt{\w} \sqrt{1 - \frac{1}{c + 0.5 \sqrt{w}}} \cdot (\beta + 0.5 \sqrt{\w})^2}}{-0.5 \cdot w^{-3/2}} \\
&= \lim_{\w \rightarrow \infty} -\frac{0.25 w}{\sqrt{1 - \frac{1}{c + 0.5 \sqrt{w}}} \cdot (\beta + 0.5 \sqrt{\w})^2} \\
&= \lim_{\w \rightarrow \infty} -\frac{0.25 w}{ \beta^2 + 0.25 \w + \beta \sqrt{\w}} \\
&= -1.
\end{align*}

\paragraph{Family 3: $\TokFn(\w) = \w^{\TokExp} \cdot (\log(\w+1)^{\TokExpSecond})$ for any $\TokExp, \TokExpSecond > 1$.}  By Lemma \ref{lemma:costderivatives}, the derivative is
 \[\TokFn'(\w) = \TokExp \cdot \w^{\TokExp-1} \cdot (\log(\w+1)^{\TokExpSecond}) + \frac{\w^{\TokExp}}{\w + 1} \cdot \TokExpSecond (\log(\w+1)^{\TokExpSecond - 1})\] and the log-derivative is:
    \[ \frac{\TokFn'(\w)}{\TokFn(\w)} =  \frac{\TokExp \cdot \log(\w+1) + \frac{\w}{\w + 1} \cdot \TokExpSecond}{\w \cdot \log(\w+1)} = \frac{\TokExp}{\w} + \frac{\TokExpSecond}{(w+1)(\log(w+1))} = \frac{1}{\w} \cdot \left(\TokExp + \frac{\TokExpSecond \cdot w}{(w+1)(\log(w+1))} \right).\]
    This means that:
    \[ \frac{\TokFn(\w)}{\TokFn'(\w)} =  \frac{w}{\TokExp + \frac{\TokExpSecond \cdot w}{(w+1)(\log(w+1))}}.\] 
This means that:
\begin{align*}
  &\lim_{\w \rightarrow \infty} \frac{\TokFn\left(\w - \frac{\TokFn(\w)}{\TokFn'(\w)} \right)}{\TokFn'(\w)}  \\
  &= 
   \lim_{\w \rightarrow \infty} \frac{\TokFn\left(\w \right)}{\TokFn'(\w)} \cdot \frac{\TokFn\left(\w - \frac{\TokFn(\w)}{\TokFn'(\w)} \right)}{\TokFn(\w)}
   \\ 
   &= \lim_{\w \rightarrow \infty} \frac{w}{\TokExp + \frac{\TokExpSecond \cdot w}{(w+1)(\log(w+1))}} \cdot \frac{g\left(w - \frac{w}{\TokExp + \frac{\TokExpSecond \cdot w}{(w+1)(\log(w+1))}} \right)}{g(w)} \\
 &= \lim_{\w \rightarrow \infty} \frac{w}{\TokExp + \frac{\TokExpSecond \cdot w}{(w+1)(\log(w+1))}} \cdot \left(\frac{w - \frac{w}{\TokExp + \frac{\TokExpSecond \cdot w}{(w+1)(\log(w+1))}} }{w}\right)^{\TokExp} \cdot \left(\frac{\log\left(1 + w - \frac{w}{\TokExp + \frac{\TokExpSecond \cdot w}{(w+1)(\log(w+1))}} \right)}{\log(1 + w)}\right)^{\TokExpSecond}.
\end{align*}
We analyze each term separately. Note that:
\[\lim_{\w \rightarrow \infty} \frac{w}{\TokExp + \frac{\TokExpSecond \cdot w}{(w+1)(\log(w+1))}}  = \infty \]
and 
\[\lim_{\w \rightarrow \infty} \left(\frac{w - \frac{w}{\TokExp + \frac{\TokExpSecond \cdot w}{(w+1)(\log(w+1))}} }{w}\right)^{\TokExp} = \lim_{\w \rightarrow \infty} \left(1 - \frac{1}{\TokExp + \frac{\TokExpSecond \cdot w}{(w+1)(\log(w+1))}}\right)^{\TokExp} = \lim_{\w \rightarrow \infty} \left(1 - \frac{1}{\TokExp}\right)^{\TokExp} \]
and 
\begin{align*}
    \lim_{\w \rightarrow \infty}  \left(\frac{\log\left(1 + w - \frac{w}{\TokExp + \frac{\TokExpSecond \cdot w}{(w+1)(\log(w+1))}} \right)}{\log(1 + w)}\right)^{\TokExpSecond} &\ge  \lim_{\w \rightarrow \infty}  \left(\frac{\log\left(1 + w - \frac{w}{\TokExp} \right)}{\log(1 + w)}\right)^{\TokExpSecond} \\
    &\ge \lim_{\w \rightarrow \infty}  \left(\frac{\log\left(1 - \frac{1}{\beta} + w - \frac{w}{\TokExp} \right)}{\log(1 + w)}\right)^{\TokExpSecond} \\
     &\ge \lim_{\w \rightarrow \infty}  \left(\frac{\log\left(1 - \frac{1}{\beta}\right) + \log\left(1 + w \right)}{\log(1 + w)}\right)^{\TokExpSecond} \\
     &= 1.
\end{align*}
This proves the desired statement.
\end{proof}

\subsection{Proof of Theorem \ref{thm:specialcasemiddlemanusageexposure}}

We prove Theorem \ref{thm:specialcasemiddlemanusageexposure} as a corollary of Theorem \ref{thm:sufficientcondition}.
\begin{proof}
By Proposition \ref{prop:costsstronger}, we know that $g(\w) = w^{\TokExp}$ for $\TokExp > 1$ satisfies the conditions of Theorem \ref{thm:sufficientcondition}. This implies the existence of thresholds $0< \LowerThreshold(\NumConsumers, \subfee, \TokExp) < \UpperThreshold(\NumConsumers, \subfee, \TokExp) < \infty$ such that the intermediary usage satisfies
\[ 
\sum_{j=1}^{\NumConsumers} \mathbb{E}[1[\action{j} = \Middleman]] =
\begin{cases}
0 & \text{ if } \min(\ic + \costadditional, \costmanual) < \LowerThreshold(\NumConsumers, \subfee, \TokExp) \\
\NumConsumers & \text{ if } \min(\ic + \costadditional, \costmanual) \in [\LowerThreshold(\NumConsumers, \subfee, \TokExp), \UpperThreshold(\NumConsumers, \subfee, \TokExp)]\\
0 & \text{ if } \min(\ic + \costadditional, \costmanual) > \UpperThreshold(\NumConsumers, \subfee, \TokExp) \\
\end{cases},
\]
where
$\LowerThreshold(\NumConsumers, \subfee, \TokExp)$ and $\UpperThreshold(\NumConsumers, \subfee, \TokExp)$ are the two unique solutions to 
\[\price \cdot \TokFn(\subfee + \max_{w} (w - \price g(w)) - \subfee \NumConsumers = 0.\] 
Using the structure of $g(\w) = w^{\TokExp}$, we see that 
\[\max_{w} (w - \price g(w)) = \price^{-\frac{1}{\beta - 1}} \left(\beta^{-\frac{1}{\beta-1}} - \beta^{-\frac{\beta}{\beta-1}} \right).\]
When we plug this into the above expression, we obtain:
\[\price \cdot \left(\subfee + \price^{-\frac{1}{\beta - 1}} \left(\beta^{-\frac{1}{\beta-1}} - \beta^{-\frac{\beta}{\beta-1}} \right) \right)^{\beta} - \subfee \NumConsumers = 0. \]
This can be rewritten as:
\[\price^{\frac{1}{\beta}} \cdot \left(\subfee + \price^{-\frac{1}{\beta - 1}} \left(\beta^{-\frac{1}{\beta-1}} - \beta^{-\frac{\beta}{\beta-1}} \right) \right) - \subfee^{\frac{1}{\beta}} \NumConsumers^{\frac{1}{\beta}} = 0. \]
This can be rewritten as:
\[\price^{\frac{1}{\beta}} \subfee^{\frac{\beta-1}{\beta}} + \price^{-\frac{1}{\beta(\beta - 1)}} \left(\beta^{-\frac{1}{\beta-1}} - \beta^{-\frac{\beta}{\beta-1}} \right) \subfee^{-\frac{1}{\beta}}   -  \NumConsumers^{\frac{1}{\beta}} = 0,\]
as desired. 
\end{proof}

\subsection{Proof of Theorem \ref{thm:middlemanusageexposure}}

First, we prove the following lemma which characterizes the sign of the derivative of $\price \cdot g(\alpha + \max_{w \ge 0} (w - \price g(w)))$.
\begin{lemma}
\label{lemma:signswitch}
Consider the setup of Theorem \ref{thm:middlemanusageexposure}.
Then, there exist a (possibly infinite or negative) threshold $\price^T$ such that the sign of the derivative of $\price \cdot g(\alpha + \max_{w \ge 0} (w - \price g(w)))$ satisfies:
\[ 
\begin{cases}
0 &\text{ if }\;\;\;  \price = \price^T \\
    \text{positive} &\text{ if  } \;\;\; \price > \price^T \\
      \text{negative} &\text{ if  }\;\;\; \price < \price^T.
\end{cases}
\]
\end{lemma}
\begin{proof}
We apply Lemma \ref{lemma:uniqueoptima} and let $w^*(\price)$ be the unique maximizer of $\max_{w \ge 0} (w - \price g(w))) $. We use Lemma \ref{lemma:derivative} to see that the sign of the derivative of $\price \cdot g(\alpha + \max_{w \ge 0} (w - \price g(w)))$ with respect to $\price$ is:
\[ 
\begin{cases}
0 &\text{ if }\;\;\;  \frac{g(w^*(\price))}{g'(w^*(\price))} = \alpha \\
    \text{positive} &\text{ if  } \;\;\; \frac{g(w^*(\price))}{g'(w^*(\price))} < \alpha \\
      \text{negative} &\text{ if  }\;\;\;  \frac{g(w^*(\price))}{g'(w^*(\price))} > \alpha ,    
\end{cases}
\]

Since $g$ is strictly log-concave, we know that $\frac{g(w)}{g'(w)}$ is strictly increasing in $\w$. Using Lemma \ref{lemma:optimacomparativestatic}, we see that $\frac{g(w^*(\price))}{g'(w^*(\price))}$ is strictly decreasing in $\price$. This guarantees that there exists a (possibly infinite or negative) threshold $\price^T$ such that the sign of the derivative is:
\[ 
\begin{cases}
0 &\text{ if }\;\;\;  \price = \price^T \\
    \text{positive} &\text{ if  } \;\;\; \price > \price^T \\
      \text{negative} &\text{ if  }\;\;\; \price < \price^T
\end{cases}
\]
as desired. 
\end{proof}

We now prove Theorem \ref{thm:middlemanusageexposure}.
\begin{proof}[Proof of Theorem \ref{thm:middlemanusageexposure}]
By Lemma \ref{lemma:equilibriumcharacterization}, disintermediation occurs if and only if $\price \cdot g(\alpha + \max_{w \ge 0} (w - \price g(w))) > \subfee \NumConsumers$. By Lemma \ref{lemma:signswitch}, there exist a (possibly infinite or negative) threshold $\price^T$ such that the sign of the derivative of $\price \cdot g(\alpha + \max_{w \ge 0} (w - \price g(w)))$ satisfies:
\[ 
\begin{cases}
0 &\text{ if }\;\;\;  \price = \price^T \\
    \text{positive} &\text{ if  } \;\;\; \price > \price^T \\
      \text{negative} &\text{ if  }\;\;\; \price < \price^T
\end{cases}
\]
This implies that there exist (possibly infinite or negative) thresholds $\LowerThreshold( \NumConsumers, \subfee, \TokFn)$ and $\UpperThreshold( \NumConsumers, \subfee, \TokFn)$ such that 
\[ 
\sum_{j=1}^{\NumConsumers} \mathbb{E}[1[\action{j} = \Middleman]] =
\begin{cases}
0 & \text{ if } \min(\ic + \costadditional, \costmanual) < \LowerThreshold(\NumConsumers, \subfee, \TokExp) \\
\NumConsumers & \text{ if } \min(\ic + \costadditional, \costmanual) \in [\LowerThreshold(\NumConsumers, \subfee, \TokExp), \UpperThreshold(\NumConsumers, \subfee, \TokExp)]\\
0 & \text{ if } \min(\ic + \costadditional, \costmanual) > \UpperThreshold(\NumConsumers, \subfee, \TokExp) \\
\end{cases}
\]
as desired. 
\end{proof}

\subsection{Proof of Theorem \ref{thm:sufficientcondition}}

First, we prove the following lemma that shows that the disintermediation boundary has exactly two solutions. 
\begin{lemma}
\label{lemma:twosolutions}
Consider the setup of Theorem \ref{thm:sufficientcondition}.
The equation 
\[ \price \cdot \TokFn(\subfee + \max_{w \ge 0}(w - \price \TokFn(w))) = \subfee \NumConsumers \]
has exactly two solutions. 
\end{lemma}
\begin{proof}
We apply Lemma \ref{lemma:uniqueoptima} and let $w^*(\price)$ be the unique maximizer of $\max_{w \ge 0} (w - \price g(w))) $. We use Lemma \ref{lemma:optima} to see that $\price \cdot \TokFn(\subfee + \max_{w \ge 0}(w - \price \TokFn(w)))$ is U-shaped and has global minimum 
\[\min_{\price > 0} (\price \cdot \TokFn(\subfee + \max_{w \ge 0}(w - \price \TokFn(w)))) = \alpha < \alpha \NumConsumers.\]
To show that $\price \cdot \TokFn(\subfee + \max_{w \ge 0}(w - \price \TokFn(w))) = \subfee \NumConsumers$ has exactly two solutions, it suffices to show that 
\[\lim_{\price \rightarrow \infty} (\price \cdot \TokFn(\subfee + \max_{w \ge 0}(w - \price \TokFn(w)))) = \infty = \lim_{\price \rightarrow 0} (\price \cdot \TokFn(\subfee + \max_{w \ge 0}(w - \price \TokFn(w)))).\]

First, we take a limit as $\nu \rightarrow \infty$. Observe that:
\[\lim_{\price \rightarrow \infty} (\price \cdot \TokFn(\subfee + \max_{w \ge 0}(w - \price \TokFn(w)))) \ge \lim_{\price \rightarrow \infty} (\price \cdot \TokFn(\subfee)) = \infty.  \]

Next, we take a limit as $\price \rightarrow 0$. Observe that 
\[\lim_{\price \rightarrow 0} (\price \cdot \TokFn(\subfee + \max_{w \ge 0}(w - \price \TokFn(w)))) \ge \lim_{\price \rightarrow 0} (\price \cdot \TokFn(\max_{w \ge 0}(w - \price \TokFn(w)))). \]
Using the first-order condition for $\max_{w \ge 0}(w - \price \TokFn(w))$, we see this is equal to:
\[ \lim_{\price \rightarrow 0} (\price \cdot \TokFn(\max_{w \ge 0}(w - \price \TokFn(w)))) = \lim_{\price \rightarrow 0} \frac{\TokFn(w^*(\price) - \price \cdot g(w^*(\price)))}{\TokFn'(w^*(\price))} = \lim_{\price \rightarrow 0} \frac{\TokFn(w^*(\price) - \frac{g(w^*(\price)))}{g'(w^*(\price)))}}{\TokFn'(w^*(\price))}. \]
Using Lemma \ref{lemma:optimacomparativestatic}, we can reparameterize and see that this is equal to: 
\[ \lim_{w \rightarrow \infty} \frac{\TokFn\left(w - \frac{g(w)}{g'(w)}\right)}{\TokFn'(w)}. \]
This is equal to $\infty$ by the assumption in the theorem statement. 
\end{proof}

Now we prove Theorem \ref{thm:sufficientcondition}
\begin{proof}[Proof of Theorem \ref{thm:sufficientcondition}]
By Lemma \ref{lemma:equilibriumcharacterization}, disintermediation occurs if and only if $\price \cdot g(\subfee + \max_{w \ge 0}(w - \price g(\w)) > \subfee C$. By Lemma \ref{lemma:optima}, we know that $\price \cdot g(\subfee + \max_{w \ge 0}(w - \price g(\w))$ is U-shaped and by Lemma \ref{lemma:twosolutions} we know that $\price \cdot g(\subfee + \max_{w \ge 0}(w - \price g(\w)) = \subfee C$ has exactly two solutions. This means that $0 < \LowerThreshold(\NumConsumers, \subfee, \TokFn)< \UpperThreshold(\NumConsumers, \subfee, \TokFn) < \infty$ can be taken to be equal to these two solutions. This also means that the lower threshold is decreasing as a function of the number of consumers $\NumConsumers$, and the upper threshold is increasing as a function of $\NumConsumers$, as desired.     
\end{proof}

\subsection{Proof of Proposition \ref{prop:counterexample}}

We prove Proposition \ref{prop:counterexample}.
\begin{proof}
Using Lemma \ref{lemma:equilibriumcharacterization}, it suffices to show that 
\[\lim_{\price \rightarrow 0} (\price \cdot g(\subfee + \max_{w \ge 0} (w - \price g(\w))) < \alpha C.\]
Using the first-order condition for $\max_{w \ge 0}(w - \price \TokFn(w))$, we see this is equal to:
\[ \lim_{\price \rightarrow 0} (\price \cdot \TokFn(\subfee + \max_{w \ge 0}(w - \price \TokFn(w)))) = \lim_{\price \rightarrow 0} \frac{\TokFn(\subfee + w^*(\price) - \price \cdot g(w^*(\price)))}{\TokFn'(w^*(\price))} = \lim_{\price \rightarrow 0} \frac{\TokFn(\subfee +  w^*(\price) - \frac{g(w^*(\price)))}{g'(w^*(\price)))}}{\TokFn'(w^*(\price))}. \]
Using Lemma \ref{lemma:optimacomparativestatic}, we can reparameterize and see that this is equal to: 
\[ \lim_{w \rightarrow \infty} \frac{\TokFn\left(\subfee +  w - \frac{g(w)}{g'(w)}\right)}{\TokFn'(w)}. \]
Using Lemma \ref{lemma:costderivatives}, we see that this is equal to:
\begin{align*}
 \lim_{w \rightarrow \infty} \frac{\TokFn\left(\subfee +  w - \frac{g(w)}{g'(w)}\right)}{\TokFn'(w)} &=  \lim_{w \rightarrow \infty} \frac{\TokFn\left(\subfee + w - \frac{g(w)}{g'(w)}\right)}{\TokFn\left(w\right)} \frac{\TokFn\left(w\right)}{\TokFn'(w)} \\
 &= \lim_{w \rightarrow \infty} \frac{\TokFn\left(\subfee + w - \frac{1}{1 + \frac{\beta}{\w}} \right)}{\TokFn\left(w\right)} \frac{1}{1 + \frac{\beta}{\w}} \\
 &\le \lim_{w \rightarrow \infty} \frac{\TokFn\left(\subfee + w \right)} {\TokFn\left(w\right)} \\
 &= e^{\subfee} \cdot \lim_{w\rightarrow \infty} \frac{(\subfee + \w)^{\beta}}{w^{\beta}} \\
 &= e^{\subfee}.
\end{align*}
Since $e^{\subfee } < \subfee \cdot \NumConsumers$, this proves the desired statement.
\end{proof}

\section{Proofs for Section \ref{sec:utility}}\label{appendix:proofsutility}

\subsection{Proof of Proposition \ref{prop:contentquality}}

We prove Proposition \ref{prop:contentquality}. This result follows easily from Lemma \ref{lemma:equilibriumcharacterization}. 
\begin{proof}[Proof of Proposition \ref{prop:contentquality}]
We split into two cases: $\price < \LowerThreshold(\NumConsumers, \subfee, \TokFn)$ or $\price > \UpperThreshold(\NumConsumers, \subfee, \TokFn)$, and $\price \in [\LowerThreshold(\NumConsumers, \subfee, \TokFn), \UpperThreshold(\NumConsumers, \subfee, \TokFn)]$.

\paragraph{Case 1: $\price < \LowerThreshold(\NumConsumers, \subfee, \TokFn)$ or $\price > \UpperThreshold(\NumConsumers, \subfee, \TokFn)$.} In this case, disintermediation occurs (Theorem \ref{thm:middlemanusageexposure}). By Lemma \ref{lemma:equilibriumcharacterization}, the consumer creates the content $\argmax_{\w \ge 0}(\w - \price \cdot \TokFn(\w))$ that maximizes their utility. 

\paragraph{Case 2: $\price \in [\LowerThreshold(\NumConsumers, \subfee, \TokFn), \UpperThreshold(\NumConsumers, \subfee, \TokFn)]$.} In this case, the intermediary survives (Theorem \ref{thm:middlemanusageexposure}). By Lemma \ref{lemma:equilibriumcharacterization}, the intermediary produces content 
\[\wmiddleman = \subfee + \max_{\w \ge 0} (\w - \price \TokFn(\w)).\] 
\end{proof}

\subsection{Proof of Theorem \ref{thm:contentquality}}

We prove Theorem \ref{thm:contentquality}.
\begin{proof}[Proof of Theorem \ref{thm:contentquality}]

First, we show that the quality is continuous in $\price$ when $\price \neq \LowerThreshold(\NumConsumers, \subfee, g)$ and $\price \neq \UpperThreshold(\NumConsumers, \subfee, g)$. This follows from the functional forms from Proposition \ref{prop:contentquality}.

Next, we show that the quality is decreasing in $\price$ when $\price \neq \LowerThreshold(\NumConsumers, \subfee, g)$ and $\price \neq \UpperThreshold(\NumConsumers, \subfee, g)$. We again use Proposition \ref{prop:contentquality}. For $\price < \LowerThreshold(\NumConsumers, \subfee, g)$ or $\price > \UpperThreshold(\NumConsumers, \subfee, g)$, the content quality is $\argmax_{w \ge 0}(w - \price g(w))$. This is equal to $w^*(\price)$ such that $g'(w^*(\price)) = 1/\price$. Since $g'$ is increasing in its argument, this is  decreasing in $\rho$. For $\price \in (\LowerThreshold(\NumConsumers, \subfee, g), \UpperThreshold(\NumConsumers, \subfee, g))$, we the content quality is $\subfee + \max_{w \ge 0}(w - \price g(w))$. By Lemma \ref{lemma:basicenvelope}, this is decreasing in $\price$. 
 
Next, we analyze the content quality at the thresholds $\price = \LowerThreshold(\NumConsumers, \subfee, g)$ and $\price = \UpperThreshold(\NumConsumers, \subfee, g)$. We again use Proposition \ref{prop:contentquality}. It suffices to show that:
\[\lim_{\price \rightarrow^- \LowerThreshold(\NumConsumers, \subfee, g)} \argmax_{w \ge 0}(w - \price g(w))  > \subfee + \max_{w \ge 0}(w - (\LowerThreshold(\NumConsumers, \subfee, g)) g(w)) \]
and 
\[\lim_{\price \rightarrow^+ \UpperThreshold(\NumConsumers, \subfee, g)} \argmax_{w \ge 0}(w - \price g(w))  < \subfee + \max_{w \ge 0}(w - (\UpperThreshold(\NumConsumers, \subfee, g)) \cdot g(w)). \]

For the first limit, we can rewrite the desired inequality as $\subfee < \LowerThreshold(\NumConsumers, \subfee, g) \cdot g(w^*(\LowerThreshold(\NumConsumers, \subfee, g)))$, where $w^*(\price) = \argmax_{w}(w - \price g(w))$. This holds because by Lemma \ref{lemma:optima} we know that at $\price = \LowerThreshold(\NumConsumers, \subfee, g)$, the sign of the derivative of $\price g(\subfee + \max_w(w - \price g(w))$ is negative, and by Lemma \ref{lemma:derivative}, we know that this means that $g(w^*(\price)) > \subfee g'(w^*(\price))$, which means that $\price g(w^*(\price)) > \subfee$ as desired. 

For the second limit, we can rewrite the desired inequality as $\subfee > \UpperThreshold(\NumConsumers, \subfee, g) \cdot g(w^*(\UpperThreshold(\NumConsumers, \subfee, g)))$, where $w^*(\price) = \argmax_{w}(w - \price g(w))$. This holds because by Lemma \ref{lemma:optima} we know at $\price = \UpperThreshold(\NumConsumers, \subfee, g)$, that the sign of the derivative of $\price g(\subfee + \max_w(w - \price g(w))$ is positive, and by Lemma \ref{lemma:derivative}, we know that this means that $g(w^*(\price)) < \subfee g'(w^*(\price))$, which means that $\price g(w^*(\price)) < \subfee$ as desired. 

\end{proof}

\subsection{Proof of Proposition \ref{prop:middlemanutility}}

We prove Proposition \ref{prop:middlemanutility}. This result follows from easily Lemma \ref{lemma:equilibriumcharacterization}.

\begin{proof}[Proof of Proposition \ref{prop:middlemanutility}]

We split into two cases: $\price < \LowerThreshold(\NumConsumers, \subfee, \TokFn)$ or $\price > \UpperThreshold(\NumConsumers, \subfee, \TokFn)$, and $\price \in [\LowerThreshold(\NumConsumers, \subfee, \TokFn), \UpperThreshold(\NumConsumers, \subfee, \TokFn)]$.

\paragraph{Case 1: $\price < \LowerThreshold(\NumConsumers, \subfee, \TokFn)$ or $\price > \UpperThreshold(\NumConsumers, \subfee, \TokFn)$.} In this case, disintermediation occurs (Theorem \ref{thm:middlemanusageexposure}). This means that the  intermediary has utility zero. 

\paragraph{Case 2: $\price \in [\LowerThreshold(\NumConsumers, \subfee, \TokFn), \UpperThreshold(\NumConsumers, \subfee, \TokFn)]$.}  In this case, the intermediary survives (Theorem \ref{thm:middlemanusageexposure}). By Lemma \ref{lemma:equilibriumcharacterization}, the intermediary produces content  at equilibrium 
\[\wmiddleman = \subfee + \max_{\w \ge 0} (\w - \price \TokFn(\w)),\] and their utility is their revenue $\subfee \cdot \NumConsumers$ minus their costs $\price \TokFn\left(\subfee +  \max_{w \ge 0} \left( w - \price \cdot \TokFn(\w)\right) \right)$.
\end{proof}

\subsection{Proof of Theorem \ref{thm:middlemanutility}}

We prove Theorem \ref{thm:middlemanutility}.
\begin{proof}[Proof of Theorem \ref{thm:middlemanutility}]

First, we show that the intermediary utility is inverse U-shaped. By Proposition \ref{prop:middlemanutility}, it suffices to show that $\price \cdot \TokFn(\subfee + \max_{\w \ge 0}(\w - \price \cdot \TokFn(\w)))$ is U-shaped as a function of $\price$. This follows from Lemma \ref{lemma:optima}.

Next, we compute the maximum intermediary utility. By Proposition \ref{prop:middlemanutility}, it suffices to find the minimum value of $\price \cdot \TokFn(\subfee + \max_{\w \ge 0}(\w - \price \cdot \TokFn(\w)))$. We again apply Lemma \ref{lemma:optima} to see that this is equal to $\alpha$, which means that the maximum intermediary utility is equal to $\alpha(\NumConsumers - 1)$ Since $\price \cdot \TokFn(\subfee + \max_{\w \ge 0}(\w - \price \cdot \TokFn(\w)))$ is U-shaped and using Theorem \ref{thm:middlemanusageexposure}, we also know that this optima is attained for $\price$ in the range where intermediation occurs.

We now turn to content $\wmiddleman$ produced at this optima. Using Lemma \ref{lemma:optima} again, we also see that the optima is attained at $\price$ such that $g(\argmax(w - \price g(w)) = \alpha \cdot g'(\argmax(w - \price g(w))$. Using that $g$ is convex, we know that $g'(\argmax(w - \price g(w)) = \frac{1}{\price}$, so this implies that:
\[\price \cdot g(\argmax(w - \price g(w)) = \alpha.\]
This means that:
\[\wmiddleman = \subfee + \max_{\w \ge 0}(\w - \price \cdot \TokFn(\w)) = \subfee + \argmax(w - \price g(w)) - \price \TokFn(\argmax(w - \price g(w))) = \argmax(w - \price g(w)) .\]
as desired. 
\end{proof}

\subsection{Proof of Theorem \ref{thm:consumerutility}}

We prove Theorem \ref{thm:consumerutility}. This follows easily from Lemma \ref{lemma:equilibriumcharacterization}.
\begin{proof}[Proof of Theorem \ref{thm:consumerutility}]
We split into two cases: $\price < \LowerThreshold(\NumConsumers, \subfee, \TokFn)$ or $\price > \UpperThreshold(\NumConsumers, \subfee, \TokFn)$, and $\price \in [\LowerThreshold(\NumConsumers, \subfee, \TokFn), \UpperThreshold(\NumConsumers, \subfee, \TokFn)]$.

\paragraph{Case 1: $\price < \LowerThreshold(\NumConsumers, \subfee, \TokFn)$ or $\price > \UpperThreshold(\NumConsumers, \subfee, \TokFn)$.} In this case, disintermediation occurs (Theorem \ref{thm:middlemanusageexposure}). By Lemma \ref{lemma:equilibriumcharacterization}, the consumer produces content $\argmax_{w \ge 0}(w - \price g(\w))$ and their utility is thus  $\text{max}_{w \ge 0}(w - \price g(\w))$. 

\paragraph{Case 2: $\price \in [\LowerThreshold(\NumConsumers, \subfee, \TokFn), \UpperThreshold(\NumConsumers, \subfee, \TokFn)]$.}  In this case, the intermediary survives (Theorem \ref{thm:middlemanusageexposure}). By Lemma \ref{lemma:equilibriumcharacterization}, the intermediary produces content  at equilibrium 
\[\wmiddleman = \subfee + \max_{\w \ge 0} (\w - \price \TokFn(\w)).\] The consumer utility is thus $\max_{\w \ge 0} (\w - \price \TokFn(\w))$. 
\end{proof}

\subsection{Proof of Corollary \ref{cor:consumerutility}}

We prove Corollary \ref{cor:consumerutility}.
\begin{proof}[Proof of Corollary \ref{cor:consumerutility}]
We apply Theorem \ref{thm:consumerutility} to see that the consumer utility is $\max_{\w \ge 0} (\w - \price \TokFn(\w))$. We apply Lemma \ref{lemma:basicenvelope} to see that the derivative of $\max_{\w \ge 0} (\w - \price \TokFn(\w))$ with respect to $\price$ is negative. This proves that $\max_{\w \ge 0} (\w - \price \TokFn(\w))$ is continuous and decreasing in $\price$.  
\end{proof}

\subsection{Proof of Proposition \ref{prop:socialwelfare}}

We prove Proposition \ref{prop:socialwelfare}.
\begin{proof}[Proof of Proposition \ref{prop:socialwelfare}]
We add up the utility of consumers, suppliers, and the intermediary. By Theorem \ref{thm:consumerutility}, the total utility of consumers is equal to  $\NumConsumers \cdot \left(\max_{\w \ge 0}\left(w - \price \cdot \TokFn(\w) \right)\right)$ regardless of whether disintermediation occurs (Theorem \ref{thm:consumerutility}). By Lemma \ref{lemma:equilibriumcharacterization}, the suppliers choose $\price = \ic$ and thus have have zero profit. 

We split into two cases: $\price < \LowerThreshold(\NumConsumers, \subfee, \TokFn)$ or $\price > \UpperThreshold(\NumConsumers, \subfee, \TokFn)$, and $\price \in [\LowerThreshold(\NumConsumers, \subfee, \TokFn), \UpperThreshold(\NumConsumers, \subfee, \TokFn)]$.

\paragraph{Case 1: $\price < \LowerThreshold(\NumConsumers, \subfee, \TokFn)$ or $\price > \UpperThreshold(\NumConsumers, \subfee, \TokFn)$.} In this case, disintermediation occurs (Theorem \ref{thm:middlemanusageexposure}). When disintermediation occurs, the social welfare is thus equal to the total utility of consumers, which is   $\NumConsumers \cdot \left(\max_{\w \ge 0}\left(w - \price \cdot \TokFn(\w) \right)\right)$. 

\paragraph{Case 2: $\price \in [\LowerThreshold(\NumConsumers, \subfee, \TokFn), \UpperThreshold(\NumConsumers, \subfee, \TokFn)]$.}  In this case, the intermediary survives (Theorem \ref{thm:middlemanusageexposure}). When the intermediary survives, the social welfare is equal to the intermediary's utility plus the total consumer utility. By Proposition \ref{prop:middlemanutility}, the intermediary's utility is $\subfee \NumConsumers - \price g(\subfee + \max_{w \ge 0}(w - \price g(w)))$. This means that the social welfare is equal to $\subfee \NumConsumers - \price g(\subfee + \max_{w \ge 0}(w - \price g(w))) + \NumConsumers \cdot \max_{w \ge 0}(w - \price g(w))$ as desired. 
    
\end{proof}

\subsection{Proof of Proposition \ref{prop:socialplanner}}

We prove Proposition \ref{prop:socialplanner}.
\begin{proof}[Proof of Proposition \ref{prop:socialplanner}]
Let $\price = \min(\ic + \costadditional, \costmanual)$. Production is done through the suppliers if $\ic + \costadditional \le \costmanual$ and using manual content creation if $\ic + \costadditional >\costmanual$. 

If the intermediary does not exist, then the market outcome that maximizes the social welfare is that each consumer produces content $\argmax_{w \ge 0}(w - \price g(w))$. The social welfare is $\NumConsumers \cdot \max_{w \ge 0}(w - \price g(w))$. 

For the case where the intermediary exists, we construct a market outcome that maximizes the social welfare: the suppliers set prices $\nu_1 = \ldots \nu_P = \ic$ equal to the supply-side costs, the intermediary produces content \[\wmiddleman = \argmax_{\w \ge 0} (\NumConsumers \w - \price \cdot \TokFn(\w)). \]
and all consumers $j \in [\NumConsumers]$ all choose consumption mode $\action{j} = \Middleman$ and consume the content $\wuser{j} = \wmiddleman$ created by the intermediary. The social welfare is $\max_{w \ge 0}(\NumConsumers \cdot w - \price g(w))$. 

\end{proof}

\subsection{Proof of Theorem \ref{thm:socialwelfare}}

We prove 
Theorem \ref{thm:socialwelfare}.
\begin{proof}[Proof of Theorem \ref{thm:socialwelfare}]
We apply Proposition \ref{prop:socialwelfare} to see that the social welfare is equal to:
\[
\begin{cases}
\NumConsumers \cdot \left(\max_{\w \ge 0}\left(w - \price \cdot \TokFn(\w) \right)\right) & \text{ if } \price < \LowerThreshold(\NumConsumers, \subfee, \TokFn)  \\
\NumConsumers \subfee 
 -\price \TokFn\left(\subfee +  \max_{w \ge 0} \left( w - \price \cdot \TokFn(\w)\right) \right) + \NumConsumers \max_{\w \ge 0}\left(w - \price \cdot \TokFn(\w) \right)  & \text{ if } \price \in [\LowerThreshold(\NumConsumers, \subfee, \TokFn), \UpperThreshold(\NumConsumers, \subfee,  \TokFn)]\\
\NumConsumers \cdot \left(\max_{\w \ge 0}\left(w - \price \cdot \TokFn(\w) \right)\right) & \text{ if } \price > \UpperThreshold(\NumConsumers, \subfee, \TokFn).
\end{cases}
\]
Throughout this proof, let $w^*(\price')$ be the unique solution to $\max_{w \ge 0}(w - \price' \cdot g(w))$ (Lemma \ref{lemma:uniqueoptima}). 

First, we show that the social welfare is continuous in the production costs. This follows immediately within each of the three regimes, and at the boundaries, it follows from the fact that $\subfee \cdot \NumConsumers - \price g(\subfee + \max_{w \ge 0}(w - \price g(w))) = 0$ so 
\[\NumConsumers \subfee 
 -\price \TokFn\left(\subfee +  \max_{w \ge 0} \left( w - \price \cdot \TokFn(\w)\right) \right) + \NumConsumers \max_{\w \ge 0}\left(w - \price \cdot \TokFn(\w) \right) = \NumConsumers \cdot \max_{\w \ge 0}\left(w - \price \cdot \TokFn(\w) \right). \]

Next, we show that the social welfare is decreasing in $\price$. For the first and third regime, the social welfare is equal to the total consumer utility. This is $\NumConsumers$ times the utility of any given consumer. So by Theorem \ref{thm:consumerutility}, this is decreasing in $\price$. For the second regime, we take a derivative of $\NumConsumers \subfee - 
 -\price \TokFn\left(\subfee +  \max_{w \ge 0} \left( w - \price \cdot \TokFn(\w)\right) \right) + \NumConsumers \max_{\w \ge 0}\left(w - \price \cdot \TokFn(\w) \right)$ with respect to $\price$. By Lemma \ref{lemma:derivative} and Lemma \ref{lemma:basicenvelope}, this is equal to:
 \[-g(\subfee + \max_{w \ge 0}(w - \price g(w)) + g'(\subfee + \max_{w \ge 0}(w - \price g(w)) \frac{g(w^*(\price))}{g'(w^*(\price))} - \NumConsumers \cdot g(w^*(\price)). \]
 where $w^*(\price) = \argmax_{w \ge 0}(w - \price g(w))$. It suffices to show that:
 \[ g(\subfee + \max_{w \ge 0}(w - \price g(w)) + \NumConsumers \cdot g(w^*(\price)) >  g'(\subfee + \max_{w \ge 0}(w - \price g(w)) \frac{g(w^*(\price))}{g'(w^*(\price))}.\]

We split into two cases: (1) $\subfee + \max_{w \ge 0}(w - \price g(w) > w^*(\price)$ and (2) $\subfee + \max_{w \ge 0}(w - \price g(w) \le w^*(\price)$.

\paragraph{Case 1: $\subfee + \max_{w \ge 0}(w - \price g(w) > w^*(\price)$.} It suffices to show that 
 \[ g(\subfee + \max_{w \ge 0}(w - \price g(w)) >  g'(\subfee + \max_{w \ge 0}(w - \price g(w)) \frac{g(w^*(\price))}{g'(w^*(\price))}.\]
We can write this:
 \[ \frac{g(\subfee + \max_{w \ge 0}(w - \price g(w))}{  g'(\subfee + \max_{w \ge 0}(w - \price g(w))} > \frac{g(w^*(\price))}{g'(w^*(\price))}.\]
 Using that $g$ is log-concave, we know that $g(x) / g'(x)$ is increasing in $x$, so we know that this holds. 
 
\paragraph{Case 2: $\subfee + \max_{w \ge 0}(w - \price g(w)\le w^*(\price)$.} It suffices to show that \[ \NumConsumers \cdot g(w^*(\price)) >  g'(\subfee + \max_{w \ge 0}(w - \price g(w)) \frac{g(w^*(\price))}{g'(w^*(\price))}.\]
We can write this as:
\[ \NumConsumers g'(w^*(\price)) >  g'(\subfee + \max_{w \ge 0}(w - \price g(w)).\]
Since $g'(x)$ is increasing in $x$, this means that $g'(w^*(\price)) \ge g'(\subfee + \max_{w \ge 0}(w - \price g(w))$. This coupled with $\NumConsumers > 0$ implies the desired statement.

Next, we show that the social welfare is strictly greater than the social planner's optimal without the intermediary when $\price \in (\LowerThreshold(\NumConsumers, \subfee, \TokFn), \UpperThreshold(\NumConsumers, \subfee,  \TokFn))$. Using Proposition \ref{prop:socialplanner}, the social planner's optimal welfare without the intermediary is equal to $\NumConsumers \cdot \max_{w \ge 0}(w - \price g(w))$. This means that it suffices to show that:
\[ \subfee \cdot \NumConsumers - \price g(\subfee + \max_{w} (w - \price g(w)) > 0.\]
This follows from the fact that $\subfee \cdot \NumConsumers - \price g(\subfee + \max_{w} (w - \price g(w)) \ge 0$ when
$\price \in (\LowerThreshold(\NumConsumers, \subfee, \TokFn), \UpperThreshold(\NumConsumers, \subfee,  \TokFn))$. To obtain a strict inequality, it suffices to show that the derivative of  $\price g(\subfee + \max_{w} (w - \price g(w))$ is positive at the boundaries. This is because by Lemma \ref{lemma:optima}, the optimum occurs when $\price g(\subfee + \max_{w} (w - \price g(w)) = \alpha < \alpha \NumConsumers$. 

Finally, we show that it is strictly below the social planner's optimal with the intermediary except at at most one bliss point. We know that the social planner's optimal is always at least as large as the market with intermediary. Thus, it suffices to show that the social welfare is \textit{not equal to} the social planner's optima 
$\max_{\w \ge 0}\left(\NumConsumers \cdot w - \price \cdot \TokFn(\w) \right)$ except at at most one bliss point. We split into cases depending on the value of $\price$ and depending on whether $\subfee + \max_{w} (w - \price g(w)) = \argmax_{w \ge 0}(C \cdot w - \price g(\w))$ holds, and show that the social welfare is not equal to the social optima unless $\price \in [\LowerThreshold(\NumConsumers, \subfee, \TokFn), \UpperThreshold(\NumConsumers, \subfee, \TokFn)]$ \text{ and } $\subfee + \max_{w} (w - \price g(w)) = \argmax_{w \ge 0}(C \cdot w - \price g(\w))$. 

\paragraph{Case 1: $\price < \LowerThreshold(\NumConsumers, \subfee, \TokFn)$ or $\price > \UpperThreshold(\NumConsumers, \subfee, \TokFn)$.} The social welfare is equal to $\NumConsumers \cdot \left(\max_{\w \ge 0}\left(w - \price \cdot \TokFn(\w) \right)\right)$. We see that:
\[\NumConsumers \cdot \left(\max_{\w \ge 0}\left(w - \price \cdot \TokFn(\w) \right)\right) = \max_{\w \ge 0}\left(\NumConsumers \cdot w - \NumConsumers \cdot \price \cdot \TokFn(\w) \right) \neq_{(A)} \max_{\w \ge 0}\left(\NumConsumers \cdot w - \price \cdot \TokFn(\w) \right), \]
where (A) holds because both optima occur at $w > 0$. 

\paragraph{Case 2: $\price \in [\LowerThreshold(\NumConsumers, \subfee, \TokFn), \UpperThreshold(\NumConsumers, \subfee, \TokFn)]$ \text{ and } $\subfee + \max_{w} (w - \price g(w)) \neq \argmax_{w \ge 0}(C \cdot w - \price g(\w))$.} In this case, we see that the social welfare is equal to:
\[C \cdot (\subfee + \max_{w} (w - \price g(w))) - \price g(\subfee + \max_{w} (w - \price g(w))) \]
and the social planner's optima is equal to $\text{max}_{w \ge 0}(C \cdot w - \price g(\w))$. Since the function $f(x) = C \cdot x - \price g(x)$ has a unique global optima, this means that \[C \cdot (\subfee + \max_{w} (w - \price g(w))) - \price g(\subfee + \max_{w} (w - \price g(w))) \neq \text{max}_{w \ge 0}(C \cdot w - \price g(\w)).  \]

\paragraph{Case 3: $\price \in [\LowerThreshold(\NumConsumers, \subfee, \TokFn), \UpperThreshold(\NumConsumers, \subfee, \TokFn)]$ \text{ and } $\subfee + \max_{w} (w - \price g(w)) = \argmax_{w \ge 0}(C \cdot w - \price g(\w))$.}  In this case, we see that the social welfare is equal to:
\[C \cdot (\subfee + \max_{w} (w - \price g(w))) - \price g(\subfee + \max_{w} (w - \price g(w))) \]
and the social planner's optima is equal to $\text{max}_{w \ge 0}(C \cdot w - \price g(\w))$. These expressions are equal because $\subfee + \max_{w} (w - \price g(w)) = \argmax_{w \ge 0}(C \cdot w - \price g(\w))$. 

Now, it suffices to show that there is at most one value of $\price$ such that $\price \in [\LowerThreshold(\NumConsumers, \subfee, \TokFn), \UpperThreshold(\NumConsumers, \subfee, \TokFn)]$ \text{ and } $\subfee + \max_{w} (w - \price g(w)) = \argmax_{w \ge 0}(C \cdot w - \price g(\w))$. It suffices to show that the derivative of \[\subfee + \max_{w} (w - \price g(w)) - \argmax_{w \ge 0}(C \cdot w - \price g(\w)) \]
with respect to $\price$ is always positive. First, let's simplify this expression. Note that $\argmax_{w \ge 0}(C \cdot w - \price g(\w))$ occurs when $g'(w) = \frac{C}{\price}$, which means that $\argmax_{w \ge 0}(\NumConsumers \cdot w - \price g(\w)) = w^*(\price / \NumConsumers)$. This means that the expression is equal to:
\[\subfee + \max_{w} (w - \price g(w)) - w^*\left( \frac{\price}{\NumConsumers} \right). \]
Now, taking a derivative and applying Lemma \ref{lemma:basicenvelope}, we obtain:
\[-g(w^*(\price)) - (w^*)'\left( \frac{\price}{\NumConsumers} \right) \cdot \frac{1}{\NumConsumers}. \]
To compute $(w^*)'\left( z \right)$, we use the fact that $g'(w^*\left( z \right)) = \frac{1}{z}$, so $w^*(z) = (g')^{-1}(1/z)$. By the inverse function theorem, this means that:
\[(w^*)'\left( z \right) = -\frac{1}{z^2 \cdot g''(w^*\left( z \right))} = -\frac{(g'(w^*(z)))^2}{ g''(w^*\left( z \right))}. \]
Plugging this in, we obtain:
\[-g(w^*(\price)) + \frac{(g'(w^*\left( \frac{\price}{\NumConsumers} \right)))^2}{ g''(w^*\left( \frac{\price}{\NumConsumers} \right))} \cdot \frac{1}{\NumConsumers},\]
This is positive if and only if:
\[ (g'(w^*\left( \frac{\price}{\NumConsumers} \right)))^2 \ge \NumConsumers g(w^*(\price)) \cdot g''(w^*\left( \frac{\price}{\NumConsumers} \right)). \]
By log-concavity, we know that: \[ (g'(w^*\left( \frac{\price}{\NumConsumers} \right)))^2 \ge g(w^*\left( \frac{\price}{\NumConsumers} \right)) \cdot g''(w^*\left( \frac{\price}{\NumConsumers} \right)). \]
Using log-concavity again and that $g'(w^*(z)) = \frac{1}{z}$, we know that:
\[\frac{\price}{\NumConsumers}  \cdot g(w^*\left( \frac{\price}{\NumConsumers} \right))= \frac{g(w^*\left( \frac{\price}{\NumConsumers} \right)) }{g'(w^*\left( \frac{\price}{\NumConsumers} \right))} > \frac{g(w^*(\price))}{g'(w^*(\price))} = \price \cdot g(w^*(\price)). \]
This means that:
\[ g(w^*\left( \frac{\price}{\NumConsumers} \right)) > \NumConsumers \cdot g(w^*(\price)). \]
Putting this all together, this implies that:
\[ (g'(w^*\left( \frac{\price}{\NumConsumers} \right)))^2 \ge g(w^*\left( \frac{\price}{\NumConsumers} \right)) \cdot g''(w^*\left( \frac{\price}{\NumConsumers} \right)) > \NumConsumers g(w^*(\price)) \cdot g''(w^*\left( \frac{\price}{\NumConsumers} \right)). \]
as desired.

\end{proof}

\subsection{Statement and Proof of Proposition \ref{prop:blisspoint}}

We state and prove Proposition \ref{prop:blisspoint}.
\begin{proposition}
\label{prop:blisspoint}
Consider the setup of Theorem \ref{thm:specialcasemiddlemanusageexposure}, and suppose that $\TokExp \ge 2$. Then, there exists a bliss point $\price$ where the social welfare of the market equals the social planner's optima.   
\end{proposition}
\begin{proof}
Following the proof of Theorem \ref{thm:socialwelfare}, we know the social welfare of the market equals the social planner's optima if and only if $\price \in [\LowerThreshold(\NumConsumers, \subfee, \TokFn), \UpperThreshold(\NumConsumers, \subfee, \TokFn)]$ \text{ and } $\subfee + \max_{w} (w - \price g(w)) = \argmax_{w \ge 0}(C \cdot w - \price g(\w))$. Observe that:
\[\subfee + \max_{w} (w - \price g(w)) =  \alpha + \price^{-\frac{1}{\beta-1}} \left(\beta^{-\frac{1}{\beta-1}} - \beta^{-\frac{\beta}{\beta-1}} \right) \]
and:
\[\argmax_{w \ge 0}(C^{\frac{1}{\beta-1}} \cdot w - \price g(\w)) = \price^{-\frac{1}{\beta-1}} \NumConsumers^{\frac{1}{\beta-1}} \beta^{-\frac{1}{\beta-1}}. \]
These two expressions are equal when:
\[ \price^{-\frac{1}{\beta-1}} \left((C^{\frac{1}{\beta-1}}-1) \beta^{-\frac{1}{\beta-1}} + \beta^{-\frac{\beta}{\beta-1}} \right) = \subfee, \]
which can be written as:
\[ \price^{-\frac{1}{\beta-1}}  = \frac{\subfee}{\left((C^{\frac{1}{\beta-1}}-1) \beta^{-\frac{1}{\beta-1}} + \beta^{-\frac{\beta}{\beta-1}} \right)}, \]
To show this occurs for $\price \in [\LowerThreshold(\NumConsumers, \subfee, \TokFn), \UpperThreshold(\NumConsumers, \subfee, \TokFn)]$, we observe that:
\begin{align*}
  \price g(\subfee + \max_{w} (w - \price g(w))) &= \price \left( \subfee + \frac{\subfee \left(\beta^{-\frac{1}{\beta-1}} - \beta^{-\frac{\beta}{\beta-1}} \right)}{\left((C^{\frac{1}{\beta-1}}-1) \beta^{-\frac{1}{\beta-1}} + \beta^{-\frac{\beta}{\beta-1}} \right)} \right)^{\beta}  \\
  &= \left(\frac{\left((C^{\frac{1}{\beta-1}}-1) \beta^{-\frac{1}{\beta-1}} + \beta^{-\frac{\beta}{\beta-1}} \right)}{\subfee}\right)^{\beta-1} \cdot \left( \subfee + \frac{\subfee \left(\beta^{-\frac{1}{\beta-1}} - \beta^{-\frac{\beta}{\beta-1}} \right)}{\left((C^{\frac{1}{\beta-1}}-1) \beta^{-\frac{1}{\beta-1}} + \beta^{-\frac{\beta}{\beta-1}} \right)} \right)^{\beta} \\
  &= \subfee \cdot \frac{\left(\left((C^{\frac{1}{\beta-1}}-1) \beta^{-\frac{1}{\beta-1}} + \beta^{-\frac{\beta}{\beta-1}} \right) + \left(\beta^{-\frac{1}{\beta-1}} - \beta^{-\frac{\beta}{\beta-1}} \right) \right)^{\beta}}{\left((C^{\frac{1}{\beta-1}}-1) \beta^{-\frac{1}{\beta-1}} + \beta^{-\frac{\beta}{\beta-1}} \right)} \\
  &= \subfee \cdot \frac{\left(C^{\frac{1}{\beta-1}} \cdot \beta^{-\frac{1}{\beta-1}}  \right)^{\beta}}{\left((C^{\frac{1}{\beta-1}}-1) \beta^{-\frac{1}{\beta-1}} + \beta^{-\frac{\beta}{\beta-1}} \right)}.
\end{align*}
We want to show this is at most $\subfee \cdot \NumConsumers$. It suffices to show that:
\[\left(C^{\frac{1}{\beta-1}} \cdot \beta^{-\frac{1}{\beta-1}}  \right)^{\beta} \le \NumConsumers \left((C^{\frac{1}{\beta-1}}-1) \beta^{-\frac{1}{\beta-1}} + \beta^{-\frac{\beta}{\beta-1}} \right).  \]
This simplifies to 
\[C^{\frac{\beta}{\beta-1}} \cdot \beta^{-\frac{\beta}{\beta-1}} \le C^{\frac{\beta}{\beta-1}} \beta^{-\frac{1}{\beta-1}} - \NumConsumers \beta^{-\frac{1}{\beta-1}} + \NumConsumers \beta^{-\frac{\beta}{\beta-1}}, \]
which simplifies to:
\[ C \ge 1,\]
which holds.  
\end{proof}

\section{Proofs for Section \ref{sec:extensions}}\label{appendix:proofsextensions}

\subsection{Proof of Lemma \ref{lemma:price} and Theorem  \ref{thm:extensionmonopolist}}

For the purposes of this result, we slightly modify the tiebreaking rules: we assume that each consumer $j$ tiebreaks in favor of direct usage (i.e., $\action{j} = \Direct$) rather than in favor of the intermediary (i.e., $\action{j} = \Middleman$) when $\price < \UpperThreshold(\NumConsumers, \subfee, \TokExp)$, but tiebreaks in favor of the intermediary when $\price \ge \UpperThreshold(\NumConsumers, \subfee, \TokExp)$.

First, we prove the following modified version of Lemma \ref{lemma:middlemansubgame} for this modified tiebreaking rule.
\begin{lemma}
\label{lemma:middlemansubgamemonopolist}
Consider the setup of Section \ref{subsec:monopolist} and Theorem \ref{thm:extensionmonopolist}. Suppose that the supplier chooses price $\price_1$ and consider the subgame between the intermediary and consumers (Stages 2-3). Let $\price = \min(\costadditional + \price_1, \costmanual)$, and let $\LowerThreshold(\NumConsumers, \subfee, \TokExp)$ and $\UpperThreshold(\NumConsumers, \subfee, \TokExp)$ be defined as in Theorem \ref{thm:specialcasemiddlemanusageexposure}. Under the tiebreaking assumptions above, there exists a unique pure strategy where: 
\begin{itemize}
\item Suppose that $\price \le \LowerThreshold(\NumConsumers, \subfee, \TokExp)$ or $\price > \UpperThreshold(\NumConsumers, \subfee, \TokExp)$. Then $\wmiddleman = 0$. Moreover,  for all $j \in [\NumConsumers]$, it holds that $\action{j} = \Direct$, 
    \[\wuser{j} = \argmax_{w \ge 0} \left(\w -  \price(1 + \marg) \TokFn(w)\right).\] Moreover, if $\costmanual < \costadditional + \min_{i \in [\NumProviders]} \price_i$, then the consumer chooses $\providerchoice{j} = 0$. Otherwise, the consumer chooses $\providerchoice{j} = \text{argmin}_{i \in [\NumProviders]} \price_i$ (tie-breaking in favor of suppliers with a lower index). 
    \item Suppose that $\price \in (\LowerThreshold(\NumConsumers, \subfee, \TokExp), \UpperThreshold(\NumConsumers, \subfee, \TokExp)]$.  Then 
    \[\wmiddleman = \wuser{j} = \subfee + \max_{w \ge 0} \left(\w - \price (1+\marg) \TokFn(w)\right).\] Moreover, if $\costmanual < \costadditional + \min_{i \in [\NumProviders]} \price_i$, then the intermediary chooses $\providerchoicemiddleman = 0$; otherwise, the intermediary chooses $\providerchoicemiddleman = \text{argmin}_{i \in [\NumProviders]} \price_i$ (tie-breaking in favor of suppliers with a lower index). Finally, it holds that $\action{j} = \Middleman$ and $\wuser{j} = \wmiddleman$ for all $j \in [\NumConsumers]$. 
\end{itemize}
\end{lemma}
\begin{proof}
    The proof follows similarly to the proof of Lemma \ref{lemma:middlemansubgame}, but additionally uses the analysis of the condition $\price g(\subfee + \max_{w \ge 0}(w-\price g(w)) - \subfee \NumConsumers$ from Lemma \ref{lemma:twosolutions} and Theorem \ref{thm:middlemanusageexposure}.
\end{proof}

We prove an analogue of Lemma \ref{lemma:equilibriumcharacterization} for this setting, which strengthens Lemma \ref{lemma:price} (this result directly implies Lemma \ref{lemma:price}).  
\begin{lemma}
\label{lemma:equilibriumcharacterizationmonopolist}
Consider the setup of Section \ref{subsec:monopolist} and Theorem \ref{thm:extensionmonopolist}. 
Under the tiebreaking assumptions described above, there exists a unique pure strategy equilibrium which takes the following form: the supplier chooses the price $\price_1$ as specified in Lemma \ref{lemma:price}, and the intermediary and consumers choose actions according to the subgame equilibrium constructed in Lemma \ref{lemma:middlemansubgamemonopolist}. 
\end{lemma}

To prove Lemma \ref{lemma:equilibriumcharacterizationmonopolist}, a key technical challenge is that the supplier can influence whether intermediation occurs in terms of how it sets its prices. To capture these effects, we separately analyze the optimal price for the suppler in each regime in the following intermediate lemmas. 

First, we bound the optimal price for the supplier in the range which induces disintermediation.
\begin{lemma}
\label{lemma:maximumpriceconsumer}
Consider the setup of Theorem \ref{thm:extensionmonopolist}. Let $\price^M$
be the value of $\price'$ that attains the maximum 
\[\max_{\price' \ge \price, \price \in [0, \LowerThreshold(\NumConsumers, \subfee, \TokExp)] \cup [\UpperThreshold(\NumConsumers, \subfee, \TokExp), \infty)} \left(\NumConsumers \cdot (\price' - \price) \cdot g\left(\argmax_{\w \ge 0} (\w - \price' g(\w)) \right). \right)\]
Then, $\price^M = \TokExp \cdot \price$ if $\TokExp \cdot \price \le \LowerThreshold(\NumConsumers, \subfee, \TokExp)$ or $\TokExp \cdot \price \ge \UpperThreshold(\NumConsumers, \subfee, \TokExp)$.  Otherwise, $\price^M \in \left\{\LowerThreshold(\NumConsumers, \subfee, \TokExp), \UpperThreshold(\NumConsumers, \subfee, \TokExp) \right\}$, and the optimal value is upper bounded by the value of $\NumConsumers \cdot (\price' - \price) \cdot g\left(\argmax_{\w \ge 0} (\w - \price' g(\w)) \right)$ when $\price' = \TokExp \cdot \price$. 
\end{lemma}
\begin{proof}
Throughout the proof, let $w^*(\price) = \argmax_{w}(w - \price g(w))$. Using the structure of $g(w) = w^{\beta}$, we know that:
\[w^*(\price') = (\price')^{-\frac{1}{\beta-1}} \beta^{-\frac{1}{\beta-1}}. \]
This means that the objective can be simplified to:
\[ C(\price' - \price) (\price')^{-\frac{\beta}{\beta-1}} \beta^{-\frac{\beta}{\beta-1}}. \]

We split into two cases: (1) $\TokExp \cdot \price \le \LowerThreshold(\NumConsumers, \subfee, \TokExp)$ or $\TokExp \cdot \price \ge \UpperThreshold(\NumConsumers, \subfee, \TokExp)$, and (2) $\TokExp \cdot \price \in  (\LowerThreshold(\NumConsumers, \subfee, \TokExp), \UpperThreshold(\NumConsumers, \subfee, \TokExp))$.

\paragraph{Case 1: $\TokExp \cdot \price \le \LowerThreshold(\NumConsumers, \subfee, \TokExp)$ or $\TokExp \cdot \price \ge \UpperThreshold(\NumConsumers, \subfee, \TokExp)$.} We take a first-order condition to obtain that $\price^M = \TokExp \cdot \price$ as desired. 

\paragraph{Case 2: $\TokExp \cdot \price \in  (\LowerThreshold(\NumConsumers, \subfee, \TokExp), \UpperThreshold(\NumConsumers, \subfee, \TokExp))$.} In this case, the function is increasing on $[\LowerThreshold(\NumConsumers, \subfee, \TokExp), \price)$ and decreasing on $(\price, \UpperThreshold(\NumConsumers, \subfee, \TokExp]$. This means that the optima for the constrained domain is attained at $\price^M \in \left\{\LowerThreshold(\NumConsumers, \subfee, \TokExp), \UpperThreshold(\NumConsumers, \subfee, \TokExp) \right\}$, and the optimal value is upper bounded by the optimal for the unconstrained domain, which is equal to $\NumConsumers \cdot (\price' - \price) \cdot g\left(\argmax_{\w \ge 0} (\w - \price' g(\w)) \right)$ when $\price' = \TokExp \cdot \price$. 

\end{proof}

Next,we bound the optimal price in the range which induces intermediation.
\begin{lemma}
\label{lemma:maximumpricemiddleman}
Let $g(\w) = w^{\TokExp}$ for $\TokExp > 1$. Let $\LowerThreshold(\NumConsumers, \subfee, \TokExp)$  and $\UpperThreshold(\NumConsumers, \subfee, \TokExp)$ be defined as in Theorem \ref{thm:specialcasemiddlemanusageexposure}, and let $\price \le \UpperThreshold(\NumConsumers, \subfee, \TokExp)$. Then, the maximum 
\[\max_{\price' \in [\max(\ic, \LowerThreshold(\NumConsumers, \subfee, \TokExp)), \UpperThreshold(\NumConsumers, \subfee, \TokExp)]} \left((\price' - \price) \cdot g\left( \subfee + \max_{w \ge 0} (w - \price' g(w)) \right) \right)\]
is uniquely attained at $\price = \UpperThreshold(\NumConsumers, \subfee, \TokExp)$.
\end{lemma}
\begin{proof}
We know that 
\[\max_{\price' \in [\max(\price, \LowerThreshold(\NumConsumers, \subfee, \TokExp)), \UpperThreshold(\NumConsumers, \subfee, \TokExp)]} \left((\price' - \price) \cdot g\left( \subfee + \max_{w \ge 0} (w - \price' g(w)) \right) \right)\] 
is equal to:
\[\max_{\price' \in [\max(\price, \LowerThreshold(\NumConsumers, \subfee, \TokExp)), \UpperThreshold(\NumConsumers, \subfee, \TokExp)]} \left(\left(1 - \frac{\price}{\price'}\right) \cdot \price' \cdot g\left( \subfee + \max_{w \ge 0} (w - \price' g(w)) \right) \right).\]
This is at most:
\[\left(\max_{\price \in [\max(\price, \LowerThreshold(\NumConsumers, \subfee, \TokExp)), \UpperThreshold(\NumConsumers, \subfee, \TokExp)]} \left(1 - \frac{\price}{\price'}\right)\right) \cdot \max_{\price' \in [\max(\price, \LowerThreshold(\NumConsumers, \subfee, \TokExp)), \UpperThreshold(\NumConsumers, \subfee, \TokExp)]} \left(\price' \cdot g\left( \subfee + \max_{w \ge 0} (w - \price' g(w)) \right) \right). \]
This is at most:
\[  \left(1 - \frac{\price}{\UpperThreshold(\NumConsumers, \subfee, \TokExp)} \right) \cdot \subfee \cdot \NumConsumers.\]
This value is uniquely attained at $\price = \UpperThreshold(\NumConsumers, \subfee, \TokExp)$ as desired. 
    
\end{proof}

Now, we use Lemma \ref{lemma:maximumpriceconsumer} and Lemma \ref{lemma:maximumpricemiddleman} to prove Lemma \ref{lemma:equilibriumcharacterizationmonopolist}.
\begin{proof}[Proof of Lemma \ref{lemma:equilibriumcharacterizationmonopolist}]
Let $\price = \min(\ic + \costadditional, \costmanual)$.
Let $w^*(\price')$ be the unique solution to $\max_{w \ge 0}(w - \price' \cdot g(w))$ (Lemma \ref{lemma:uniqueoptima}). 

The supplier can choose to induce disintermediation or intermediation, which affects how they set their optimal price. 
We use Lemma \ref{lemma:maximumpriceconsumer} and Lemma \ref{lemma:maximumpricemiddleman} to narrow down the set of possible optimal prices in each regime. When disintermediation is induced, the supplier earns profit  $(\price' - \price) \cdot \NumConsumers \cdot g(w^*(\price'))$ from choosing price $\price'$. By Lemma \ref{lemma:maximumpriceconsumer}, the optimal price in the disintermediation range is $\price \cdot \TokExp$ if that price induces disintermediation, or otherwise is in the set $\left\{\LowerThreshold(\NumConsumers, \subfee, \TokExp), \lim_{\epsilon \rightarrow^+ 0} (\UpperThreshold(\NumConsumers, \subfee, \TokExp) + \epsilon) \right\}$, where we use $\lim_{\epsilon \rightarrow^+ 0} (\UpperThreshold(\NumConsumers, \subfee, \TokExp) + \epsilon)$ to denote that the optimum in that range doesn't exist and the supplier would want to set their price arbitrarily close to $\UpperThreshold(\NumConsumers, \subfee, \TokExp)$. At these prices, the supplier's profit is  upper bounded by $(\nu' - \price) \cdot \NumConsumers \cdot g(w^*(\price'))$ where $\nu' = \price \cdot \TokExp$. The realized profit is:
\[\NumConsumers  \left(1 - \frac{\price}{\price'} \right) \cdot \price' g(w^*(\price')) = \NumConsumers \left(1 - \frac{\price}{\price'} \right) \cdot (\price')^{-\frac{1}{\beta-1}} \cdot \beta^{-\frac{\beta}{\beta-1}}. \]
When intermediation is induced, the profit is $(\price' - \price) \cdot g(\alpha + \max_{w \ge 0}(w - \price' g(w)))$. By Lemma \ref{lemma:maximumpricemiddleman}, the optimal price in the intermediation range is $\price = \UpperThreshold(\NumConsumers, \subfee, \TokExp)$. The obtained profit is:
\[\left(1 - \frac{\price}{\price'} \right) \alpha \NumConsumers.   \]

We claim that prices of the form $\UpperThreshold(\NumConsumers, \subfee, \TokExp) + \epsilon$ for sufficiently small $\epsilon$ are dominated by $\UpperThreshold(\NumConsumers, \subfee, \TokExp)$. Based on the above analysis, it suffices to show that $\price' g(w^*(\price')) < \subfee$. This holds by Lemma \ref{lemma:derivative} and Lemma \ref{lemma:optima}.   

We split into several cases depending on the value of $\price \cdot \TokExp$ and $\price$. 

\paragraph{Case 1: $\price \ge \UpperThreshold(\NumConsumers, \subfee, \TokExp)$.} In this case, we know that the supplier will set a price of $\price > \ic \ge \UpperThreshold(\NumConsumers, \subfee, \TokExp)$ to earn positive profit. This means that disintermediation occurs regardless of the price that they set. They thus set the price to $\price_1 = \price \cdot \TokExp$ to maximize their profit. 

\paragraph{Case 2: $\price \cdot \TokExp > \UpperThreshold(\NumConsumers, \subfee, \TokExp)$ and $\price < \UpperThreshold(\NumConsumers, \subfee, \TokExp)$.} By the above analysis, we know that the supplier will either set the price to be $\price \cdot \TokExp$ or to be $\UpperThreshold(\NumConsumers, \subfee, \TokExp)$. We show that there is a threshold value $\UpperThresholdMon(\NumConsumers, \subfee, \TokExp)$ in this range such that the supplier sets the price to be  $\UpperThreshold(\NumConsumers, \subfee, \TokExp)$ if $\price \le \UpperThresholdMon(\NumConsumers, \subfee, \TokExp)$ and sets the price to be $\price \cdot \TokExp$ otherwise. Note that the realized profit at $\price \cdot \TokExp$ is 
\[\NumConsumers \left(1 - \frac{1}{\beta} \right) (\price')^{-\frac{1}{\beta-1}} \beta^{-\frac{\beta}{\beta-1}} = \NumConsumers \left(1 - \frac{1}{\beta} \right) \price^{-\frac{1}{\beta-1}} \beta^{-\frac{\beta+1}{\beta-1}} \]
and at $\UpperThreshold(\NumConsumers, \subfee, \TokExp)$ is  
\[\NumConsumers \left(1 - \frac{\price}{\UpperThreshold(\NumConsumers, \subfee, \TokExp)} \right) \subfee.  \]
Let's consider the ratio: 
\[\frac{\NumConsumers \left(1 - \frac{1}{\beta} \right) \price^{-\frac{1}{\beta-1}} \beta^{-\frac{\beta+1}{\beta-1}}}{ \NumConsumers \left(1 - \frac{\price}{\UpperThreshold(\NumConsumers, \subfee, \TokExp)} \right) \subfee } = \frac{\left(1 - \frac{1}{\beta} \right) \beta^{-\frac{\beta+1}{\beta-1}}}{\subfee } \frac{\price^{-\frac{1}{\beta-1}} }{\left(1 - \frac{\price}{\UpperThreshold(\NumConsumers, \subfee, \TokExp)} \right) }.\]
The derivative is:
\[\frac{\left(1 - \frac{1}{\beta} \right) \beta^{-\frac{\beta+1}{\beta-1}}}{\subfee }  \cdot \frac{\UpperThreshold(\NumConsumers, \subfee, \TokExp) \price^{-\frac{\TokExp}{\TokExp-1}} \left( \TokExp \cdot \price - \UpperThreshold(\NumConsumers, \subfee, \TokExp)\right)}{(\TokExp - 1)(\UpperThreshold(\NumConsumers, \subfee, \TokExp) - \price)^2} ,\]
which is positive in this regime, so the ratio is increasing in this regime. The ratio approaches $\infty$ as $\price \rightarrow \UpperThreshold(\NumConsumers, \subfee, \TokExp)$. As $\price \cdot \TokExp \rightarrow \UpperThreshold(\NumConsumers, \subfee, \TokExp)$, the price 
$\price \cdot \TokExp$ is of the form $\UpperThreshold(\NumConsumers, \subfee, \TokExp) + \epsilon$ which we already proved to be dominated by $\UpperThreshold(\NumConsumers, \subfee, \TokExp)$, meaning that the ratio is less than 1. This proves the desired statement.

\paragraph{Case 3: $\LowerThreshold(\NumConsumers, \subfee, \TokExp) < \price \cdot \TokExp < \UpperThreshold(\NumConsumers, \subfee, \TokExp)$ and $(\price \cdot \TokExp) \cdot g(w^*(\price \cdot \TokExp)) < \alpha$.} We know that the supplier will either set the price to be $\LowerThreshold(\NumConsumers, \subfee, \TokExp)$ or $\UpperThreshold(\NumConsumers, \subfee, \TokExp)$. The realized profit at $\LowerThreshold(\NumConsumers, \subfee, \TokExp)$ is upper bounded by:
\[\NumConsumers (\TokExp \cdot \price - \price) \cdot g(w^*(\TokExp \cdot \price))  = \NumConsumers \left(1 - \frac{\price}{\TokExp \cdot \price} \right) \cdot (\TokExp \cdot \price) g(w^*(\TokExp \cdot \price)) < \NumConsumers \left(1 - \frac{\price}{\UpperThreshold(\NumConsumers, \subfee, \TokExp)} \right) \cdot \subfee,  \]
which is the profit at $\UpperThreshold(\NumConsumers, \subfee, \TokExp)$. This means that $\price_1 = \UpperThreshold(\NumConsumers, \subfee, \TokExp)$. 

\paragraph{Case 4: $\LowerThreshold(\NumConsumers, \subfee, \TokExp) < \price \cdot \TokExp < \UpperThreshold(\NumConsumers, \subfee, \TokExp)$ and $(\price \cdot \TokExp) \cdot g(w^*(\price \cdot \TokExp)) > \alpha$.} We know that the supplier will either set the price to be $\LowerThreshold(\NumConsumers, \subfee, \TokExp)$ or $\UpperThreshold(\NumConsumers, \subfee, \TokExp)$. We show that there is a threshold value $\LowerThresholdMon(\NumConsumers, \subfee, \TokExp)$ in this range such that the supplier sets the price to be  $\LowerThreshold(\NumConsumers, \subfee, \TokExp)$ if $\price \le \LowerThresholdMon(\NumConsumers, \subfee, \TokExp)$ and sets the price to be  $\UpperThreshold(\NumConsumers, \subfee, \TokExp)$ otherwise. Note that the realized profit at $\LowerThreshold(\NumConsumers, \subfee, \TokExp)$ is 
\[\NumConsumers \left(1 - \frac{\price}{\LowerThreshold(\NumConsumers, \subfee, \TokExp)} \right) \LowerThreshold(\NumConsumers, \subfee, \TokExp) g(w^*(\LowerThreshold(\NumConsumers, \subfee, \TokExp)))  \]
and at $\UpperThreshold(\NumConsumers, \subfee, \TokExp)$ is  
\[\NumConsumers \left(1 - \frac{\price}{\UpperThreshold(\NumConsumers, \subfee, \TokExp)} \right) \subfee.  \]
It is easy to see that the ratio is decreasing in $\price$. At $\price \cdot \TokExp = \LowerThreshold(\NumConsumers, \subfee, \TokExp)$, we see that the ratio is at least:
\[ (1 - \frac{1}{\beta}) \NumConsumers \cdot \LowerThreshold(\NumConsumers, \subfee, \TokExp) g(w^*(\LowerThreshold(\NumConsumers, \subfee, \TokExp))) > (1 - \frac{1}{\beta}) \NumConsumers \subfee > 1, \]
using our assumption that $\NumConsumers > \frac{\beta}{\beta-1}$. This proves the desired statement. 

\paragraph{Case 5: $\ic \cdot \TokExp < \LowerThreshold(\NumConsumers, \subfee, \TokExp)$.} We know that the supplier will either set the price to be $\price \cdot \TokExp$ or $\UpperThreshold(\NumConsumers, \subfee, \TokExp)$.
The above analysis for Case 4 shows that:
\[\NumConsumers \left(1 - \frac{\price}{\UpperThreshold(\NumConsumers, \subfee, \TokExp)} \right) \subfee < \NumConsumers \left(1 - \frac{\price}{\LowerThreshold(\NumConsumers, \subfee, \TokExp)} \right) \LowerThreshold(\NumConsumers, \subfee, \TokExp) g(w^*(\LowerThreshold(\NumConsumers, \subfee, \TokExp))).  \]
We also know that:
\[ \NumConsumers \left(1 - \frac{\price}{\LowerThreshold(\NumConsumers, \subfee, \TokExp)} \right) \LowerThreshold(\NumConsumers, \subfee, \TokExp) g(w^*(\LowerThreshold(\NumConsumers, \subfee, \TokExp))) < \NumConsumers (\price \cdot \TokExp - \price) g(w^*(\price \cdot \TokExp)).  \]
This proves that $\price_1 = \price \cdot \TokExp$ as desired. 
\end{proof}

We prove Theorem \ref{thm:extensionmonopolist} from Lemma \ref{lemma:equilibriumcharacterizationmonopolist}. 

\begin{proof}
By Lemma \ref{lemma:equilibriumcharacterizationmonopolist}, we know that the supplier's price satisfies:
\[ 
\price_1 =
\begin{cases}
\TokExp \cdot \ic & \text{ if } \ic < \TokExp^{-1} \cdot \LowerThreshold(\NumConsumers, \subfee, \TokExp)  \\
\LowerThreshold(\NumConsumers, \subfee, \TokExp) & \text{ if } \ic \ge \TokExp^{-1} \cdot \LowerThreshold(\NumConsumers, \subfee, \TokExp) \text{ and }
\ic \le \LowerThresholdMon(\NumConsumers, \subfee, \TokExp),   \\
\UpperThreshold(\NumConsumers, \subfee, \TokExp) & \text{ if } \ic \in (\LowerThresholdMon(\NumConsumers, \subfee, \TokExp), \UpperThresholdMon(\NumConsumers, \subfee, \TokExp)]\\
\TokExp \cdot \ic & \text{ if } \ic > \UpperThresholdMon(\NumConsumers, \subfee, \TokExp).
\end{cases}
\] 
This, coupled with Lemma \ref{lemma:middlemansubgamemonopolist}, gives us:
\[ 
\sum_{j=1}^{\NumConsumers} \mathbb{E}[1[\action{j} = \Middleman]] =
\begin{cases}
0 & \text{ if } \min(\ic + \costadditional, \costmanual) \le \LowerThresholdMon(\NumConsumers, \subfee, \TokExp)  \\
\NumConsumers & \text{ if } \min(\ic + \costadditional, \costmanual) \in (\LowerThresholdMon(\NumConsumers, \subfee, \TokExp), \UpperThresholdMon(\NumConsumers, \subfee, \TokExp)]\\
0 & \text{ if }  \min(\ic + \costadditional, \costmanual) > \UpperThresholdMon(\NumConsumers, \subfee, \TokExp).
\end{cases}
\]
as desired. 
\end{proof}

\subsection{Proof of Theorem  \ref{thm:extensionmarginal}}

We prove Theorem \ref{thm:extensionmarginal}. First, we prove the following analogue of Lemma \ref{lemma:middlemansubgame}.
\begin{lemma}
\label{lemma:middlemansubgamemarginal}
Consider the setup of Section \ref{subsec:marginalcosts} and Theorem \ref{thm:extensionmarginal}. Suppose that suppliers choose prices $\price_1, \ldots, \price_P$ and consider the subgame between the intermediary and consumers (Stages 2-3). Let $\price = \min(\costadditional + \min_{i \in [\NumProviders]} \price_i, \costmanual)$, and consider the condition 
\begin{equation}
\label{eq:conditionmarginal}
 \price (1 + \marg \NumConsumers) \cdot \TokFn\left(\subfee + \max_{w \ge 0} \left(\w -\price  (1+\marg) \TokFn(w)\right) \right) > \subfee \NumConsumers.
\end{equation}
Under the tiebreaking assumptions discussed in Section \ref{subsec:equilibriumexistence}, there exists a unique pure strategy where: 
\begin{itemize}
    \item If \eqref{eq:conditionmarginal} holds, then $\wmiddleman = 0$. Moreover,  for all $j \in [\NumConsumers]$, it holds that $\action{j} = \Direct$, 
    \[\wuser{j} = \argmax_{w \ge 0} \left(\w -  \price(1 + \marg) \TokFn(w)\right).\] Moreover, if $\costmanual < \costadditional + \min_{i \in [\NumProviders]} \price_i$, then the consumer chooses $\providerchoice{j} = 0$. Otherwise, the consumer chooses $\providerchoice{j} = \text{argmin}_{i \in [\NumProviders]} \price_i$ (tie-breaking in favor of suppliers with a lower index). 
    \item If \eqref{eq:condition} does not hold, then 
    \[\wmiddleman = \wuser{j} = \subfee + \max_{w \ge 0} \left(\w - \price (1+\marg) \TokFn(w)\right).\] Moreover, if $\costmanual < \costadditional + \min_{i \in [\NumProviders]} \price_i$, then the intermediary chooses $\providerchoicemiddleman = 0$; otherwise, the intermediary chooses $\providerchoicemiddleman = \text{argmin}_{i \in [\NumProviders]} \price_i$ (tie-breaking in favor of suppliers with a lower index). Finally, it holds that $\action{j} = \Middleman$ and $\wuser{j} = \wmiddleman$ for all $j \in [\NumConsumers]$. 
\end{itemize}
\end{lemma}
\begin{proof}
Like in the proof of Lemma \ref{lemma:middlemansubgame},
recall that when consumers or the intermediary produce content, they choose the option that minimizes their production costs. If 
$\costadditional + \min_{i \in [\NumProviders]} \price_i < \costmanual$, they leverage the technology of the supplier who offers the lowest price, and otherwise, they produce content without using the technology. This means that they face production costs $\price = \min(\costadditional + \min_{i \in [\NumProviders]} \price_i, \costmanual)$. 

The main difference from Lemma \ref{lemma:middlemansubgame} is that the consumers and the intermediary face marginal costs. Taking into account these marginal costs, when consumer $j$ chooses $\action{j} = \Direct$, then they maximize their utility and thus produce content $w^*(\price) = \argmax(w - \price (1+\marg) g(w))$ and achieve utility $\max(w -  (1+\marg) \price g(w))$. Since the consumer pays the intermediary a fee of $\subfee$, the intermediary must produce content satisfying $w' \ge \subfee + \max_{w \ge 0}(w - \price  (1+\marg)  g(w))$ to incentivize the consumer to choose $\action{j} = \Middleman$. Producing content $w' \ge \subfee + \max_{w \ge 0}(w - \price  (1+\marg)  g(w))$ would incentivize all of the consumers to choose the intermediary, so the intermediary would earn utility
\[ \subfee \cdot \NumConsumers - \price \cdot (1+\marg \NumConsumers) \cdot g(w').\]
This also means that the intermediary prefers producing content $\subfee + \max_{w \ge 0}(w - \price (1+\marg)  g(w))$ over any $w' > \subfee + \max_{w \ge 0}(w - \price  (1+\marg)  g(w))$ in order to minimize costs. The intermediary prefers producing this content over producing content $w = 0$ which would not attract any consumers if and only if:
\[ \subfee \cdot \NumConsumers - \price \cdot(1+\marg \NumConsumers) \cdot g(\subfee + \max_{w \ge 0}(w - \price (1+\marg)  g(w))) \ge 0.\]
This, coupled with the tiebreaking rules, proves the desired statement.
\end{proof}

Using this lemma, we can characterize the pure strategy equilibria in this extended model. 
\begin{lemma}
\label{lemma:equilibriumcharacterizationmarginal}
Consider the setup of Section \ref{subsec:marginalcosts} and Theorem \ref{thm:extensionmarginal}. 
Under the tiebreaking assumptions in Section \ref{subsec:equilibriumexistence}, there exists a pure strategy equilibrium which takes the following form: all suppliers choose the price $\price_i = \ic$ for $i \in [\NumProviders]$, and the intermediary and consumers choose actions according to the subgame equilibrium constructed in Lemma \ref{lemma:middlemansubgamemarginal}. The actions of the intermediary and consumers are the same at every pure strategy equilibrium; moreover, the production cost $\price = \min(\costadditional + \min_{i \in [\NumProviders]} \price_i, \costmanual)$ is the same at every pure strategy equilibrium.  
\end{lemma}
\begin{proof}
The proof of the equilibrium construction in the first sentence is analogous to the proof of Lemma \ref{lemma:equilibriumcharacterizationmarginal}. The proof of the second sentence is analogous to the proof of Theorem \ref{thm:equilibriumuniqueness}.
\end{proof}

Using these characterization results, we prove Theorem \ref{thm:extensionmarginal}.
\begin{proof}[Proof of Theorem \ref{thm:extensionmarginal}]  By Lemma \ref{lemma:equilibriumcharacterizationmarginal}, we know that the intermediary survives if and only if  
\[ \price (1 + \marg \NumConsumers) \cdot \TokFn\left(\subfee + \max_{w \ge 0} \left(\w -\price (1 + \marg) \TokFn(w)\right) \right) \le \subfee \NumConsumers.\]
Let us change variables and let $\price' = \price (1+\marg)$ and let $C' = \frac{C(1+\marg)}{1+\marg \NumConsumers}$. Then we can write the condition as:
\[ \price' \cdot \TokFn\left(\subfee + \max_{w \ge 0} \left(\w -\price' \TokFn(w)\right) \right) \le \subfee \NumConsumers'.\]
This means that 
\[ 
\sum_{j=1}^{\NumConsumers} \mathbb{E}[1[\action{j} = \Middleman]] =
\begin{cases}
0 & \text{ if } \min(\ic + \costadditional, \costmanual) < \LowerThresholdMarg(\NumConsumers, \subfee, \TokExp, \marg) \\
\NumConsumers & \text{ if } \min(\ic + \costadditional, \costmanual) \in [\LowerThresholdMarg(\NumConsumers, \subfee, \TokExp, \marg), \UpperThresholdMarg(\NumConsumers, \subfee, \TokExp, \marg)]\\
0 & \text{ if } \min(\ic + \costadditional, \costmanual) > \UpperThresholdMarg(\NumConsumers, \subfee, \TokExp, \marg) \\
\end{cases}
\]
where $\LowerThresholdMarg(\NumConsumers, \subfee, \TokExp, \marg) = (1+ \marg)^{-1} \cdot \LowerThreshold(\NumConsumers', \subfee, g)$ and $\UpperThresholdMarg(\NumConsumers, \subfee, \TokExp, \marg) = (1+ \marg)^{-1} \cdot \UpperThreshold(\NumConsumers', \subfee, \TokExp)$ as desired. 
\end{proof}
\subsection{Proof of Theorem  \ref{thm:extensionfees}}

We prove Theorem \ref{thm:extensionfees}. 
First, we prove the following analogue of Lemma \ref{lemma:middlemansubgame}.
\begin{lemma}
\label{lemma:middlemansubgamefees}
Consider the setup of Section \ref{subsec:transfers} and Theorem \ref{thm:extensionfees}. Suppose that suppliers choose prices $\price_1, \ldots, \price_P$ and consider the subgame between the intermediary and consumers (Stages 2-3). Let $\price = \min(\costadditional + \min_{i \in [\NumProviders]} \price_i, \costmanual)$, and consider the condition 
\begin{equation}
\label{eq:conditionfees}
 \subfee^{\frac{1}{\beta-1}} (1-\alpha) < \NumConsumers^{-\frac{1}{\beta-1}} \left( \beta^{-\frac{1}{\beta-1}} - \beta^{-\frac{\beta}{\beta-1}} \right)
\end{equation}
Under the tiebreaking assumptions discussed in Section \ref{subsec:equilibriumexistence}, there exists a unique pure strategy where: 
\begin{itemize}
    \item If \eqref{eq:conditionfees} holds, then $\wmiddleman = 0$. Moreover,  for all $j \in [\NumConsumers]$, it holds that $\action{j} = \Direct$, 
    \[\wuser{j} = \argmax_{w \ge 0} \left(\w -  \price\TokFn(w)\right).\] Moreover, if $\costmanual < \costadditional + \min_{i \in [\NumProviders]} \price_i$, then the consumer chooses $\providerchoice{j} = 0$. Otherwise, the consumer chooses $\providerchoice{j} = \text{argmin}_{i \in [\NumProviders]} \price_i$ (tie-breaking in favor of suppliers with a lower index). 
    \item If \eqref{eq:conditionfees} does not hold, then 
    \[\wmiddleman = \wuser{j} = (1 - \subfee)^{-1} \max_{w \ge 0} \left(\w - \price \TokFn(w)\right).\] Moreover, if $\costmanual < \costadditional + \min_{i \in [\NumProviders]} \price_i$, then the intermediary chooses $\providerchoicemiddleman = 0$; otherwise, the intermediary chooses $\providerchoicemiddleman = \text{argmin}_{i \in [\NumProviders]} \price_i$ (tie-breaking in favor of suppliers with a lower index). Finally, it holds that $\action{j} = \Middleman$ and $\wuser{j} = \wmiddleman$ for all $j \in [\NumConsumers]$. 
\end{itemize}
\end{lemma}
\begin{proof}
Like in the proof of Lemma \ref{lemma:middlemansubgame},
recall that when consumers or the intermediary produce content, they choose the option that minimizes their production costs. If 
$\costadditional + \min_{i \in [\NumProviders]} \price_i < \costmanual$, they leverage the technology of the supplier who offers the lowest price, and otherwise, they produce content without using the technology. This means that they face production costs $\price = \min(\costadditional + \min_{i \in [\NumProviders]} \price_i, \costmanual)$. 

The main difference from Lemma \ref{lemma:middlemansubgame} is in the fee structure. Like before, when consumer $j$ chooses $\action{j} = \Direct$, then they maximize their utility and thus produce content $w^*(\price) = \argmax(w - \price  g(w))$ and achieve utility $\max(w -  \price g(w))$. Since the consumer pays the intermediary a fee of $\subfee \cdot \w$, the intermediary must produce content satisfying $w' \ge \subfee \w' + \max_{w \ge 0}(w - \price  g(w))$ to incentivize the consumer to choose $\action{j} = \Middleman$. Producing content $w' \ge (1 - \subfee)^{-1} \cdot \max_{w \ge 0}(w - \price   g(w))$ would incentivize all of the consumers to choose the intermediary. We can use the structure of $g(w)$ to simplify this condition to:
\begin{equation}
\label{eq:lowerbound}
 w' \ge (1 - \subfee)^{-1} \cdot \max_{w \ge 0}(w - \price   g(w)) = (1 - \subfee)^{-1} \cdot \price^{-\frac{1}{\beta-1}} \left( \beta^{-\frac{1}{\beta-1}} - \beta^{-\frac{\beta}{\beta-1}} \right)   
\end{equation}
If the intermediary produces content $w'$, they would earn utility
\[ \subfee \cdot w' \cdot \NumConsumers - \price \cdot g(w') = \subfee \cdot w' \cdot \NumConsumers - \price \cdot (w')^{\TokExp}.\]
The intermediary prefers producing this content over producing content $w = 0$ which would not attract any consumers if and only if:
\[ \subfee \cdot w' \cdot \NumConsumers - \price \cdot (w')^{\TokExp} \ge 0.\]
We can solve this to obtain:
\begin{equation}
\label{eq:upperbound}
w' \le \price^{-\frac{1}{\beta-1}} \left(\subfee \cdot \NumConsumers \right)^{\frac{1}{\beta-1}}
\end{equation}
Because of the structure of the tiebreaking rules, the intermediary survives in the market if and only if there exist  $w'$ satisfying both \eqref{eq:upperbound} and \eqref{eq:lowerbound}. This happens if and only if: 
\[\price^{-\frac{1}{\beta-1}} \left(\subfee \cdot \NumConsumers \right)^{\frac{1}{\beta-1}} \ge  (1 - \subfee)^{-1} \cdot \price^{-\frac{1}{\beta-1}} \left( \beta^{-\frac{1}{\beta-1}} - \beta^{-\frac{\beta}{\beta-1}} \right).\]
This simplifies to:
\[ \subfee^{\frac{1}{\beta-1}} (1-\alpha) \ge \NumConsumers^{-\frac{1}{\beta-1}} \left( \beta^{-\frac{1}{\beta-1}} - \beta^{-\frac{\beta}{\beta-1}} \right).\]
\end{proof}

Using this lemma, we can characterize the pure strategy equilibria in this extended model. 
\begin{lemma}
\label{lemma:equilibriumcharacterizationfees}
Consider the setup of Section \ref{subsec:transfers} and Theorem \ref{thm:extensionfees}. 
Under the tiebreaking assumptions in Section \ref{subsec:equilibriumexistence}, there exists a pure strategy equilibrium which takes the following form: all suppliers choose the price $\price_i = \ic$ for $i \in [\NumProviders]$, and the intermediary and consumers choose actions according to the subgame equilibrium constructed in Lemma \ref{lemma:middlemansubgamefees}. The actions of the intermediary and consumers are the same at every pure strategy equilibrium; moreover, the production cost $\price = \min(\costadditional + \min_{i \in [\NumProviders]} \price_i, \costmanual)$ is the same at every pure strategy equilibrium.  
\end{lemma}
\begin{proof}
The proof of the equilibrium construction in the first sentence is analogous to the proof of Lemma \ref{lemma:equilibriumcharacterizationmarginal}. The proof of the second sentence is analogous to the proof of Theorem \ref{thm:equilibriumuniqueness}.
\end{proof}

Using these characterization results, we prove Theorem \ref{thm:extensionfees}.
\begin{proof}[Proof of Theorem \ref{thm:extensionfees}]  By Lemma \ref{lemma:equilibriumcharacterizationfees}, we know that the intermediary survives if and only if  
\[  \subfee^{\frac{1}{\beta-1}} (1-\alpha) \ge \NumConsumers^{-\frac{1}{\beta-1}} \left( \beta^{-\frac{1}{\beta-1}} - \beta^{-\frac{\beta}{\beta-1}} \right).\]
This condition is independent of $\price$ as desired. Moreover, the condition becomes weaker as $\NumConsumers$ gets larger. 
\end{proof}

\end{document}